\documentclass[11pt]{article}

\usepackage[T1]{fontenc}
\usepackage{lmodern}

\usepackage{array}
\usepackage{amsmath, amssymb}
\usepackage{yhmath}
\usepackage{amsthm}
\usepackage{proof}
\usepackage{enumitem}
\usepackage{algpseudocodex}
\usepackage{algorithm,setspace}

\usepackage{mathtools}
\usepackage{fullpage}
\usepackage{scalerel}
\usepackage{hhline}
\usepackage{bbm}
\usepackage[normalem]{ulem}
\usepackage{threeparttable}
\usepackage{booktabs}
\usepackage[pdfstartview=FitH,pdfpagemode=UseNone,colorlinks=true,citecolor=blue,linkcolor=blue]{hyperref}
\usepackage{centernot}
\usepackage[
	lambda,
	operators,
	advantage,
	sets,
	adversary,
	landau,
	probability,
	notions,
	logic,
	ff,
	mm,
	primitives,
	events,
	complexity,
	asymptotics,
	keys]{cryptocode}
\usepackage[T1]{fontenc}
\usepackage{graphicx}
\usepackage{framed}
\usepackage{mdframed}
\usepackage[noabbrev]{cleveref}
\usepackage{verbatim}
\floatname{algorithm}{Protocol}
\crefname{algorithm}{protcol}{protocol}
\Crefname{algorithm}{Protocol}{Protocol}

\usepackage[utf8]{inputenc}
\usepackage{csquotes}
\usepackage{microtype}
\usepackage{multirow}
\usetikzlibrary{decorations.pathreplacing,matrix,shapes,arrows,positioning,chains,calc}

\tikzset{
    ncbar angle/.initial=90,
    ncbar/.style={
        to path=(\tikztostart)
        -- ($(\tikztostart)!#1!\pgfkeysvalueof{/tikz/ncbar angle}:(\tikztotarget)$)
        -- ($(\tikztotarget)!($(\tikztostart)!#1!\pgfkeysvalueof{/tikz/ncbar angle}:(\tikztotarget)$)!\pgfkeysvalueof{/tikz/ncbar angle}:(\tikztostart)$)
        -- (\tikztotarget)
    },
    ncbar/.default=0.5cm,
}

\tikzset{square left brace/.style={ncbar=0.5cm}}
\tikzset{square right brace/.style={ncbar=-0.5cm}}

\tikzset{round left brace/.style={ncbar=0.5cm,out=-115,in=115}}
\tikzset{round right brace/.style={ncbar=0.5cm,out=-65,in=65}}

\usepackage{subcaption}
\usepackage{wrapfig}
\usepackage{xcolor}
\usepackage{xspace}

\usepackage[normalem]{ulem}
\usepackage{bm}

\newif\ifFOCS
\ifFOCS
\Crefname{figure}{Fig.}{Fig.}
\else
\usepackage[colorinlistoftodos,textsize=small]{todonotes}
\fi

\theoremstyle{remark}
\newtheorem{remark}{Remark}

\theoremstyle{definition}
\newtheorem{definition}{Definition}

\theoremstyle{plain}
\newtheorem{claim}{Claim}

\newtheorem{lemma}{Lemma}
\newtheorem{theorem}{Theorem}
\newtheorem{proposition}{Proposition}
\newtheorem{corollary}{Corollary}

\newcommand{\newuser}[3]{%
  \expandafter\newcommand\csname todo#1\endcsname[2][]{%
    \quitvmode
    \texorpdfstring{\todo[inline,color=#3,##1]{\textbf{#2:} ##2}\xspace}{(TODO: #2: ##2)}%
  }
}

\newcolumntype{?}{!{\vrule width 1pt}}

\DeclarePairedDelimiter{\lrag}{\lag}{\rag}

\newcommand{\err}{\mathsf{err}}

\newcommand{\polylog}[1]{\mathrm{polylog}{#1}\xspace}
\renewcommand{\poly}[1]{{\mathrm{poly}{#1}}}

\newcommand{\tO}{\tilde{O}}

\mathchardef\mhyphen="2D

\let\lag\langle
\let\rag\rangle

\crefname{conjecture}{conjecture}{Conjecture}
\crefname{corl}{corollary}{Corollary}
\crefname{question}{question}{Question}

\makeatletter
\DeclareOldFontCommand{\bf}{\normalfont\bfseries}{\mathbf}
\DeclareOldFontCommand{\it}{\normalfont\itshape}{\mathit}
\DeclareOldFontCommand{\rm}{\normalfont\rmfamily}{\mathrm}
\DeclareOldFontCommand{\sc}{\normalfont\scshape}{\@nomath\sc}
\DeclareOldFontCommand{\sf}{\normalfont\sffamily}{\mathsf}
\DeclareOldFontCommand{\sl}{\normalfont\slshape}{\@nomath\sl}
\DeclareOldFontCommand{\tt}{\normalfont\ttfamily}{\mathtt}
\makeatother
\def\gkeywords{}
\def\gabstract{}

\newcommand\makealltitles{
\ifnum\style=\sigalternate
    {
        \begin{abstract}{\gabstract}\end{abstract} \maketitle%
    }
\else
  {\maketitle%
  \ifnum\style=\lncs
    {\begin{abstract}{\gabstract \keywords{\gkeywords}}\end{abstract}}
  \else
    {\begin{abstract}{\gabstract}\end{abstract}
      \ifnum\style=\ieeetr
      \begin{IEEEkeywords}
        {\gkeywords}
      \end{IEEEkeywords}
      \fi
    }
  \fi
  }
\fi
}

\mathchardef\mhyphen="2D

\newcommand{\ignore}[1]{}

\newcommand{\prot}[1]{\ensuremath{(\cP_{#1}, \cV_{#1})}\xspace}
\newcommand{\protU}[1]{\ensuremath{(\cP{#1}, \cV{#1})}\xspace}

\newcommand{\cL}{\mathcal{L}}
\newcommand{\cM}{\mathcal{M}}
\newcommand{\cR}{\mathcal{R}}
\newcommand{\cS}{\mathcal{S}}

\newcommand{\hyb}{\ensuremath{\mathsf{Hyb}}\xspace}

\let\oldprocedure\procedure
\renewcommand\procedure[3][]{\hspace{0pt}\oldprocedure[#1]{#2}{#3}}

\NewDocumentCommand\inputanon{m}{
\ifnum\paperversion=\anonymous \input{#1} \fi
}

\NewDocumentCommand\inputcam{m}{
\ifnum\paperversion=\cameraready \input{#1} \fi
}

\NewDocumentCommand\inputfull{m}{
\ifnum\paperversion=\fullversion \input{#1} \fi
}

\newcommand{\cB}{\mathcal{B}}

\newcommand{\cU}{{\ensuremath{\mathcal{U}}\xspace}}

\newcommand{\cP}{{\ensuremath{\mathsf{P}}\xspace}}
\newcommand{\cV}{{\ensuremath{\mathsf{V}}\xspace}}

\newcommand{\EE}{\mathbb{E}}
\newcommand{\KK}{\mathbb{K}}

\newcommand{\GF}{\mathbb{GF}}

\newcommand{\cksum}{\ensuremath{\mathsf{cksum}}\xspace}

\newcommand{\Deltac}{\ensuremath{\Delta_c}\xspace}
\newcommand{\pval}{\textsf{PVAL}\xspace}
\newcommand{\pvalF}{\textsf{PVAL}_\FF\xspace}
\newcommand{\PVAL}{\textsf{PVAL}\xspace}

\newcommand{\eps}{\ensuremath{\epsilon}\xspace}

\newcommand{\GKR}{\ensuremath{\textsf{GKR}}\xspace}
\newcommand{\tc}{\ensuremath{\textsf{tc}}\xspace}
\renewcommand{\AND}{\ensuremath{\textsf{AND}}\xspace}
\newcommand{\ADD}{\ensuremath{\textsf{ADD}}\xspace}
\newcommand{\MULT}{\ensuremath{\textsf{MULT}}\xspace}

\renewcommand{\RR}{\ensuremath{\textsf{RR}}\xspace}
\newcommand{\DcRR}{\ensuremath{\Delta_c\textsf{RR}}\xspace}

\newcommand{\RRrow}{{\ensuremath{\Delta_c\RR}\xspace}}

\newcommand{\pvalU}{{\ensuremath{\bm j}\xspace}}
\newcommand{\pvalT}{{\ensuremath{T}\xspace}}

\newcommand{\pvalv}{{\ensuremath{v}\xspace}}
\newcommand{\cksumU}{{\ensuremath{\bm U}\xspace}}
\newcommand{\cksumT}{{\ensuremath{T_\cksum}\xspace}}

\newcommand{\nround}{{\ensuremath{{\rho}}\xspace}}
\newcommand{\K}{{\ensuremath{{M}}\xspace}}
\newcommand{\rowball}{{\ensuremath{{\cB_{d, \FF}}}\xspace}}
\newcommand{\batchK}{{\ensuremath{{k}}\xspace}}

\newcommand{\ncol}
{{\ensuremath{{L}}\xspace}}
\newcommand{\lncol}
{{\log{\ensuremath{{L}}\xspace}}}
\newcommand{\RRS}{{\ensuremath{{\mathcal{Q}}}\xspace}}
\newcommand{\lgK}{{\ensuremath{{m}}\xspace}}
\newcommand{\cksumv}{{\ensuremath{\bm S}\xspace}}
\newcommand{\GKRcirc}{{\ensuremath{{\Cmid}}}\xspace}

\renewcommand{\secpar}{{\ensuremath{{\sigma}}\xspace}}

\newcommand{\bUP}{\textsf{Batch-UP}\xspace}

\newcommand{\NP}{\textsf{NP}\xspace}
\newcommand{\UP}{\textsf{UP}\xspace}
\newcommand{\bUIP}{\textsf{Batch-UIP}\xspace}

\newcommand{\UIP}{\textsf{UIP}\xspace}
\newcommand{\IP}{\textsf{IP}\xspace}

\newcommand{\IOP}{\textsf{IOP}\xspace}
\newcommand{\IPP}{\textsf{IPP}\xspace}
\newcommand{\UIPP}{\textsf{UIPP}\xspace}

\newcommand{\TS}{{\ensuremath{\mathsf{DTISP}(T, S)\xspace}}}

\newcommand{\LDE}{\textsf{LDE}\xspace}

\newcommand{\depth}{\ensuremath{{\mathsf{depth}}\xspace}}
\newcommand{\size}{\ensuremath{{\mathsf{size}}\xspace}}
\newcommand{\Vtime}{\ensuremath{{\mathsf{Vtime}}\xspace}}
\newcommand{\Ptime}{\ensuremath{{\mathsf{Ptime}}\xspace}}

\newcommand{\old}{\ensuremath{{(j - 1)}\xspace}}
\newcommand{\transc}{\ensuremath{{\bm a}\xspace}}
\newcommand{\answer}{\ensuremath{{\bm a}\xspace}}
\newcommand{\transcMat}{\ensuremath{{\bm a}\xspace}}
\newcommand{\answerMat}{\ensuremath{{\bm a}\xspace}}
\newcommand{\new}{\ensuremath{{(j)}\xspace}}

\newcommand{\I}{\ensuremath{{\mathcal{I}}\xspace}}

\newcommand{\C}{\ensuremath{{\Phi\xspace}}}

\newcommand{\cSold}{\ensuremath{{\cS^{\old}}\xspace}}
\newcommand{\Cold}{\ensuremath{{\C^{\old}}\xspace}}

\newcommand{\paramold}{\ensuremath{{\lrag \cSold, \lrag \Cold}\xspace}}
\newcommand{\parammid}{\ensuremath{{\lrag \cSmid, \lrag \Cmid}\xspace}}

\newcommand{\Cnew}{\ensuremath{{\C^{\new}}\xspace}}
\newcommand{\cSnew}{\ensuremath{{\cS^{\new}}\xspace}}

\newcommand{\cSmid}{\ensuremath{{\cS_{\Distance}}\xspace}}
\newcommand{\Cmid}{\ensuremath{{\C_{\Distance}}\xspace}}
\newcommand{\paramnew}{\ensuremath{{\lrag \cSnew, \lrag\Cnew}\xspace}}
\newcommand{\paramC}[1]{\ensuremath{{\lrag{\cS^{(#1)}}, \lrag{\C^{(#1)}}}\xspace}}
\newcommand{\param}[1]{\ensuremath{{\lrag{\cS{#1}}, \lrag{\C#1}}\xspace}}

\newcommand{\protphaseone}{Generating $\Delta_c$-Distance via Checksums.}

\newcommand{\protphasetwo}{Instance Reduction for $\Delta_c$-Distance.}

\newcommand{\iterinput}{{\textbf{Input:}  $\bm x = (x_1,\ldots,x_{\batchK})$ are the UIP statements. $\param{}$ are the descriptions for the claim $(\cS, \C) \in \cL'$}}

\newcommand{\iteroutput}{{$\param{'}$, new descriptions for the claim $(\cS', \C') \in \cL'$}}

\newcommand{\itermidone}{{$\bm x,\lrag{\cSmid}$, $\lrag{\Cmid}$}}

\newcommand{\itermidthree}{{descriptions $\lrag \RRS$ of a set of rows $\RRS \subsetneq [\K]$,
and $\lrag{\C'}$ of a predicate $\C'$}}

\newcommand{\distcircone}{{$\C(\transcMat[\cS, :]) = 1$ is true.
    Note that the submatrix $\transcMat[\cS,:]$ is indexed by the subset $\cS$ on $\transcMat$. 
    (This is valid because $\cS \subset \cS^\hyb$.)}}

\newcommand{\distcirctwo}{{For all $(i, \bm q) \in \cS^{\hyb}$,
    $\cV(x_i, \bm q, \transc'[(i, \bm q), :]) = 1$,
    and that for every $j \in [\ell]$,
    the $(j\cdot a)$-th prefix of $\transcMat[(i, \bm q^\hyb_j), :]$ and $\transcMat[(i, \bm q), :]$ are identical.}}

\newcommand{\distcircthree}{{$\cksum_{\bm \cksumU}(\transcMat) = (\cksumv_1,\ldots,\cksumv_\ell)$.}}

\newcommand{\distcircfour}{{$\transcMat \in \bin^{\K \times \ncol}$ (i.e. all elements strictly lies in $\GF(2)$.)}}

\newcommand{\printchecksumchecks}{{\begin{enumerate}
    \item \distcircone
    \item \distcirctwo 
    \item \distcircthree
    \item \distcircfour
\end{enumerate}
}}

\newcommand{\Fbits}{{\ensuremath{\polylog{\abs{\FF}}}}}
\newcommand{\ippinput}{{\ensuremath{\transcMat^{\cSmid}}}}
\newcommand{\Flog}{{\ensuremath{\log{\abs{\FF}}}}}

\newcommand{\FboundT}{{\ensuremath{32 \cdot 2^{\secpar}(\pvalT \log \K \log \ncol)^C}}}
\newcommand{\Fboundconcrete}{{\ensuremath{32 \cdot 2^{\secpar} \cdot (16 \log(\batchK) \cdot \secpar a \ell \lgK)^{C_0}}}}
\newcommand{\FboundconcreteT}{{\ensuremath{32 \cdot 2^{\secpar} (d a \ell \lgK)^C}}}
\newcommand{\Tbound}{{\ensuremath{8d\ncol\lgK}}}
\newcommand{\dbound}{{\ensuremath{\ceil{16 \cdot \log (\batchK \ell) \cdot \secpar} \cdot \ell}}}
\newcommand{\dboundm}{{\ensuremath{16\lgK \secpar}}}
\newcommand{\cbound}{{\ensuremath{8C_1C_3 + 2\log_n (\batchK S(n)) + 1}}}

\newcommand{\Distance}{{\textsf{{Dist}}\xspace}}
\newcommand{\Reduce}{{\textsf{{Reduce}}\xspace}}

\newcommand{\DCMainStmt}{{
Suppose $\secpar$, $d$, $\lgK$, $\lncol \in \NN$, $\K = 2^{\lgK}$,
and $\FF$ is a field.
Let $\tO$ be omitting $\polylog(\abs{\FF}, \K, \ncol)$ factors.
Suppose $\pvalT \ge \Tbound = \tO (d\ncol)$.
Let $\bm \pvalU = (\bm \pvalU_1, \ldots, \bm \pvalU_\pvalT) \in (\FF^{\lgK + \lncol})^\pvalT$ and $\bm \pvalv = (v_1, \ldots, v_\pvalT) \in \FF^\pvalT$.

There exists a constant $C$, 
and a public-coin unambiguous \IPP $\prot{\RRrow}$ for $\pval(\bm \pvalU, \bm \pvalv)$,
where both the prover and verifier gets the input $\pvalU, \bm \pvalv$,
while the prover additionally gets the input $\transcMat \in \FF^{\K \times \ncol}$,
such that if both parties additionally gets parameters $\secpar, d \in \NN$ and $\FF$,
that satisfy the following preconditions:
\begin{itemize}
    \item $d \ge \dboundm$,
    \item $\abs{\FF} \ge \FboundT$,
\end{itemize}
then the protocol outputs a succinct description $\lrag \RRS$ of 
\itermidthree,
\begin{itemize}[label=-]
    \item \textbf{Prescribed Completeness:} If $\cV_{\RRrow}$ interacts with $\cP_{\RRrow}$,
    then $\C'(\transcMat[\RRS, :]) = 1$ 
    iff $\transcMat \in \pval(\bm \pvalU, \bm \pvalv)$.
    \item $2^{-\secpar-2}$-\textbf{Unambiguity:} 
    Suppose $\Delta_c(\pval(\bm \pvalU,\bm 0)) \geq 4d$,
    then for any (unbounded) cheating prover strategy $\cP^*$ that deviates from $\cP_\RRrow$ first in round $r^* \in [\ell_{\RRrow}]$,
    with probability at least $1 - 2^{-\secpar - 2}$ over the verifier's remaining coins,
    either $\cV_\RRrow$ rejects,
    or $\C'(\transcMat[\RRS, :]) = 0$.
    \item \textbf{Reduced Query:} 
    The subset of rows $\RRS \subsetneq [\K]$ has size $\abs{\RRS} \le \ceil{8\secpar \cdot \frac{\K}{d}}$.
    \item Regardless of the prover's strategy, the distribution of $\RRS$ only depends on the verifier's random coins.
\end{itemize}
The complexity of the protocol is as follows.
\begin{itemize}
    \item $\ell_{\RRrow} = \tO(1)$.
    \item $a_{\RRrow} = b_{\RRrow} = \tO(\ncol + \poly(d))$.
    \item $\Ptime_{\RRrow} = \poly(\K \ncol, \pvalT \cdot \Flog)$.
    \item $\Vtime_{\RRrow} = \tO(\ncol + \poly(d))$.
\end{itemize}
The bit-lengths are $\abs{\lrag{\RRS}} = \tO(\poly(d))$ and $\abs{\lrag{\C'}} = \tO(\ncol + \poly(d))$.
Let $G_{\RRS}(i, \lrag \RRS)$ be the circuit that returns the $i$-th element in $\RRS$,
and $C'(\transcMat, \lrag {\C'})$ be the circuit that computes $\C'(\transcMat)$,
then they satisfy the following.
\begin{itemize}
    \item $\size(G_{\RRS}) = \tO(\poly(d))$.
    \item $\depth(G_{\RRS}) = \tO(1)$.
    \item $\size(C') = \tO(\abs{\RRS} \cdot \ncol)$.
    \item $\depth(C') = \tO(1)$.
\end{itemize}
The verifier's verdict circuit (which outputs 0 iff it rejects amidst the protocol) satisfies
\begin{itemize}
    \item $\size(\cV_{\RRrow}) = \tO(\ncol + \poly(d))$.
    \item $\depth(\cV_{\RRrow}) = \tO(1)$.
\end{itemize}
}}

\FOCSfalse    
\title{Efficiently Batching Unambiguous Interactive Proofs}
\usepackage{authblk}
\author{Bonnie Berger}
\author{Rohan Goyal}
\author{Matthew M. Hong}
\author{Yael Tauman Kalai}
\affil{MIT}
\date{}

\begin{document}
\maketitle

\begin{abstract}
    We show that if a language $\mathcal{L}$ admits a public-coin unambiguous interactive proof (UIP) with round complexity $\ell$, where $a$ bits are communicated per round, then the \emph{batch language} $\mathcal{L}^{\otimes k}$,
    i.e. the set of $k$-tuples of statements all belonging to $\cL$,
    has an unambiguous interactive proof with round complexity $\ell\cdot\mathsf{polylog}(k)$,
    per-round communication of $a\cdot \ell\cdot\mathsf{polylog}(k) + \poly(\ell)$ bits,
    assuming the verifier in the $\UIP$ has depth bounded by $\mathsf{polylog}(k)$.  
    Prior to this work, the best known batch $\UIP$ for $\mathcal{L}^{\otimes{k}}$ required communication complexity at least
    $(\mathsf{poly}(a)\cdot k^{\epsilon} + k) \cdot \ell^{1/\epsilon}$
    for any arbitrarily small constant $\epsilon>0$ (Reingold-Rothblum-Rothblum, STOC 2016). 
    
    As a corollary of our result, 
    we obtain a \emph{doubly efficient proof system},
    that is, 
    a proof system whose proving overhead is polynomial in the time of the underlying computation,
    for any language computable in polynomial space and in time at most $n^{O\left(\sqrt{\frac{\log n}{\log\log n}}\right)}$. 
    This expands the state of the art of doubly efficient proof systems: 
    prior to our work, such systems were known for languages computable in polynomial space and in time $n^{({\log n})^\delta}$ for a small $\delta>0$ significantly smaller than $1/2$ (Reingold-Rothblum-Rothblum, STOC 2016).
\end{abstract}

\thispagestyle{empty}
\newpage
\enlargethispage{1cm}
\tableofcontents
\addtocontents{toc}{\protect\thispagestyle{empty}}
\thispagestyle{empty}
\newpage
\pagenumbering{arabic}
\section{Introduction}
\label{sec:intro}

Verification is one of the most fundamental concepts in computer science and is the basis for the definition of the complexity class 
$\NP$. In the mid-eighties, a flurry of works expanded the notion of proofs beyond the classical notion of a mathematical proof to interactive proofs~\cite{GolMicRac89,BabaiM88}, multi-prover interactive proofs~\cite{STOC:BGKW88,FOCS:BabForLun90}, probabilistically checkable proofs (PCPs) \cite{STOC:BFLS91,FOCS:FGLSS91,FOCS:AroSaf92,FOCS:ALMSS92},
interactive PCPs~\cite{ICALP:KalRaz08},
and interactive oracle proofs~\cite{TCC:BenChiSpo16,STOC:ReiRotRot16}.
This line of work has proven instrumental in cryptography and complexity theory, leading to breakthroughs in hardness of approximation and to the emergence of fundamental concepts in cryptography, such as zero-knowledge proofs \cite{GolMicRac89}, which underlie many cryptographic primitives today.
It is also instrumental in constructing ``succinct proofs'' used in many blockchain applications.

\paragraph{Interactive proofs} 
Interactive proof systems were shown to be extremely powerful in the celebrated 
$\IP=\mathsf{PSPACE}$ theorem \cite{JACM:LFKN92,Shamir92}
Specifically, it was proven that the correctness  of any time-$T$ space-$S$ computation could be verified by a $\poly(S,n)$ time verifier, via an interactive proof.  
The main drawback of this result is that the time required by the prover to convince the verifier of the correctness is $2^{O(S\cdot \log S)}$, rendering it impractical for real-world applications.\footnote{
We note that in those works (and at that time) the primary focus was on the power of the proof systems themselves, where the verifier was assumed to run in polynomial time, while the prover was thought of as being all-powerful. In fact, in \cite{BabaiM88} the all-powerful prover was even named after the famous wizard Merlin.
}

\paragraph{Doubly efficient interactive proofs.}  
The work of \cite{JACM:GolKalRot15} initiated the study of \emph{doubly efficient} interactive proofs, in which the prover’s runtime is required to be at most 
$p(T)$, where 
$T$ is the time needed to carry out the underlying computation and $p$ is a polynomial independent of $T$,    
while the verifier runs in time significantly less than $T$. They showed that the correctness of any computation that can be performed by a (log-space uniform) circuit of depth $D$
and size $T$
can be proven via an interactive proof whose communication complexity is 
$D\cdot\polylog(T)$, the verifier runs in time 
$D\cdot\polylog(T)+ \tilde{O}(n)$, and the prover runs in time 
$\poly(T)$.\footnote{In particular, this implies an improved 
$\IP=\mathsf{PSPACE}$ theorem where the communication complexity is 
$\poly(S)$, the verifier runs in time
$\poly(S)+ 
\tilde{O}(n)$, and the prover runs in time
$2^{O(S)}$.} In a breakthrough result, Reingold, Rothblum, and Rothblum \cite{STOC:ReiRotRot16} constructed a doubly efficient (constant round) interactive proof for every language in time $T=n^{(\log n)^\delta}$ and polynomial space, for a sufficiently small constant $\delta>0$. 
Since this work which was posted nearly 10 years ago, no improvements have been made to this fundamental problem.
One of the building blocks in \cite{STOC:ReiRotRot16} is a doubly efficient \emph{batch unambiguous interactive proof}, which is the focus of our work.

\paragraph{Doubly efficient batch interactive proofs.} Another important line of work is the one that focuses  on doubly efficient \emph{batch interactive proofs} \cite{STOC:ReiRotRot16,RRR18,TCC:RotRot20}, where the goal is to prove many statements “at the price of one.” Concretely, suppose we want to verify statements 
$x_1,\ldots,x_k$, each of which belongs to a language that admits an interactive proof (i.e., a language in 
$\mathsf{PSPACE}$). 
Observing that the total space needed to run $k$ computations grows only by an (additive) $\polylog(k)$ factor, it follows, from $\IP = \mathsf{PSPACE}$, that all 
$k$ statements can also be proven via a single (batch) interactive proof with 
$\polylog(k)$ overhead in communication complexity. This naturally raises the question:

\begin{quote} 
Which languages in \IP can be batched \emph{doubly-efficiently},
with $\polylog(k)$ overhead in communication complexity,
while preserving the efficiency of the prover and verifier? \end{quote}

\noindent
The work of \cite{STOC:BKPRV24} proved that if a language $\cL \in \NP$ has a batch interactive proof with polylogarithmic overhead,
where the prover runs in polynomial time given the witnesses (we refer to such proof as a \emph{doubly efficient batch interactive proof}),
then $\cL$ also has a \emph{statistical witness indistinguishable} (WI) proof system. This result can be viewed as an indication that not all of $\NP$ has doubly efficient batch interactive proofs, since it seems likely that not all $\NP$ languages admit statistical WI proofs.

Indeed, all known positive results in this regime focus either on $\UP$, which is the class of $\NP$ languages for which each instance $x \in \cL$ has a \emph{unique} witness, or more  generally, on the class $\UIP$, consisting of all languages $\cL$ that have an \emph{unambiguous} interactive proofs. 

An unambiguous interactive proof for a language $\cL$ is an interactive proof system in which, 
for every $x\in\cL$ and for every possible verifier's message there is only one answer that a prover can send that will later allow it convince the verifier to accept conditioned on this answer (with non-negligible probability). In other words, even though the prover may be all-powerful, there is at most one way for it to produce an accepting proof.  We refer the reader to \Cref{sec:prelim:UIP} for the formal definition.

For the case of $\UP$, it was shown in \cite{TCC:RotRot20} that there exists a doubly efficient batch interactive proof for any language $\cL \in \UP$, where proving $x_1, \ldots, x_k \in \cL$ requires communication complexity $m \cdot \polylog(k)$, 
where $m$ is the length of a single witness, and the prover runs in polynomial time given all the witnesses.

It was shown by \cite{STOC:ReiRotRot16} that for every language $\cL$ that has a public-coin $\UIP$ consisting of $\ell$ rounds, where in each round the prover and verifier send $a$ bits, there exists a batch unambiguous interactive proof for proving that $x_1,\ldots,x_k\in \cL$, where the communication complexity is at least $(\mathsf{poly}(a)\cdot k^\eps + k) \cdot \ell^{1/\eps}$, the number of rounds increase only by a constant factor (that grows exponentially with $1/\eps$), and the prover's and verifier's runtime increase by $\poly(k,n,\ell)$, where $n=|x_i|$. 
It is critical that the underlying protocol being batched is public-coin\footnote{Note that the Goldwasser-Sipser transformation from any interactive proof to a public-coin interactive proof does not preserve prover efficiency \cite{STOC:GolSip86,STOC:Vadhan00}.}, i.e. that the verifier's messages are simply uniformly random coins.

\subsection{Our Results}

We construct a batch unambiguous interactive proof for any language  $\cL$ that has an $\ell$-round  public-coin unambiguous interactive proof, where the communication complexity of the batch interactive proof for proving $x_1,\ldots,x_k\in\cL$ grows by a factor of $\ell\cdot \polylog(k)$,  the round complexity grows by a factor of $\polylog(k)$, and the runtime of the prover and verifier grow polynomially in $\ell, n, k$ where $n=|x_i|$.

\begin{theorem}[Our \bUIP, Informal]
\label{thm:informal-bUIP}
Let $\cL$ be a language with a public-coin \emph{unambiguous interactive proof} system that on input $x$ of length $n$ (and where the prover may have an additional input $w$)
runs in $\ell = \ell(n)$ rounds, in each round $a=a(n)$ bits are sent, 
and the verifier's verification circuit is a log-space uniform boolean circuit of depth at most $\polylog(k)$.
Then there is a public-coin unambiguous interactive proof $\prot{\bUIP}$ for $\cL^{\otimes k}$, that on input  $(x_1,\ldots, x_\batchK)$, each of length $n$
(and where the prover may have additional input $w_1,\ldots,w_\batchK$),
runs in 
$\ell \cdot \polylog(\batchK)$ rounds, where in each round $a\cdot \ell \cdot \polylog(\batchK) + \poly(\ell)$ bits are sent.  Moreover, the prover's computational complexity increases multiplicatively by $\poly(\batchK, \ell, n)$ and the verifier's computational complexity becomes 
$(\ell \cdot \Vtime + \batchK n a^2\cdot\poly(\ell)) \cdot \polylog(\batchK, n, a, \ell)$,
where $\Vtime$ is the verifier runtime of the base protocol.
\end{theorem}
\begin{remark} 
    The assumption that the verification circuit is log-space uniform is without loss of generality by the Cook-Levin reduction~\cite{cook1971complexity,levin1973universal} since any efficient verifier, as a Turing machine, can be simulated by a log-space uniform boolean circuit of a comparable size.
    The additional assumption that the verifier's verdict circuit is of low depth is relatively mild, since we can generically flatten out this circuit by having the prover send the values of all the wires in the verdict circuit. 
    This increases the communication complexity (in a single round).
    \label{rmk:Vdepth}
\end{remark}
We provide a high-level overview of our batch $\UIP$ protocol in \Cref{sec:overview}.  We note that it uses a component of the batch $\UP$ protocol of \cite{TCC:RotRot20}, which we refer to as the Instance Reduction protocol.  We need to generalize this protocol to hold with respect a distance measure, which we call execution-wise distance. We elaborate on this in \Cref{sec:overview:RR}.

\paragraph{Applications to doubly efficient interactive proofs}
We show how one can use \Cref{thm:informal-bUIP} to improve the state-of-the-art on doubly efficient interactive proofs.
\begin{theorem}[Our Doubly-efficient Interactive Proof, Informal]
    Every language computable in time $T(n) \le n^{O\left(\sqrt{\frac{\log n}{\log\log n}}\right)}$ and space $S(n) = \poly(n)$ has a (public-coin) doubly efficient interactive proof,
    with verifier runtime and communication complexity as $\poly(n)\cdot S(n))$ and prover runtime $T \cdot \poly(n)$.
\end{theorem}
\begin{corollary}[Our Doubly Efficient $\IP$, Informal]
\label{thm:informal-deIP}
Every language decidable in time $T(n)=n^{O\left(\sqrt{\frac{\log n}{\log\log n}}\right)}$ and polynomial space has a (public-coin) doubly efficient interactive proof.
\end{corollary}
For more details, see \Cref{thm:TS,cor:TS}. This improves upon the doubly efficient interactive proof from \cite{STOC:ReiRotRot16}, which is only for languages computed in polynomial space and time $T=n^{(\log n)^\delta}$ for a small constant $\delta>0$ (which is significantly smaller than $1/2$).

\subsection{Additional Related Work}

The area of proof systems has become a popular area of research with numerous papers published every year.  
Most of these works are in the computational setting, where soundness is guaranteed to hold only against computationally bounded cheating provers. We do not elaborate on these works here, since they are less relevant for our result.  We emphasize that our work is in the information theoretic setting, where soundness holds against any (even computationally unbounded) cheating prover. Most of the work in this regime is focused on guaranteeing privacy, such as zero-knowledge or witness indistinguishability. Again, we do not  elaborate on these works since these are less relevant for our result. 

A proof model which turns out to be useful for us, is that of \emph{interactive proofs of proximity} ($\IPP$) \cite{STOC:RotVadWig13}, where a verifier has oracle access to a long input, and it interacts with a prover to check whether the  input is close to one that satisfies a given property, while running in sublinear time. 
The non-interactive variant was explored in \cite{ITCS:GurRot15},
and the significance of the round complexity of $\IPP$s was studied in \cite{ITCS:GurRot17}.
Although we do not focus on sublinear verification in our work, such proofs serve as a stepping stone to obtain our result.
\section{Technical Overview}
\label{sec:overview}

Our batch UIP protocol takes inspiration from the  \bUP protocols of \cite{RRR18,TCC:RotRot20}.  

\paragraph{The batch $\UP$ protocol from \cite{RRR18,TCC:RotRot20}.} The intuition underlying these protocols stems from the simple observation that if we are given $k$ instances $x_1,\ldots,x_k$ that are $d$-far from being in $\cL^{\otimes k}$, in the sense that at least $d$ of the instances are not in $\cL$, then by choosing $k\cdot \polylog(k)/d$ of the instances at random, we have the guarantee that at least one of the selected instances is not in $\cL$ with probability $1-\mathsf{negl}(k)$. In other words, distance from $\cL^{\otimes k}$ allows us to shrink the number of instances. Indeed, in the batch $\UP$ protocol of \cite{TCC:RotRot20},
they create distance and then shrink the number of instances, and they do this iteratively.  
To create distance, they use the $\GKR$ protocol.

\paragraph{The $\GKR$ protocol}  The $\GKR$ protocol \cite{JACM:GolKalRot15} is an unambiguous interactive protocol between a prover who wishes to prove to the verifier that a circuit $C$ on input~$x$ outputs a bit $b$. The communication complexity and the number of rounds grow linearly with the depth of $C$, denoted by $D$, and poly-logarithmically with the size of $C$, denoted by $S$. It has the desired properties of being public-coin and the verdict function requires access to the input $x$ at a single point in the low-degree extension of~$x$.\footnote{The low-degree extension is a specific linear error correcting code (similar to the Reed-Muller encoding).  We define it formally in \Cref{sec:LDE} but the details can be ignored here. }
In particular, if the verifier is given oracle access to the low-degree extension of the input $x$,
then its runtime is $D\cdot\polylog(S)$, which may be poly-logarithmic in the input length. 

The $\GKR$ protocol can be thought of as a reduction from verifying that $C(x)=b$ to  verifying that the low-degree extension of $x$ at point $j$ is equal to a value $v$, where $j$ is a uniformly random point determined by the random messages sent by the verifier in the $\GKR$ protocol, and $v$ is  determined by the answers sent by the prover in $\GKR$ protocol. 
This is also denoted by $x\in \pval(j,v)$, a notation that originated in \cite{STOC:RotVadWig13}.   
The $\GKR$ protocol has the guarantee that if $C(x)\neq b$ then with high probability $x\notin \pval(j,v)$, where we can make this probability as close as we want to~$1$, and obtain probability $1-\mu$ at the price of the communication complexity growing polynomially with $\log ({\mu}^{-1})$. 
Importantly,
the \GKR protocol is run in the \emph{holographic} mode,
namely that the verifier never access the input $x$ directly
--- it just outputs claims about the low-degree extension of $x$.

The batch $\UP$ protocol of \cite{TCC:RotRot20} starts by running the $\GKR$ protocol on the batch verification circuit that has the instance $\bm x=(x_1,\ldots,x_k)$ hardwired (supposedly in $\cL^{\otimes{k}}$), and takes as input a witness $\bm w=(w_1,\ldots,w_k)$, and outputs~$1$ if and only if $\bm w$ is a valid witness for $\bm x\in\cL^{\otimes k}$.  Note that the depth of the batch verification circuit grows logarithmically with $k$, and hence the $\GKR$  protocol has communication complexity that grows only logarithmically in~$k$ (as desired).  It has the guarantee that if $\bm x\notin\cL^{\otimes k}$ then for every $\bm w= (w_1,\ldots,w_k)$,  with high probability $\bm w\notin \pval(j,v)$. 
 The batch $\UP$ protocol runs this $\GKR$ protocol $T=d\cdot m\cdot \polylog(k)$ times, where $d=\polylog(k)$ is the desired distance, in the sense that $d$ witnesses will need to change in order to satisfy the resulting $\PVAL$ constraint, and $m$ is the length of a single witness.  Running the protocol $T$ times in parallel reduces the soundness error to be exponentially small in~$T$.  They obtain the guarantee that if $\bm x\notin \cL^{\otimes k}$ then for any $\bm w=(w_1,\ldots,w_k)$,  with high probability, which is exponentially (in $T$) close to~$1$, it holds that $\bm w\notin \cap_{i=1}^{T}\pval( j_i,v_i)$.  From now on, we use the condensed notation \[\pval(\bm j, \bm v)=\cap_{i=1}^{T}\pval( j_i,v_i),
\]
where $\bm j=(j_1,\ldots,j_T)$ and $\bm v=(v_1,\ldots, v_T)$.  They then rely on the union bound to argue that even if a single $x_i$ is not in $\cL$, 
with high probability,
every $\bm w$ that is $d$-close to the unambiguous witness,\footnotemark{} in the sense that it differs from the unambiguous one on at most $d$ witnesses, satisfies that $\bm w\notin \pval(\bm j,\bm v)$, and hence obtain the desired distance guarantee!  
\footnotetext{If $x_i \notin \cL$, we define the unambiguous witness to be $w_i = 0^m$.}

We emphasize that if $\bm x \notin \cL^{\otimes \batchK}$ then
the batch verification circuit computes the identically zero function, and in particular outputs zero for every witness vector that is close to the unambiguous one.  This is precisely what allows us to obtain distance from running the $\GKR$ protocol.  As we shall see shortly, this will no longer be the case in the $\UIP$ setting, which brings with it new technical challenges.

Note that after creating distance by running the $\GKR$ protocol $T$ times, the new instances are no longer batch instances, rather are instances of the form $\bm x=(x_1,\ldots,x_k)$ together with a joint $\pval$ constraint $\Phi=\Phi_{\bm j,\bm v}$ on the corresponding witnesses $\bm w=(w_1,\ldots,w_k)$.  The guarantee they obtain is that (with high probability) the unambiguous witness $\bm w$ corresponding to $\bm x$ is $d$-far from satisfying the constraint $\Phi$, where distance is defined in terms of the minimum number of witnesses that need to be modified in order for the constraint $\Phi$ to be satisfied. In other words, they consider a new $\UP$ language $\cL'$, where $(\bm x,\Phi)\in\cL'$ if and only if the unambiguous witness $\bm w$ corresponding to $\bm x\in\cL^{\otimes k}$ satisfies the $\pval$ constraint $\Phi(\bm w)=1$.  They obtain the guarantee that 
 $(\bm x,\Phi)$ is $d$-far from $\cL'$ in the sense that $(\bm x',\Phi)\notin \cL'$, for every $\bm x'$ that differs from $\bm x$ on at most $d$ instances.

They then show how to convert any such instance $(\bm x,\Phi)$ that is $d$-far from $\cL'$ into a new instance $(x_{i_1},\ldots,x_{i_{k/{d}}},\Phi')$ that is not in $\cL'$, but without any distance guarantee.\footnote{We note that the instance reduction protocol from \cite{TCC:RotRot20} needs the stronger distance guarantee that at least $d\cdot m$ bits of the unambiguous witness of $\bm x$ need to change in order to satisfy the constraint $\Phi$.  This distance guarantee is indeed promised by the $\GKR$ protocol. We chose to present distance in terms of the number of witnesses that need to be changed since it is consistent with our approach. 
 See \Cref{sec:DistanceOver}.} Recall that $d$ was chosen so that $d={\polylog(k)}$, and hence the number of instances decreased by a factor of $\polylog(k)$. 

They continue as above, and use the $\GKR$ protocol to convert this instance into a new instance of the form $(x_{i_1},\ldots,x_{i_{k/d}},\Phi'')$, where $\Phi''$ is a new $\pval$ instance, with the guarantee that if $(x_{i_1},\ldots,x_{i_{k/d}},\Phi')\notin \cL'$ then $(x_{i_1},\ldots,x_{i_{k/d}},\Phi'')$ is $d$-far from $\cL'$. 
To obtain this reduction, they run the  $\GKR$ protocol (repeated in parallel $T$ times) with respect to the circuit that has $(x_{i_1},\ldots,x_{i_{k/d}},\Phi')$ hardwired and on input $(w_{i_1},\ldots,w_{i_{k/d}})$ outputs~$1$ if and only if these are valid witnesses of $(x_{i_1},\ldots,x_{i_{k/d}})$ and they satisfy the constraint $\Phi'$. They continue to iterate between creating distance and shrinking the number of instances until they are left with $\polylog(k)$ number of instances, at which point the prover simply sends them over.

\paragraph{Our batch $\UIP$ protocol}
We refer to the verifier messages in the $\UIP$ protocol as \emph{queries} and the prover messages as \emph{answers}. 
Our batch $\UIP$ protocol follows a similar blueprint to the batch $\UP$ protocol of \cite{RRR18,TCC:RotRot20}, but for $\UIP$ as opposed to $\UP$, which brings with it several challenges that we elaborate on below. 
The protocol starts by creating distance.

\subsection{Creating Distance}
\label{sec:DistanceOver}

In the $\UIP$ setting, the instances $x_1,\ldots,x_k$ may lack witnesses certifying membership in $\cL$, since $\cL$ need not belong to $\NP$.  
Hence, 
we cannot apply the $\GKR$ protocol on the batch verification circuit that takes as input the witnesses (as such witnesses do not exist).  
Instead, in our $\UIP$ setting, distance is created by running a checksum over the $k$ $\UIP$ protocols. 

\paragraph{Checksum over the $\UIP$s} 
We begin by performing a checksum over the $k$ $\UIP$ instances, 
following the batch $\UIP$ construction of \cite{STOC:ReiRotRot16}.\footnote{We note that in \cite{STOC:ReiRotRot16} the checksum was not done directly on the $\UIP$ but rather the $\UIP$ was converted to an $\mathsf{IOP}$ and the checksum was done on the $\mathsf{IOP}$. We elaborate on this in \Cref{overview:RRR}.}  
Namely, we run the $k$ $\UIP$ protocols in parallel, where the verifier uses the same queries in all the executions and the prover sends its answers via a checksum.  

We assume that the underlying $\UIP$ protocol is public-coin and hence it can be run in the following way. 
Denoting by $\ell$ the number of rounds in the underlying $\UIP$ protocol, in each round $r\in[\ell]$ of the batch $\UIP$ protocol,
the verifier sends a random message $q_r\in\{0,1\}^a$ as in a single $\UIP$ protocol.  The prover generates the $k$ answers $\bm a_r=(a_{r,1},\ldots,a_{r,k})\in\left(\{0,1\}^a\right)^k$, but rather than sending $\bm a_r$ (which is too long), it sends a checksum of $\bm a_r$ (which should be thought of as a ``digest'' of $\bm a_r$). 
Formally, the prover sends the low-degree extension of $\bm a_r$ at $T=d\cdot a\cdot \polylog(k)$ random points, where $d=\polylog(k)$, 
and $a$ is the length of each of the prover's messages in a single execution.  
These points are sampled by the verifier at the beginning of the protocol. 

Now we show that distance is created: 
if the checksummed transcript of answers only deviate from the unambiguous one by at most $d$ executions, 
then at least one execution in the transcript will be rejecting with overwhelming probability.
Suppose the cheating prover sends checksums that correspond to answers deviating in at most $d$ executions, 
then for every $r\in [\ell]$, 
by the checksum's error-correcting property,
it is possible to extract from the $r$-th round checksums the entire set of $r$-th round answers.
By unambiguity,
at least one of the extracted answers will lead to a rejecting transcript with overwhelming probability when $(x_1,\ldots,x_k)\notin \cL^{\otimes k}$.

However, note that the verifier cannot verify the validity of the answers given by the prover, since the answers were not sent in the clear, but rather in the form of a checksum. To this end, we use the $\GKR$ protocol, this time not to create distance, but rather to allow the verifier to verify the validity of the answers.  
Specifically, as we describe below, the $\GKR$ protocol is used to convert the check-summed transcript with the distance guarantee mentioned above, into a $\pval$ claim with an analogous distance guarantee.  

\begin{remark}
    We note that this is where our approach differs substantially from \cite{STOC:ReiRotRot16}, which verified correctness by having the prover open some of the bits in {\em all} $k$ instances (as opposed to using a succinct interactive proof, such as $\GKR$, to verify correctness).  
    As a result,
    denoting the number of iterations in their protocol as $1/\epsilon$ (for some $\epsilon > 0$),
    their batch $\UIP$ incurs an additive factor of $k$ as well as a multiplicative factor of $(1/\epsilon)^{\poly(1/\epsilon)}$ to the communication complexity. 
    We elaborate on this difference in \Cref{overview:RRR}.
\end{remark}

\paragraph{The $\GKR$ protocol} In order to verify the check-summed $\UIP$s, the prover and verifier run the $\GKR$ protocol, this time on the batch verification circuit that has hardwired into it the instances $\bm x=(x_1,\ldots,x_k)$, the verifier's queries $\bm q=(q_1,\ldots,q_\ell)$, and a constraint $\Phi$ on the prover's answers (as defined by the checksum).  The circuit takes as input all the prover's answers $\bm a= (\bm a_1,\ldots,\bm a_\ell)$ and outputs~$1$ if and only if for every $j\in[k]$ the transcript $(q_1, a_{1,j},\ldots q_\ell,  a_{\ell,j})$ is accepted with respect to $x_j$, and the answers $\bm a$ satisfy $\Phi$.  

The distance guarantee of the check-summed protocol implies that with high probability, every transcript of answers $\bm a= (\bm a_1,\ldots,\bm a_\ell)$ that is $d$-close to the unambiguous one is rejected by the batch verification circuit, where here (and throughout our work) we define $d$-closeness to mean that it differs in at most $d$ executions. 

\begin{remark}
    We emphasize that our definition of $d$-closeness is not the standard Hamming distance (which is bit-wise distance), but rather it is an execution-wise distance.  
    Since our definition of distance deviates from the one used in \cite{TCC:RotRot20} (which is Hamming distance) we need to prove that the instance reduction protocol from \cite{TCC:RotRot20} works with our notion of execution-wise distance.  We elaborate on this in \Cref{sec:overview:RR}.
\end{remark}

After running the $\GKR$ protocol (in parallel) the parties obtain $\bm j$ and $\bm v$ such that for every $\bm a$ that is $d$-close to the unambiguous one, with high probability, $\bm a\notin \pval(\bm j,\bm v)$. We then rely on the union bound, following the approach in the $\UP$ setting,   to argue that with high probability for every $\bm a$ that is $d$-close to the unambiguous one it holds that $\bm a\notin \pval(\bm j,\bm v)$. 
 Therefore, we can apply an instance reduction protocol (similar to the one from \cite{TCC:RotRot20}, which we elaborate on in \Cref{sec:overview:RR}), which reduces the problem to verifying that $(x_{i_1},\ldots,x_{i_{k/d}},\bm q,\Phi')\in \cL'$ for some $\Phi'$, where the language $\cL'$ is defined so that $(x_{i_1},\ldots,x_{i_{k/d}},\bm q,\Phi')\in \cL'$  if and only if the unambiguous answers, corresponding to the instances $x_{i_1},\ldots,x_{i_{k/d}}$ and queries  $\bm q$, satisfy $\Phi'$.\footnote{We note that $\cL'$ is not necessarily in $\NP$ (as was not $\cL$).}

The guarantee we have is that if $(x_1,\ldots,x_k,\bm q, \Phi)$ is $d$-far from $\cL'$, in the sense that every answer $\bm a'$ that is $d$-close to the unambiguous one, does not satisfy $\Phi$, then with high probability $(x_{i_1},\ldots,x_{i_{k/d}},\bm q, \Phi')\notin \cL'$.  We note that this instance does not have a distance guarantee, and generating distance for instances in $\cL'$ introduces new challenges that we elaborate on below.   

\paragraph{Obtaining distance for $\cL'$}  It is tempting to try and continue as in the $\UP$ setting,  and apply the $\GKR$ protocol to create distance.  The problem is that at this point the cheating prover knows all the verifier's queries $\bm q$ in advance, and thus can easily generate accepting answers $\bm a$.  Moreover, a cheating prover can generate answers $\bm a$ that are close to the unambiguous ones and which cause the $\GKR$ circuit to accept! 
Hence, the $\GKR$ protocol cannot be used to create distance.  

We emphasize that the problem stems from the fact that the prover knows all the queries $\bm q$ ahead of time, and thus can easily generate accepting answers even for $x\notin \cL$.  We get around this problem by adding the following random continuation protocol from \cite{STOC:ReiRotRot16}.  

\paragraph{Random continuation to the rescue} 
  
In the random continuation protocol, the verifier chooses fresh queries $q'_1,\ldots,q'_\ell\gets \{0,1\}^a$, and the prover and verifier run a checksum of $\ell\cdot k/d$ $\UIP$ protocols, where in the $(i\cdot k/d +j)$'th protocol, where $i\in\{0,1,\ldots,\ell-1\}$ and $j\in\{1,\ldots,k/d\}$, the verifier sends messages $(q_1,\ldots,q_i,q'_{i+1},\ldots,q'_\ell)$ and the prover replies with answers corresponding to the $j$'th instance.  Thus, the prover and verifier run $\ell$ $\UIP$ protocols per instance, and thus a total of $\ell\cdot k/d$ $\UIP$ protocols, all run in a checksum.

Each of the prover's checksums consist of $T=d'\cdot a\cdot \polylog(k)$ points in the low-degree extension of each vector of answers, 
where here we take $d'=\ell\cdot d$. 
\begin{remark}
The reason for this choice of distance $d'$ is that the instance reduction protocol, applied after the distance is generated, reduces the number of instances by a factor of~$d'$.  However, the random continuation increases the number of protocols by a factor of $\ell$, and hence to ensure that the number of instances goes down by a factor of~$d$, we need to take $d'=\ell\cdot d$. 
\end{remark}

One can argue (as above) that if the prover deviates from the (unambiguous) random continuation protocol on at most $d'$ of the executions (where in each execution it sends a message of length $a$) then one can efficiently extract the messages of the prover from the check-summed transcript. Therefore, if there exists $i\in [\ell]$ such that the $i$'th round check-summed answers sent by the prover, corresponding to $(q_1,\ldots,q_i)$, were not the unambiguous ones, then by the unambiguity property the prover will fail to provide accepting answers for random continuation $(q'_{i+1},\ldots, q'_\ell)$, even if the corresponding instance is in $\cL$.  This implies that if there exists $i\in [\ell]$ such that the $i$'th check-summed answers corresponding to $(q_1,\ldots,q_i)$, are not the unambiguous ones, and the prover deviates from the unambiguous answers in at most $d'$ executions,  then we can efficiently extract from the check-summed transcript a rejecting transcript of answers sent by the prover. 
This gives us the desired distance guarantee,
and we can now use the $\GKR$ protocol, as above, to reduce this to a $\pval$ instance $\pval(\bm j, \bm v)$ with the guarantee for every $\bm w$ that is $d'$-close to the unambiguous one, $\bm w\notin \pval(\bm j, \bm v)$.  Hence, we can use the instance reduction protocol (similar to the one given in \cite{TCC:RotRot20}) to reduce the number of instances from $\ell\cdot k/d$ to $\ell\cdot k/d\cdot 1/d'=k/(d^2)$.

\begin{remark}
    We note that in our actual batch $\UIP$ protocol, for the sake of simplicity, we use distance $d'=\ell\cdot \polylog(k)$ from the beginning of the protocol, and do not treat the initial reduction, from a batch instance $(x_1,\ldots,x_k)\in \cL^{\otimes k}$ to an instance in $\cL'$, differently.
\end{remark}

\subsection{Instance Reduction}\label{sec:overview:RR}
\paragraph{The Need for Execution-wise Distance}   
Recall that the distance guarantee provided by our distance generation protocol, described in \Cref{sec:DistanceOver}, is that if we start with a false statement on say $k'$ instances, then after the distance generation protocol (which includes the checksum protocol followed by the $\GKR$ protocol), we obtain a $\pval$ constraint $\pval(\bm j, \bm v)$ that is not satisfied by any 
transcript of answers that differs on at most $d'$ executions from the prescribed transcript of answers (with respect to fixed verifier's queries). For the next step, we would like to use the instance reduction protocol from \cite{TCC:RotRot20} to further reduce this to a false claim on a subset of $k'/d'$ instances (at the cost of losing our distance guarantee).

Unfortunately, the \cite{TCC:RotRot20} protocol is guaranteed to produce a false claim on a subset of $k'/d'$ instances, only if the original $\pval$ constraint $\pval(\bm j, \bm v)$ is not satisfied by any 
transcript that differs from the prescribed transcript in at most $(d'\cdot \ell\cdot a)$-bits!  Specifically, the issue is that these two protocols use a different distance measure, and these are not compatible: One guarantees execution-wise distance while the other requires Hamming distance. 

We note that being $d'$-far in execution distance trivially implies being $d'$-far in Hamming distance.  However, to reduce the number of instances by a $d'$-factor, the \cite{TCC:RotRot20} protocol requires the Hamming distance to be $d'\cdot \ell\cdot a$, where $\ell\cdot a$ is the length of a single transcript.  This requirement is much stronger than requiring the execution distance to be~$d'$.  In other words, if we only have the guarantee that the execution distance is $d'$, then applying \cite{TCC:RotRot20} will result with the number of instances being $k'\cdot \frac{d'}{\ell\cdot a}$
, which offers no decrease in the number of instances, unless we take $d'$ to be larger than $\ell\cdot a$, in which case we invest $d'\cdot \ell\cdot a\cdot \polylog(k)$ in communication to decrease the number of instances by a factor of $d'$. This will result in a suboptimal batch $\UIP$ result. While further optimizations are possible here, it seems that any purely Hamming distance based approach will result in a larger overhead than what we achieve.

Instead, we choose to make the two distance measures compatible by changing the instance reduction protocol to guarantee that if the $\pval$ constraint is not satisfied by any transcript that is $d'$-close in execution distance to the prescribed one, then it results with a false claim on a subset of $k'/d'$ instances.  We next elaborate on how this is done.

\paragraph{Overview of the original instance reduction protocol for $\PVAL$ from \cite{TCC:RotRot20}}
The verifier in the instance reduction protocol of \cite{TCC:RotRot20} takes as input $\bm j, \bm v$ that defines the $\pval$ constraint.
In the YES case, the claim $\bm w \in \pval(\bm j, \bm v)$ is true, where  $\bm w$ is the implicit
vector of unique witnesses $\bm w = (w_1,\ldots,w_k) \in (\bin^m)^k$.
In the NO case, $\bm w$ as a string in $\bin^{k \cdot m}$ is $d$-far in Hamming distance from $\pval(\bm j,\bm v)$ with $d = m \cdot \polylog(k)$.

They show how to subsample a $\frac{1}{d}$ fraction of the instances together with a new false $\pval$ claim on the corresponding witnesses.  
This is not straightforward since the initial constraint is a global constraint, and cannot be applied to each witness $w_i$ separately.  
The authors in \cite{TCC:RotRot20} propose to transform the global constraint into a bunch of localized constraints,
each defined over $\frac{1}{d}$ fraction of the instances. 
In order to obtain these localized claims,
the \cite{TCC:RotRot20} protocol proceeds in at most $\log k$ iterations as follows:
\begin{enumerate}[label=\arabic*.]
    \item The multi-linearity of $\pval(\bm j,\bm v)$
    allows the prover to break the $\pval$ claim into two halves: $\pval_0$ and $\pval_1$, where each claim is over $k/2$ witnesses.  Specifically, the claim $\pval_b$ is over ${\bm w}_b$ where ${\bm w}_0$ and ${\bm w}_1$ are defined by: 
    \ifFOCS
    \begin{align*}
    {\bm w}_0&=(w_1,\ldots, w_{k/2}),\\
    \text{ and } {\bm w}_1&=(w_{k/2+1},\ldots, w_{k}).
    \end{align*}
    \else
    \[{\bm w}_0=(w_1,\ldots, w_{k/2})~~\mbox{ and }~~ {\bm w}_1=(w_{k/2+1},\ldots, w_{k}).\]
    \fi

    This results in $(\bm w_0,\bm w_1)$ being $d$-far from the new claim $\pval_0\times \pval_1$. 
    For this overview, assume that these $w_0$ and $w_1$ are both $d/2$ far from the respective $\pval$ instances.

    \item 
    The next step would be to try to fold these two claims into a single claim of ``size $k/2$'' {\em while keeping the absolute distance to be~$d$}.
    
    A naive folding by taking a random linear combination of $\bm w_0$ and $\bm w_1$ along with the corresponding linear combination of $\pval_0$ and $\pval_1$ may reduce the absolute distance to $d/2$ since the errors in ${\bm w}_0$ and ${\bm w}_1$ may overlap.  Instead \cite{TCC:RotRot20} applies a more sophisticated folding, as follows.
    
    \item The verifier uses a \textit{pairwise-independent} family $\pi$ to transform $\bm w_1$ and $\pval_1$ into a new claim $\pval_1'$ about $\pi(\bm w_1)$. The purpose is to scramble the $d/2$ errors that prevent $\bm w_1$ from satisfying $\pval_1$.
    
    \item The verifier ``folds'' the two claims into one,
    by taking a random linear combination of $\bm w_0$ and $\pi(\bm w_1)$,
    and a random linear combination of $\pval_0$ and $\pval_1'$
    to get a new claim that is also roughly $d$-far.
    
\end{enumerate}
Each iteration allows us to reduce a $d$-far claim to another $d$-far claim but on a ``folded'' instance of half the size.
The verifier repeats this until it is left with a $d$-far claim over roughly $d$ ``folded'' witnesses. 
Then the verifier kindly asks the prover to send the batch of ``folded'' witnesses $\bm w^* = (\bm w^*_1, \ldots, \bm w^*_{d})$ in the clear,
which requires $d\cdot m$ communication.
Then it checks whether the batch of the ``folded'' witnesses $\bm w^*$ satisfies the final ``folded'' $\pval$ claim. 
If not, it rejects immediately. 
Otherwise, 
in the NO case,
$\bm w^\ast = (\bm w^\ast_1,\ldots, \bm w^\ast_{d})$ is still $d$-far from any $\bm w^*$ that satisfies the final $\pval$ claim,
so the sent $\bm w^*$ must differ from $\bm w^\ast$ on at least $d$ instances.
Therefore,
the verifier can indeed subsample constantly many of these instances $\bm w^*_i$ to detect an inconsistency with high probability, and the inconsistency depends only on the ($k/d$) witnesses folded in $\bm w^*_i$.

\paragraph{Working with execution-wise distance} 
We generalize (and tweak) the protocol from \cite{TCC:RotRot20} for Hamming distance to work for the execution-wise distance.  
The modification of the folding step is illustrated in \Cref{fig:folding}.
To this end, 
    we view the batch of transcripts $\bm a$ as a matrix where each row corresponds to a transcript of a single instance (\Cref{fig:a-matrix}).
    We split $\bm a$ into its top half submatrix $\bm a_0$ and bottom half submatrix $\bm a_1$ (\Cref{fig:a-matrix-split}).
    Execution-wise distance is preserved in the folding step by choosing a sufficiently random permutation $\pi$ to scramble each column of $\bm a_1$ independently (\Cref{fig:a-matrix-permuted}).
    (Note that this is different from scrambling the entire $\bm a_1$, which moves errors across columns, as in the original protocol from \cite{TCC:RotRot20}.)

    \ifFOCS
    \input{FOCS-permutation-picture.tex}
    \else
    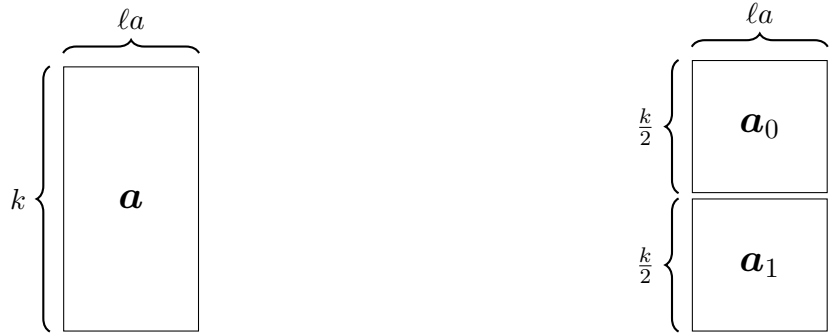
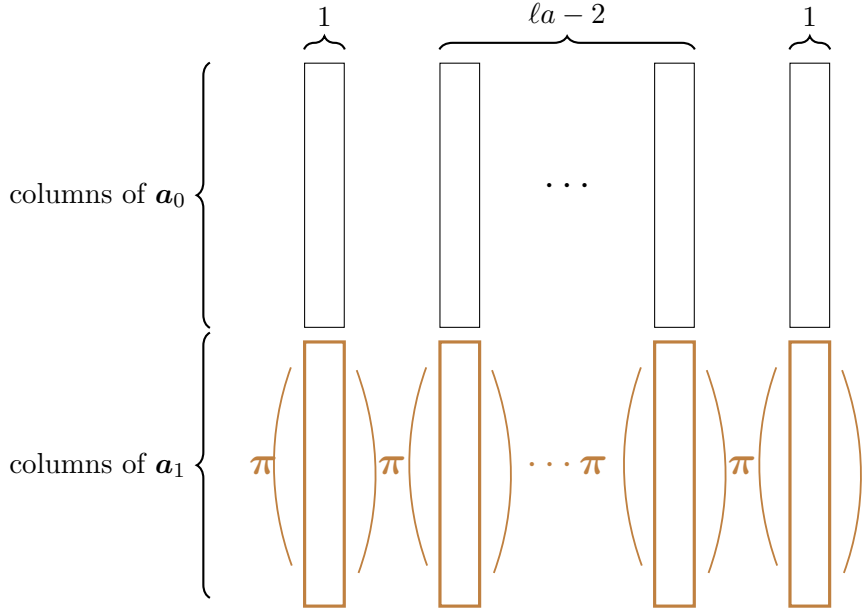
\begin{figure}[H]
    \begin{subfigure}[b]{0.5\textwidth}
        \centering
        \begin{tikzpicture}
            \matrix(m)[matrix of math nodes,
                    nodes={draw, minimum width=51pt, minimum height=100pt, font=\fontsize{15}{18}\selectfont},
                    ]{
                    \bm a\\
                };
            \draw[decorate,
                decoration={
                    mirror,
                    brace,
                    amplitude = 5pt ,
                    raise = 5pt,
                },
                thick,
            ] (m-1-1.north west) -- (m-1-1.south west) node [midway, left=10pt] {$k$};
            \draw[decorate,
                decoration={
                    brace,
                    amplitude = 5pt ,
                    raise = 5pt,
                },
                thick,
            ] (m-1-1.north west) -- (m-1-1.north east) node [midway, above=10pt] {$\ell a$};
        \end{tikzpicture}
        \subcaption{$\bm a$'s rows are the transcripts of the instances.}
        \label{fig:a-matrix}
        \end{subfigure}
        \begin{subfigure}[b]{0.5\textwidth}
            \centering
        \begin{tikzpicture}[decoration=brace]
            \matrix(m)[matrix of math nodes,
                    row sep=2pt,
                    nodes={draw, minimum width=51pt, minimum height=50pt, font=\fontsize{15}{18}\selectfont}]{
                    \bm a_0 \\
                    \bm a_1 \\
                };
            \draw[decorate,
                decoration={
                    mirror,
                    brace,
                    amplitude = 5pt ,
                    raise = 5pt,
                },
                thick,
            ] (m-1-1.north west) -- (m-1-1.south west) node [midway, left=10pt] {$\frac{k}{2}$};
            \draw[decorate,
                decoration={
                    mirror,
                    brace,
                    amplitude = 5pt ,
                    raise = 5pt,
                },
                thick,
            ] (m-2-1.north west) -- (m-2-1.south west) node [midway, left=10pt] {$\frac{k}{2}$};
            \draw[decorate,
                decoration={
                    brace,
                    amplitude = 5pt ,
                    raise = 5pt,
                },
                thick,
            ] (m-1-1.north west) -- (m-1-1.north east) node [midway, above=10pt] {$\ell a$};
        \end{tikzpicture}
        \subcaption{$\bm a$ split into two halves $\bm a_0$ and $\bm a_1$.}
        \label{fig:a-matrix-split}
    \end{subfigure}
    \begin{subfigure}[b]{\textwidth}
        \centering
        \begin{tikzpicture}[decoration=brace]
            \matrix(m)[matrix of math nodes,
                    row sep=2pt,
                    nodes={draw, minimum width=15pt, minimum height=100pt, font=\fontsize{15}{18}\selectfont},
                    column 1/.style={nodes={minimum width=30pt}},
                    column 2/.style={nodes={minimum width=15pt}}, 
                    column 3/.style={nodes={minimum width=35pt}},
                    column 4/.style={nodes={minimum width=15pt}},
                    column 5/.style={nodes={minimum width=65pt}},
                    column 6/.style={nodes={minimum width=15pt}},
                    column 7/.style={nodes={minimum width=35pt}},
                    column 8/.style={nodes={minimum width=15pt}},
                    column 9/.style={nodes={minimum width=30pt}}]{
                    |[draw=none]| \ & \ & |[draw=none]| \ & \ & |[draw=none]| \cdots & \ & |[draw=none]| \ & \ & |[draw=none]| \ \\
                    |[draw=none]| \bm{\textcolor{brown}{\pi}} & |[draw=brown, line width=1.2pt]| & |[draw=none]| \bm{\textcolor{brown}{\pi}} & |[draw=brown, line width=1.2pt]| & |[draw=none]| \textcolor{brown}{\cdots \bm\pi} & |[draw=brown, line width=1.2pt]| \ & |[draw=none]| \bm{\textcolor{brown}{\pi}} & |[draw=brown, line width=1.2pt]| & |[draw=none]| \ \\
                };
                \draw[brown, line width=0.7pt, bend right=20, shorten >=13pt, shorten <=14pt] (m-2-1.north east) to (m-2-1.south east);
                \draw[brown, line width=0.7pt, bend left=20, shorten >=14pt, shorten <=12pt] (m-2-2.north east) to (m-2-2.south east);
                \draw[brown, line width=0.7pt, bend right=20, shorten >=13pt, shorten <=14pt] (m-2-3.north east) to (m-2-3.south east);
                \draw[brown, line width=0.7pt, bend left=20, shorten >=14pt, shorten <=12pt] (m-2-4.north east) to (m-2-4.south east);
                \draw[brown, line width=0.7pt, bend right=20, shorten >=13pt, shorten <=14pt] (m-2-5.north east) to (m-2-5.south east);
                \draw[brown, line width=0.7pt, bend left=20, shorten >=14pt, shorten <=12pt] (m-2-6.north east) to (m-2-6.south east);
                \draw[brown, line width=0.7pt, bend right=20, shorten >=13pt, shorten <=14pt] (m-2-7.north east) to (m-2-7.south east);
                \draw[brown, line width=0.7pt, bend left=20, shorten >=14pt, shorten <=12pt] (m-2-8.north east) to (m-2-8.south east);
            \draw[decorate,
                decoration={
                    mirror,
                    brace,
                    amplitude = 5pt ,
                    raise = 5pt,
                },
                thick,
            ] (m-1-1.north west) -- (m-1-1.south west) node [midway, left=10pt] {columns of $\bm a_0$};
            \draw[decorate,
                decoration={
                    mirror,
                    brace,
                    amplitude = 5pt ,
                    raise = 5pt,
                },
                thick,
            ] (m-2-1.north west) -- (m-2-1.south west) node [midway, left=10pt] {columns of $\bm a_1$};
            \draw[decorate,
            decoration={
                brace,
                amplitude = 5pt,
                raise = 5pt,
            },
            thick,
            ] (m-1-2.north west) -- (m-1-2.north east) node [midway, above=10pt] {$1$};
            \draw[decorate,
            decoration={
                brace,
                amplitude = 5pt,
                raise = 5pt,
            },
            thick,
            ] (m-1-4.north west) -- (m-1-6.north east) node [midway, above=10pt] {$\ell a - 2$};
            \draw[decorate,
                decoration={
                    brace,
                    amplitude = 5pt,
                    raise = 5pt,
                },
                thick,
            ] (m-1-8.north west) -- (m-1-8.north east) node [midway, above=10pt] {$1$};
        \end{tikzpicture} 
        \subcaption{The same permutation \bm{\textcolor{brown}{$\pi$}}$ : \bin^{\frac{k}{2}} \to \bin^{\frac{k}{2}}$ is applied to the columns of $\bm a_1$.}
        \label{fig:a-matrix-permuted}
    \end{subfigure}
    \caption{Our folding step works with execution-wise distance by permuting each column randomly.}
    \label{fig:folding}
\end{figure}
    \fi

    We show that this protocol (as well as the \cite{TCC:RotRot20} protocol) also has unambiguous soundness when the underlying field is large. 
    Thus, it can also be used as a building block for a \UIP protocol, leading to many of our applications.

\subsection{Efficiency Gain over the RRR16 Protocol}\label{overview:RRR}
The \bUIP protocol of \cite{STOC:ReiRotRot16} has communication complexity at least 
\[\batchK \cdot (\ell/\epsilon)^{\poly(1/\epsilon)} + \poly(a) \cdot (\ell/\epsilon)^{\poly(1/\epsilon)} \cdot \batchK^{\epsilon}\]
for an arbitrarily small constant $\epsilon > 0$.
On a closer inspection,
their protocol allows the choice of a slightly sub-constant $\epsilon = \epsilon(\batchK)$,
at the cost of creating a super-polynomial dependency on $\ell$ (i.e. the terms involving $\ell^{1/\epsilon}$) in the communication complexity.
In contrast,
our \bUIP always achieves a communication complexity of $O(a \cdot \ell^2 \cdot \polylog(\batchK) + \poly(\ell))$.

Both \bUIP protocols use the idea of running many protocols in a checksum and running the random continuation protocol (also in a checksum).  These are used to enforce that the cheating prover will be caught if it deviates from the protocol only in a few instances. Nevertheless, our protocol is very different in how the verifier catches the prover if it indeed only deviated in a few instances. 
While the verifier in \cite{STOC:ReiRotRot16} checks explicitly on its own that this is the case, by asking the prover to reveal some of these check-summed messages,  
the verifier in our protocol delegates these checks  back to the prover (via the \GKR protocol).
This reduces the overall costs significantly.

Specifically,
in the protocol of \cite{STOC:ReiRotRot16}, the prover first encodes the base \UIP into an (unambiguous) \emph{Interactive Oracle Proof} (\IOP).
In such a system,
the prover and verifier interact,
with the verifier sending random coins and receiving answers from the prover.
The verifier does not read the prover's messages immediately,
and only after the interaction is complete, it makes its verdict based on a set of random queries into very few locations of the stored transcript.
Similar to our usage,
the prover and verifier run the $\batchK$ instances of the \IOP implicitly,
with the prover only sending checksums of the underlying messages.

Importantly, after the interaction,
in \cite{STOC:ReiRotRot16}, the verifier sends the set of query locations that it would like to check in the \IOP to the prover,
and the prover responds 
with {\em all} the messages on these locations.
In addition to checking that all $\batchK$ instances of the \IOP are accepting,
the verifier also confirms that the checksums on those messages agree with what the prover committed to during the interaction. Note that this requires communicating  $\Omega(\batchK)$ bits.

The guarantee they obtain is that if originally there was at least one false statement, then if the verifier accepts the openings then it must be the case that the prover deviated on many instances of the \IOP,
so the verifier can subsample the instances to reduce the number of instances.

At this point, the random queries of the verifier in the original $\IOP$ have been revealed through the interaction,
and to remedy this, the verifier in \cite{STOC:ReiRotRot16} runs the \IOP corresponding to a random continuation of the base \IOP (constructed from the base \UIP) in the next distance-creation (checksum) phase.
The random continuation always incurs a multiplicative factor of $\ell$ blowup in the \IOP,
so if this process is iterated for $1/\epsilon$ times, the transcript length of the final \IOP is $\poly(a)\cdot \ell^{1/\epsilon}$.

In our protocol,
the verifier avoids this accumulation of the factor $\ell$ in the transcript length by instead ``flattening'' the transcripts produced by the random continuation. Namely, 
instead of treating the implicit transcripts as $\batchK$ transcripts of the random continuation protocol (where in each round the prover sends a message of length $a \cdot \ell$),
it treats them as $\batchK \cdot \ell$ instances of the \emph{base \UIP}.
This is desirable because if we generate a distance of $d \ge \ell\cdot\polylog(\batchK)$ by the checksum,
we can offset this factor of $\ell$ in the number of instances in the instance reduction step.
A caveat is that the verifier should additionally check a joint constraint over each of the $\ell$ hybrid transcripts derived from the same instance by the random continuation,
namely that they have shared prefixes.
This check can be delegated to the prover via the \GKR protocol.
This optimization is not possible under the framework of \IOP of \cite{STOC:ReiRotRot16},
since a priori the \IOP of the random continuation protocol cannot be flattened into $\ell$ independent \IOP's.

\subsection{Doubly Efficient Interactive Proofs}
Our \bUIP protocol can be used to obtain new (and more powerful) doubly-efficient interactive proofs for bounded-space computations,
by applying a similar technique as that presented in \cite{STOC:ReiRotRot16}.
Intuitively, any \bUIP protocol can be used to obtain a doubly-efficient interactive proof for poly-space computations.

Formally,
let $\cL\in \TS$ be a language that is decidable by a Turing machine $\cM$ in time $T(n)$ and space $S(n)$.
Let $t = t(n) \le T(n)$ be a parameter,
and let $\cL_t$ be the set of all tuples $(x, w_1, w_2) \in \bin^n \times \bin^{O(S(n))} \times \bin^{O(S(n))} $ such that the Turing machine $\cM$ transitions from configuration $w_1$ to configuration $w_2$ on input $x$ in exactly $t$ steps.

Crucially,
for every $\batchK = \batchK(n) \in \NN$,
the language $\cL_{\batchK \cdot t}$ can be expressed as a batch of $\cL_{t}$.
Specifically, for input $(x, w_0, w_{\batchK \cdot t})$:
\begin{itemize}
    \item The prover runs $\cM$ on input $x$ for $\batchK \cdot t$ steps,
    and obtains the intermediate configurations $w_1, w_2, \ldots, w_{\batchK \cdot t}$. (Note this step can be skipped if the prover already has the intermediate configurations.)
    \item The prover sends the $\batchK$ intermediate configurations $w_t, w_{2t}, \ldots, w_{\batchK \cdot t}$ to the verifier.
    \item Both parties run the \bUIP protocol on the $\batchK$ instances $\{(x, w_{(i-1) \cdot t}, w_{i \cdot t})\}_{i=1}^{\batchK}$ of $\cL_{t}$,
    and the verifier accepts if and only the \bUIP protocol accepts.
\end{itemize}
Therefore,
an efficient batching protocol such as our protocol allows the verifier to verify $(\batchK \cdot t)$-step transitions by batch-verifying $t$-step transitions,
at the cost of a small (polylogarithmic) increase in complexity.
By repeating this process with a careful balancing of the parameters,
the verifier gradually increases the transition length $t$ that it can verify,
until it reaches $\cL_{T(n)}$.

We show that with $\batchK = \poly(n)$,
this process can be repeated for $O(\sqrt{\frac{\log n}{\log \log n}})$ times
before the communication complexity becomes superpolynomial,
obtaining a doubly-efficient interactive proof for languages in $\TS$ where $T = n^{O\left(\sqrt{\frac{\log n}{\log \log n}}\right)}$ and $S = \poly(n)$,
where the prover runtime is $T \cdot \poly(n)$,
and the verifier runtime is $\poly(n)$.  

\section{Preliminaries}
We adhere to the following convention when possible.
\begin{itemize}
    \item Lower case letters $a,b$  mean scalars, 
        while bolded lower case letters $\bm a, \bm b$ mean vectors. 
        Upper case bold letters like $\bm A, \bm B$ are matrices.
    \item For an integer $n$, we denote by $[n]$ the set $\{1,2,\ldots,n\}$.
        We let $\FF$ denote a finite field.
        $\GF(2)$ denotes the finite field with 2 elements.
    \item
        For a vector $\bm u \in \FF^n$,
        and a subset $\cS \subset [n]$,
        $\bm u\vert_\cS \in \FF^{\abs{\cS}}$ denotes the subvector of $\bm u$ indexed by $\cS$.
        For a sequence of vectors $\set{\bm u_i}_{i \in \I}$, where each $\bm u_i\in \FF^n$, 
        the notation $(\bm u_i)_{i \in \I} \in \FF^{n \times |\I|}$ represents the matrix whose columns are the vectors $\bm u_i$ for every $i \in \I$.
    \item
        Given a matrix $\bm A = (\bm a_1,\ldots,\bm a_n) \in \FF^{m \times n}$, 
        $\bm a_j \in \FF^m$ denotes its $j$-th column.  
        We also use $\bm A[i, :]$ to denote the $i$-th row of $A$.
        For a subset $\cS \subset [m]$,
        $\bm A[\cS,:]$ denotes the submatrix of $\bm A$ consisting of rows indexed by $\cS$.
        We also consider matrices $A \subset \FF^{k \times n}$ whose rows are indexed by an ordered set of $k$ pairs,
        $\mathcal{T} = \{(u^1, v^1), \ldots, (u^k, v^k)\}$,
        and we denote by $\bm A[(u^i, v^i), :]$ the row of $\bm A$ indexed by $(u^i, v^i) \in \mathcal{T}$,
        and for $\cS \subset \mathcal{T}$, $\bm A[\cS, :]$ denotes the submatrix of $\bm A$ consisting of rows indexed by $\cS$.
    \item
        $\Delta(\bm u,\bm v)$ is the (absolute) Hamming distance between the vectors $\bm u$ and $\bm v$.
        Denote by $\Delta(\bm u)$ the vector $\bm u$'s \emph{Hamming weight},
        defined to be the number of non-zero elements in $\bm u$.
        For any $d \in \NN$,
        two vectors $\bm u, \bm v$ are $d$-close (in Hamming distance) if $\Delta(\bm u,\bm v) \le d$.
        On a linear space $\cU\subset \FF^a$,
        let $\Delta(\cU) \coloneqq \min_{\bm u \in \cU, \bm u \neq \bm 0}\Delta(\bm u)$.
    \item
        Given a Boolean circuit $\cV$,
        $\size(\cV)$ is the number of gates in the circuit,
        and $\depth(\cV)$ is the depth of the circuit. 
\end{itemize}

An important distance notion that we consider is the ``column distance,'' denoted by \emph{$\Delta_c$}, between matrices.
\begin{definition}[$\Delta_c$-distance]
\label{def:Deltac-dist}
Let $\FF$ be a finite field.
For matrices $\bm A = (\transcMat_1,\ldots, \transcMat_\ncol) \in \FF^{\K \times \ncol}$ and $\bm B = (\bm b_1, \ldots, \bm b_\ncol)\in \FF^{\K \times \ncol}$,
the \emph{$\Delta_c$-distance} between $\bm A$ and $\bm B$, denoted by  $\Delta_c(\bm A, \bm B)$, is the maximum Hamming distance between corresponding columns of $\bm A$ and $\bm B$,
i.e. 
\[\Delta_c(\bm A, \bm B) = \max_{i \in [\ncol]}\Delta(\transcMat_i, \bm b_i).
\]
If $\Delta_c(\bm A, \bm B) \le d$,
we say $\bm A$ and $\bm B$ are \emph{$\Delta_c$-$d$-close},

Furthermore, let $d \in \NN$ be some distance parameter,
then given $\bm A \in \FF^{\K \times \ncol}$, 
its \emph{$\Delta_c$-$d$-ball} is defined to be 
\[\rowball(\bm A) \coloneqq \set{\bm A' \in \FF^{\K \times \ncol}: \Delta_c(\bm A, \bm A') \le d}.
\]
\end{definition}

\begin{remark}
  Observe that $\Delta_c$ is a (non-tight) lower bound on how many rows have to be modified to transform one matrix into another.  
\end{remark}

\subsection{Low Degree Extensions and Polynomial Valuation (\pval)}
\label{sec:LDE}
All finite fields $\FF$ considered in this work are always \emph{constructible} in the following sense.
\begin{definition}[Constructible Field Ensemble]
We say a field ensemble $\FF = (\FF_n)_{n \in \NN}$ is \emph{constructible} if every element in $\FF_n$ has a
$O(\log{\abs{\FF_n}})$-bit representation,
and addition, multiplication, and inverses can be computed in $\polylog(\abs{\FF_n})$ time given the representations.  
\end{definition}
It is well known that for every $S = S(n)$,
constructible field ensembles $\FF = (\FF_n)$ with $\abs{\FF_n} = \Theta(S)$ exist.
Moreover,
we make the (mild) assumption that addition, multiplication and inverses can be computed in $\tO(\log \abs{\FF_n})$ time,
where $\tO$ omits $\poly\log\log(\abs{\FF_n})$ factors. 
This implies that degree-$d$ polynomials in $\FF[X]$ can be evaluated in $\tO(d\log\abs{\FF_n})$ time (with $\tO$ omitting $\polylog(d, \log \abs{\FF_n})$ factors)\footnote{Both these assumptions can also be relaxed by further delegating operations over large fields. For simplicity, we do not deal with that here.}. These properties are satisfied in fields that support FFT. (See table $8.6$ in \cite{vzGG13})

We utilize the Schwartz-Zippel lemma,
a basic fact about low-degree polynomials over finite fields.
\begin{lemma}[Schwartz-Zippel Lemma, \cite{Zippel79,Schwartz80}]
    Let $\FF$ be a field,
    let  $m \in \NN$, and let $f \in \FF[X_1,\ldots,X_m]$ be a non-zero polynomial of total degree $d$. 
    Then 
    
    \[\Pr_{\bm r\gets \FF^m}[f(\bm r) = 0] \le \frac{d}{\abs{\FF}}.\]
    \label{lem:SZ}
\end{lemma}

Given a subset $H \subset \FF$
and an integer $\lgK \in \NN$,
the \emph{low-degree extension} (LDE) of a function $f : H^\lgK \to \FF$ is the unique $(\abs{H} - 1)$-individual-degree polynomial $\hat {f}: \FF^{\lgK} \to \FF$ such that for all $\bm s \in H^\lgK$, $\hat {f}(\bm s) = f(\bm s)$.
In this work we focus on the special case when $H = \{0,1\}$,
where $\hat {f}$ is multilinear over $\FF^\lgK$,
and is called the multilinear extension of $f$.
Abusing the notation, 
given a string $x \in \FF^{2^\lgK}$,
$\hat{x}: \FF^\lgK \to \FF$ denotes the multilinear extension of the function $f_{x} : \{0,1\}^\lgK \to \FF$ whose function table is $x$.
On a sequence $\bm \pvalU = (\bm \pvalU_1, \ldots, \bm \pvalU_\pvalT)\in (\FF^\lgK)^{\pvalT}$ of length $\pvalT$,
we use the shorthand $\hat x(\bm \pvalU) \coloneqq (\hat x(\bm \pvalU_1),\ldots,\hat x(\bm \pvalU_\pvalT)) \in \FF^\pvalT$.

Consider the following \emph{Polynomial Valuation} (\pval) set, which is an affine subspace over $\FF$, first defined in~\cite{STOC:RotVadWig13}.

\begin{definition}[The \pval set]
\label{def:pval}
    Let $m, \pvalT \in \NN$.
    The set $\pvalF(\bm \pvalU, \bm \pvalv) \subset \FF^{2^\lgK}$ is parameterized by the sequences 
    $\bm \pvalU = (\pvalU_1,\ldots,\pvalU_\pvalT) \in (\FF^m)^\pvalT$ and $\bm \pvalv = (\pvalv_1,\ldots,\pvalv_\pvalT) \in \FF^\pvalT$. 
    It consists of all strings $x \in \FF^{2^m}$ whose corresponding multilinear extension $\hat x: \FF^m \to \FF$ satisfies $\hat x(\bm \pvalU) = \bm \pvalv$.
\end{definition}
Note that $\pval(\bm \pvalU, \bm 0)$ is a linear subspace of $\FF^{2^\lgK}$,
so we can define $\Delta(\pval(\bm \pvalU, \bm 0))$ to be the minimum Hamming distance of a non-zero vector in $\pval(\bm \pvalU, \bm 0)$, 
and the following proposition holds.
\begin{proposition}[Random \pval kernel has large distance]
    Let $\secpar, d, \lgK \in \NN$.
    Suppose $\FF$ is a field of characteristic $2$,
    $\abs{\FF} \ge 2 \cdot \lgK$,
    and $T \ge 2d(\lgK + \Flog) + \secpar$.
    Then,
    \[\Pr_{\bm \pvalU  \gets (\FF^\lgK)^T}[\Delta(\pval(\bm \pvalU, \bm 0)) \le 2d]\leq 2^{-\secpar}.\]
    \label{lem:pval}
\end{proposition}
\begin{proof}
    Fix a non-zero $\bm x \in \FF^{2^\lgK}$
    with Hamming weight at most $2d$.
    Since $\hat{\bm x}$  
    is a multilinear polynomial with total-degree $\lgK$ in $\FF[X_1,\ldots,X_\lgK]$,
    by \Cref{lem:SZ},
    the probability that $\hat x(\bm \pvalU_i) = 0$ is at most $\lgK/\abs{\FF}$.
    As $\pvalU$ contains $T$ independent samples,
    all of these are zero with probability at most $(\lgK/\abs{\FF})^{T}$.
    Therefore, $\bm x \in \pval(\bm \pvalU, \bm 0)$ with probability at most $(\lgK/\abs{\FF})^{T}$.

    There are at most $\binom{2^\lgK}{2d}\abs{\FF}^{2d}$ choices for $\bm x \in \FF^{2^\lgK}$ with Hamming weight at most $2d$,
    so by a union-bound and the assumption $T \geq 2d(\lgK + \Flog) + \secpar$,
    the probability that some of them are in $\pval(\bm \pvalU, \bm 0)$ is at most 
    \[
        \binom{2^\lgK}{2d}\abs{\FF}^{2d}\left(\frac{\lgK}{\abs{\FF}}\right)^{T} \le 2^{2dm} \cdot \abs{\FF}^{2d} \cdot 2^{-\pvalT}\le 2^{-\secpar}.\qedhere
    \]
\end{proof}
\subsection{Succinct Descriptions of Sets and Predicates}
\label{sec:desc}
\begin{definition}[Uniform Arithmetic Circuits]
    Let $\FF$ be a field.
    Let $C = \set{C_n : \FF^n \to \FF^m}_{n \in \NN}$ be a family of arithmetic circuits,  
    consisting of fan-in 2 \ADD and \MULT gates over a field $\FF$.
    For any $f = f(n)$, we say that $C$ is $f$-space uniform if there exists a fixed $O(f(n))$-space Turing machine $\cM$ that on input $1^n$ outputs the full description of the circuit $C_n$.
    When $n$ is clear from the context, we omit the subscript $n$ and write $C$ instead of $C_n$.
\end{definition}
An important special case is when $f(n) = \log(n)$, 
in which case we say that $C$ is log-space uniform.
\begin{remark}
    Any $f$-space uniform Boolean circuit $C$ can be trivially extended to by a $f$-space uniform arithmetic circuit $C'$ of the same size and depth over any field $\FF$.
\end{remark}
The following \emph{succinct descriptions} of sets can be used to recover the entire set.
\begin{definition}[Descriptions of Sets]
    A bit string $\lrag \cS \in \bin^B$ is a description of a set $\cS = \set{s_1, \ldots, s_k} \subset \bin^p$
    if there exists a (multi-output) $p$-space uniform
    $G: [k]\times\bin^{B} \to \bin^{p}$ of fan-in 2,
    such that $G(i, \lrag \cS) = s_i$ for all $i \in [k]$.
    The description is succinct if $\abs{\lrag \cS} = B < k \cdot p$.
    \label{def:setdesc}
\end{definition}

Similarly,
\emph{succinct descriptions} of functions can be used to configure a uniform circuit family to implement the function.
\begin{definition}[Descriptions of Functions]
    Let $\FF$ be a field.
    We say that $\lrag \C \in \bin^B$ is a description of a function $\C :\FF^n \to \FF$ if there exists a log-space uniform circuit $C:\FF^{n + B} \to \FF$ of fan-in 2 such that $C(x, \lrag{\C}) = \C(x)$ for all $x \in \FF^n$.
    The description is succinct if $\abs{\lrag \C} = B < \size(\C)$.
    \label{def:circdesc}
\end{definition}
The Turing machines that generate the uniform $G$ and $C$ in \Cref{def:setdesc,def:circdesc} take in their ``shape parameters'' $(k, 1^B, 1^p)$ and $(1^n, 1^B)$ as input, respectively.

\subsection{Unambiguous Interactive Proof (\UIP)}\label{sec:prelim:UIP}
An $(\ell, a, b, \Ptime, \Vtime, \Sigma)$-protocol is a \emph{public-coin,
$\ell$-round interactive protocol, with alphabet  $\Sigma$,
per-round prover message length $a$ and per-round verifier message length $b$, and prover runtime $\Ptime$ and verifier runtime $\Vtime$}.
Specifically, such a protocol consists of a pair of interacting Turing machines $\prot{}$ each having as input a string $x$ of length $n = \abs{x}$.
The machines may also take in other parameters as additional input,
such as parameters controlling the soundness error,\footnote{In particular, the prover might have some extra information that helps make it more efficient, for example the $\UP$ witnesses in case the language was in $\UP$.}
but we omit them from the notation for simplicity.
The machine $\cP$ is deterministic and runs in time $\Ptime$, and is called \emph{the prover}, while $\cV$ is probabilistic, runs in time $\Vtime$, and is called \emph{the verifier}.
$\cP$ and $\cV$ are the \emph{two parties} of the protocol.
We omit the specification of $\Sigma$ when $\Sigma = \bin$.

In each round $j \in [\ell]$ of the protocol:
\begin{enumerate}
    \item $\cV$ sends a random message $q_j \gets_R \Sigma^b$ to $\cP$, referred to as \emph{queries}.
    \item $\cP$ responds with a message $a_j \in \Sigma^a$ determined by the prescribed next-message function, referred to as \emph{answers}, which (abusing notation) is denoted by 
    \[a_j \coloneqq \cP(x, j, (q_1, \ldots, q_j)) \in \Sigma^a.\]
\end{enumerate}
We make the simplifying assumption that $\cV$ never rejects in the middle of an execution.  We abuse notation and denote the verdict circuit of $\cV$, that reads the entire transcript and decides if to accept or reject, also by $\cV$: 
\[\cV(x, (q_1, \ldots, q_\ell), (a_1, \ldots,a_\ell)) \in \bin.\]
We denote the sequence of verifier random coins by $\bm q \coloneqq (q_1,\ldots, q_\ell)$, 
and for $j \in [\ell]$ we denote by 
\[\bm q_{\le j} \coloneqq (q_1,\ldots, q_j)~~\mbox{ and }~~ 
\bm q_{> j} \coloneqq (q_{j+1},\ldots, q_\ell).
\]
Finally,
let $\cP(x, \bm q) \coloneqq (a_1, \ldots, a_\ell)$,
where $a_j = \cP(x, j, \bm q_{\le j})$ are the prescribed messages.

The \emph{total communication complexity} of the protocol is the number of bits exchanged between the prover and the verifier, i.e. $(a + b)\ell\cdot\log(\abs{\Sigma})$.

Let $\cL$ be a language.  We next define the notion of an 
unambiguous interactive proof of $\cL$.

\begin{definition}[$\epsilon$-unambiguous ($\ell$, $a$, $b$, $\Ptime$, $\Vtime$, $\Sigma$)-\UIP]
\label{def:UIP}
An $(\ell, a, b, \Ptime, \Vtime, \Sigma)$-protocol $\prot{}$ is an \emph{unambiguous interactive proof} (\UIP) with an unambiguity error $\epsilon$
for a language $\cL$, 
if it satisfies the following completeness and unambiguity conditions.
\begin{itemize}[leftmargin=*,label=-]
    \item \textbf{Prescribed Completeness:}
    For any $x \in \bin^n$ and
    any verifier coins $\bm q = (q_1,\ldots, q_\ell)\in(\Sigma^b)^\ell$,
    the answers $\answer = \cP(x, \bm q) \in(\Sigma^a)^\ell$ satisfy that $\cV(x, \bm q, \answer) = 1$ iff $x \in \cL$.
    \begin{remark}\label{remark:prescribed-comp}
    Completeness is usually defined only for $x \in \cL$. \emph{Prescribed completeness} also requires the prescribed prover to send rejected messages on $x \notin \cL$. We assume that when $x \notin L$, $\cP$ sends the default string $\bm 0 \in \Sigma^a$ in every round, and $\cV$ rejects this transcript.
    \end{remark}
    
    \item \textbf{$\epsilon$-Unambiguity}:
    For any $x\in\bin^n$, any $j^*\in[\ell]$,  any prefix $\bm q_{\le j^*} = (q_1,\ldots,q_{j^*})$, 
    and any (computationally unbounded) prover strategy $\cP^*$ 
    that deviates from $\cP$ firstly in round $j^* \in [\ell]$, namely
    \ifFOCS
    \begin{align*}
    \cP^*(x, j^*, \bm q_{\le j^*}) &\neq \cP(x, j^*, \bm q_{\le j^*}),\\
    \text{ and }\cP^*(x, j, \bm q_{\le j}) &= \cP(x, j, \bm q_{\le j}) \text{ for all }j < j^*,
    \end{align*}
    \else
    \[\cP^*(x, j^*, \bm q_{\le j^*}) \neq \cP(x, j^*, \bm q_{\le j^*})~~\mbox{ and }~~\cP^*(x, j, \bm q_{\le j}) = \cP(x, j, \bm q_{\le j}) \text{ for all }j < j^*,\]
    \fi
    we have that $\cP^*$ is accepted with probability at most $\epsilon = \epsilon(n)$ over the remaining random coins of \cV.
    Namely,
    \[
    \Pr[\cV(x, (\bm q_{\le j^*}, \bm q_{> j^*}), \answer) = 1] \le \epsilon,
    \]
    where the probability is taken over $\bm q_{>j^*}$ $=(q_{j^* + 1}, \ldots, q_\ell)$, and where $\answer = \cP^*(x, (\bm q_{\le j^*}, \bm q_{> j^*}))$.
\end{itemize}
\end{definition}
Note that by \Cref{remark:prescribed-comp}, unambiguity implies standard soundness.

\subsubsection{The \GKR protocol}
We restate the main result from \cite{JACM:GolKalRot15} (in the language of \pval)\footnote{In what follows we formulate the \GKR protocol as a reduction (as opposed to a proof system). 
Specifically, at the end of the protocol the verifier either rejects or outputs a $\pval$ claim.}.
We note that unambiguity is proven in \cite{STOC:JKKZ21}.
\begin{theorem}[The \GKR protocol]
    \label{lem:GKR}
    Let $\FF$ be a field,
    and let $\C: \FF^n \to \FF$ be an arithmetic circuit with addition and multiplication gates of fan-in 2 over $\FF$,
    with description $\lrag{\C}$.
    Let $G$ be the log-space uniform circuit that outputs $\C(x)$ on input $x \in \FF^n$ and $\lrag{\C} \in \bin^{\abs{\lrag{\C}}}$.
    Denote the depth and size of $G$ by $D = D(n) \ge \log n$, 
    and $S = S(n) \ge n$.

    There exists an $(\ell, a, b, \Ptime, \Vtime)$-protocol $(\cP_{\GKR},\cV_{\GKR})$
    such that the following holds.
    The prover and verifier get the description $\lrag \C$,
    and the prover additionally gets the input $x \in \FF^n$.
    Let $\lgK = \ceil{\log n}$.
    At the end,
    either $\cV_\GKR$ rejects,
    or both parties output a point $\pvalU \in \FF^\lgK$ and a value $\pvalv \in \FF$, such that:
    \begin{itemize}[label=-]
        \item \textbf{Prescribed Completeness:} If both parties follow the protocol then $\hat x(\pvalU) = \pvalv$ iff $\C(x) = 1$.
        \item \textbf{$\epsilon_{\GKR}$-Unambiguity:} For any (unbounded) cheating prover strategy $\cP^*$ that deviates from $\cP_\GKR$ firstly in round $j^* \in [\ell]$, 
        \ifFOCS
        \begin{align*}
        &\,{}\epsilon_{\GKR}(D, \log S, \abs{\FF}) \\
        &\coloneqq \Pr[\cV_\GKR\text{ accepts}~~\wedge~~ \hat x(\pvalU) = \pvalv] \\
        &= O\left(\frac{D \log S}{\abs{\FF}}\right) 
        \end{align*}
        \else
        \[
        \epsilon_{\GKR}(D, \log S, \abs{\FF}) \coloneqq \Pr[\cV_\GKR\text{ accepts}~~\wedge~~ \hat x(\pvalU) = \pvalv] \\
        = O\left(\frac{D \log S}{\abs{\FF}}\right)
        \]
        \fi
        where the probability is taken over the verifier's remaining coins (after round $j^*$).
        In other words,
        with probability at least $1-\epsilon_\GKR(D, \log S, \abs{\FF})$,
        either $\cV_{\GKR}$ rejects
        or $x \notin \pval(\pvalU, \pvalv)$.
        \item Regardless of the prover's strategy, $\pvalU$ is uniformly distributed in $\FF^\lgK$, 
        and only depends on the random coins of $\cV_\GKR$.
    \end{itemize}
    The complexity of the protocol is as follows (with $\tO$ ignoring polylogarithmic factors in $D, \log S$):
    \begin{itemize}
        \item $\ell = O(D \cdot \log S)$. 
        \item $a = O(\Flog)$.
        \item $b = O(\Flog)$.
        \item $\Ptime = \tO(\poly(S)\cdot \Fbits)$.
        \item $\Vtime = \tO(D \log S\cdot \Flog  + \abs{\lrag{\C}}\cdot\Flog)$ 
        (and $\cV$ does not access $x$). 
    \end{itemize}
    Furthermore,
    the verifier's verdict (on whether to output the strings $\pvalU, \bm \pvalv$ or reject at the end) has the following complexities (with $\tO$ ignoring polylogarithmic factors in $D, \log S, \log \abs{\FF}$):
    \begin{itemize}
        \item $\size(\cV) = \tO(D \log S\cdot \Flog + \abs{\lrag{\C}} \cdot \Flog)$. 
        \item $\depth(\cV) = \tO(1)$.
    \end{itemize}
\end{theorem}
\begin{remark}[The role of the circuit description ${\lrag{\C}}$]
    After running the original \GKR protocol from \cite{JACM:GolKalRot15} to the circuit $G : \FF^{n + \abs{\lrag{\C}}} \to \FF$ and input $(x, \lrag{\C})$ (which verifies the claim $\C(x) = 1$),
    the verifier obtains $\pvalU' \in \FF^{\ceil{\log(n + \abs{\lrag \C})}}$, 
    $\pvalv' \in \FF$ and has to verify the claim $\widehat{x \Vert \lrag{\C}}(\pvalU') = \pvalv'$ (where $\widehat{x \Vert \lrag{\C}}$ denotes the \LDE of the pair $(x, \lrag{\C})$ as a string under some appropriate encoding). 
    The protocol in \Cref{lem:GKR} additionally reduces this to verifying $\hat{x}(\pvalU) = \pvalv$ for some $\pvalU, \pvalv$.
    This is without additional communication as follows:
    using $\lrag{\C}$, the verifier first computes $v_0=\widehat{\bm 0^n \Vert \lrag{\C}}(\pvalU')$,
    and subtracts it from $\pvalv'$,
    which yields the claim $\widehat{x \Vert \bm 0^{\abs{\lrag \C}}}(\pvalU') = \pvalv' - v_0$.
    This is equivalent to checking that $\hat{x}(\pvalU) = \frac{\pvalv' - v_0}{\chi(\pvalU')} \in \FF$ where $\chi(\pvalU')$ is a factor that depends only on $\pvalU'$ and computable by the verifier on its own.
    The post-processing procedure amounts to the additional $O(\abs{\lrag \C} \cdot \Flog)$ verifier runtime.
\end{remark}
\paragraph{The \GKR protocol is $\Delta_c$-distance-preserving}
Let $\prot{\GKR}$ be the \GKR protocol from \Cref{lem:GKR}.
In \cite{STOC:RotVadWig13},
the authors observed that when $\prot{\GKR}$ is parallel-repeated $\pvalT$ times for a large enough $\pvalT$,
it is $\Delta_c$-distance-preserving in the following sense:
suppose the input $x$ is $d$-far from satisfying $\C$,
and $\bm \pvalU = (\pvalU_1,\ldots,\pvalU_\pvalT)$ and $\bm \pvalv = (\pvalv_1,\ldots,\pvalv_\pvalT)$ are the outputs of $\pvalT$ parallel repetitions of $\prot{\GKR}$,
then $x$ is $d$-far from the set $\pval(\bm \pvalU, \bm \pvalv)$.
Note that this increases the overall cost of the protocol by a factor of $\pvalT$.

This observation generalizes to $\Delta_c$-distance,
captured in the following lemma.
\begin{lemma}[\GKR is $\Delta_c$-Distance-Preserving]
    \label{lem:RVW}
    Let $\K, \ncol$, $d\in \NN$ and assume $\K = 2^\lgK$ for some $\lgK \in \NN$.
    Let $\FF$ be a field of characteristic 2.
    Let $\C: \FF^{\K \times \ncol} \to \bin$ be a log-space uniform arithmetic circuit over $\FF$,
    with addition and multiplication gates of fan-in 2,
    of size $S$ and depth $D$,
    and with description $\lrag{\C}$.

    Let $\transcMat \in \FF^{\K \times \ncol}$.
    Denote by $\cS\subset \FF^{\K \times \ncol}$ the set of matrices accepted by $\C$, 
    and suppose that
    \[
    \abs{\cS \cap \rowball(\transcMat)} \le 1,
    \]
    i.e. at most one element in $\rowball(\transcMat)$ satisfies the circuit. (See \Cref{def:Deltac-dist} for the definition of $\rowball$.)

    When applying $\prot{\GKR}$ to the claim $\C(\transcMat) = 1$\footnotemark{} with $\pvalT \in \NN$ parallel repetitions,
    we obtain $\pvalU \in (\FF^{\lgK + \log \ncol})^\pvalT, \bm \pvalv \in \FF^\pvalT$, with the following guarantee (in addition to the ones stated in \Cref{lem:GKR}, with all complexities but the round complexity multiplied by $\pvalT$).
    \footnotetext{Recall that 
    both $\cP_\GKR$ and $\cV_\GKR$ get the description $\lrag{\C}$. 
    Additionally, $\cP_\GKR$ gets the input $\transcMat$.} 

    \paragraph{Unambiguous Distance Preservation:} 
    For every prover $\cP^*$,
    the distance $\Delta_c(\transcMat, \cS)$ is preserved by the interaction $(\cP^*, \cV_\GKR)$ in the following sense:  
    \begin{itemize}[label=*]
        \item If $\Delta_c(\transcMat, \cS) \le d$ (i.e. $\cS \cap \rowball(\transcMat) \neq \varnothing$)
        and $\cP^*$ answers according to the prescribed strategy $\cP_\GKR$  corresponding to the unique $\transcMat^*\in \cS \cap \rowball(\transcMat)$,
        then the output $(\pvalU, \bm \pvalv)$ satisfies
        \ifFOCS
        \begin{align*}
            &\Pr\begin{bmatrix}
            \begin{aligned}
            &\Delta_c(\transcMat, \pval(\pvalU, \bm \pvalv)) \le d \text{ and }\\
            &\transcMat^* \text{ remains the unique element}\\
            &\text{in } \rowball(\transcMat) \cap \pval(\pvalU, \bm \pvalv) \\
            \end{aligned}
            \end{bmatrix}\\
            &\ge 1 - (\epsilon_{\GKR}(D, \log S, \abs{\FF}))^T \cdot\left(\binom{\K}{d} \abs{\FF}^d\right)^\ncol.
        \end{align*}
        \else
        \begin{align*}
            &\Pr[\Delta_c(\transcMat, \pval(\pvalU, \bm \pvalv)) \le d \text{ and } \transcMat^* \text{ remains the unique element in } \rowball(\transcMat) \cap \pval(\pvalU, \bm \pvalv)] \\
            &\ge 1 - (\epsilon_{\GKR}(D, \log S, \abs{\FF}))^T \cdot\left(\binom{\K}{d} \abs{\FF}^d\right)^\ncol.
        \end{align*}
        \fi
        \item On the other hand, if $\Delta_c(\transcMat, \cS) > d$ (i.e. $\cS \cap \rowball(\transcMat) = \varnothing$) or $\cP^*$ does not answer according to $\cP_\GKR$ corresponding to the unique $\transcMat^*\in \rowball(\transcMat)$, 
        then the output $(\pvalU, \bm \pvalv)$ satisfies
        \ifFOCS
        \begin{align*}
            &\Pr[\Delta_c(\transcMat, \pval(\pvalU, \bm \pvalv)) > d] \\
            &\ge 1 - \left( \vphantom{\binom{\K}{d}}\epsilon_{\GKR}(D, \log S, \abs{\FF})\right.\\
            &+ \left. (\epsilon_{\GKR}(D, \log S, \abs{\FF}))^T \cdot\left(\binom{\K}{d} \abs{\FF}^d\right)^\ncol\right).
        \end{align*}
        \else
        \begin{align*}
            &\Pr[\Delta_c(\transcMat, \pval(\pvalU, \bm \pvalv)) > d] \\
            &\ge 1 - \left(\epsilon_{\GKR}(D, \log S, \abs{\FF}) + (\epsilon_{\GKR}(D, \log S, \abs{\FF}))^T \cdot\left(\binom{\K}{d} \abs{\FF}^d\right)^\ncol\right).
        \end{align*}
        \fi
    \end{itemize}
    Here $\epsilon_{\GKR}(D, \log S, \abs{\FF}) = O(\frac{D\log S}{\abs{\FF}})$ is the soundness error of one \GKR protocol.
\end{lemma}
\begin{proof}
    Fix $\transcMat' \in \rowball(\transcMat)$, and consider the following two cases:
    \begin{enumerate}[label = Case \arabic*, leftmargin=6em]
        \item $\C(\transcMat') = 0$.
            
            The claim $\C(\transcMat') = 1$ is false.
            For any prover strategy $\cP^*$,
            by \GKR's standard soundness amplification,
            $\Pr[\transcMat' \in \pval(\pvalU, \bm \pvalv)] \le \left(\eps_{\GKR}(D, \log S, \abs{\FF})\right)^\pvalT$.
            \label{apd:lem:obs1}
        \item $\C(\transcMat') = 1$.
        
            If the prover $\cP^*$ deviates from the prescribed prover $\cP$ with respect to $\transcMat^*$ in some round $j^*$, 
            then it must be deviating in one of the $\pvalT$ parallel executions of \GKR.
            By the unambiguity of \GKR \footnote{Note that $\cP^*$ does not necessarily deviate in all $\pvalT$ executions, so the error term is not necessarily amplified.},
            $\Pr[\transcMat' \in \pval(\pvalU, \bm \pvalv)] \le \eps_{\GKR}(D, \log S, \abs{\FF})$.
            \label{apd:lem:obs2}
    \end{enumerate}
    Crucially, note that there are at most $(\binom{\K}{d} \abs{\FF}^d)^\ncol$ points $\transcMat'$ in $\rowball(\transcMat)$.
    Let $E$ denote the event that $\exists \transcMat' \in \rowball(\transcMat) \text{ s.t. } \C(\transcMat') = 0 \text{ and } \transcMat' \in \pval(\pvalU, \bm \pvalv)$.
    By \ref{apd:lem:obs1} and a union bound,
    \begin{align*}
    \Pr[E] \le \epsilon' \coloneqq \eps_{\GKR}(D, \log S, \abs{\FF})^\pvalT \cdot \left(\binom{\K}{d} \abs{\FF}^d\right)^\ncol.
    \end{align*}
    Using this bound on $E$, we have
    \begin{itemize}[label=*]
        \item When $\cS \cap \rowball(\transcMat) = \set{\transcMat^*}$ and $\cP^*$ answers according to the prescribed strategy $\cP_\GKR$ corresponding to $\transcMat^*$:

        In this case, 
        any other $\transcMat' \in \rowball(\transcMat)$ satisfies $\C(\transcMat') = 0$.
        Additionally, by the prescribed completeness of the \GKR protocol,
        $\transcMat^* \in \pval(\pvalU, \bm \pvalv)$.
        Therefore,
        \begin{align*}
            &\Pr\begin{bmatrix}
                \begin{aligned}
                &\Delta_c(\transcMat, \pval(\pvalU, \bm \pvalv)) \le d \text{ and } \\
                &\transcMat^* \text{ remains the unique element}\\
                &\text{in } \rowball(\transcMat) \cap \pval(\pvalU, \bm \pvalv)
            \end{aligned}
            \end{bmatrix}\\
            &\ge \Pr\left[\transcMat \in \pval(\pvalU, \bm \pvalv) \text{ and } \overline{E}\right]\\
            &= 1- \epsilon'.
        \end{align*}
        \item When $\cS \cap \rowball(\transcMat) = \set{\transcMat^*}$ but $\cP^*$ does not answer according to $\cP_\GKR$ corresponding to $\transcMat^*$:
        
        By \ref{apd:lem:obs2},
        the probability that $\transcMat^* \in \pval(\pvalU, \bm \pvalv)$ is at most $\eps_{\GKR}(D, \log S, \abs{\FF})$.
        Therefore,
        \begin{align*}
        &\Pr[\rowball(\transcMat) \cap \pval(\pvalU, \bm \pvalv) = \varnothing] \\
        &\ge 1 - (\Pr[\transcMat^* \in \pval(\pvalU, \bm \pvalv)] + \Pr[E])\\
        &\ge 1 - \left(\eps_{\GKR}(D, \log S, \abs{\FF}) + \epsilon'\right). 
        \end{align*}

        \item When $\cS \cap \rowball(\transcMat) = \varnothing$,
        no matter what $\cP^*$ does,
        \begin{align*}
        \Pr[\rowball(\transcMat) \cap \pval(\pvalU, \bm \pvalv) = \varnothing] 
        &\ge 1 - \Pr[E]\\
        &\ge 1 - \epsilon'.  \qedhere
        \end{align*}
    \end{itemize}
\end{proof}

\subsection{Almost \texorpdfstring{$d$}{d}-wise independent permutations} \label{subsec: affine perms}
We begin by defining \textit{almost $d$-wise independent permutations}:
\begin{definition}[Almost $d$-wise Independent Permutations] \label{def: almost independent permutations}
    A distribution $\Pi$ on $S_N$\footnote{The symmetric group on $N$ elements i.e. the set of all permutations} is $\eta$-almost $d$-wise independent if for all $1\le x_1 <x_2<\ldots <x_d\le N$, the distribution $(\pi(x_1),\ldots,\pi(x_d))_{\pi\sim \Pi}$ is at most $\eta$ far from the uniform distribution on distinct $d$-tuples over $[N]$ in $\mathsf{TV}$ distance.
\end{definition}

We will be using a family of almost $d-$wise independent permutations on $\{0,1\}^\lgK$ in \Cref{sec:DcRR}, which has a \textit{pairwise independence} property and has succinct descriptions. The properties we will desire from our permutation family are:
\begin{itemize}
    \item Almost $d$-wise independence.
    \item $\pi\sim \Pi$ have succinct descriptions, short seeds, and can be efficiently implemented.
    \item $\pi^{-1}$ for $\pi\sim \Pi$ have succinct descriptions, short seeds, and can be efficiently implemented. 
\end{itemize}

The question of sampling from almost $d$-wise independent permutations has been well studied \cite{Gowers_1996,ICALP:HMMR04,BH08,KNR09,C:GHKO25,SODA:GreHePel25}. From this line of work, we get the following result:

\begin{theorem}[Theorem $2$ in \cite{SODA:GreHePel25}] \label{thm: random reversible circuits are independent}
    For any $\lgK$, and $d\le 2^{\lgK/50}$, a random reversible circuit with $\tilde{O}(d\lgK\cdot \log(1/\eta))$ width-2 gates computes an $\eta$-almost $d$-wise independent permutation where $\tilde{O}$ hides $\polylog(d,\lgK)$ factors.
\end{theorem}

Observe that since the sampled circuit is a random reversible circuit, it is always a permutation, the inverse is easy to compute, and the seed is just the set of gates which can also be described with $\tilde{O}(nk\cdot \log(1/\eta))$ bits and the permutation is also efficient to compute.

\begin{definition}[$\Pi_{\lgK, d, \eta}$]
    We let $\Pi_{\lgK, d, \eta}$ be the distribution from \Cref{thm: random reversible circuits are independent}. 
\end{definition}

\begin{definition}[Uniform distributions]
    Let $\mathcal U_{\lgK, d}$ be the uniform distribution of distinct $d-$tuples over $[2^\lgK]$. Let $\mathcal U'_{\lgK, d}$ be the uniform distribution of $d-$tuples over $[2^\lgK]$.
\end{definition}

\begin{theorem} \label{thm: random set intersection is tiny}
    Let $A_0, A_1 \subseteq \{0,1\}^{\lgK}$ such that $|A_0|, |A_1|\le d$ such that $d\le 2^{\lgK/50}$. Then, \[\Pr_{\pi \sim \Pi_{\lgK, d, \eta}}[|A_0\cap \pi(A_1)|\ge \epsilon d]\le \eta+ \mathrm{exp}(-\epsilon d/6)\]
    if $\epsilon d \ge 2$.
\end{theorem}
\begin{proof}
    Observe that since the indicator random variable for $|A_0\cap \pi(A_1)|\ge \epsilon d$ is Bernoulli, and the TV distance, we have that 
    \ifFOCS
    \begin{align*}
    &\Pr_{\pi \sim \Pi_{\lgK, d, \eta}}[|A_0\cap \pi(A_1)|\ge \epsilon d]\\
    &\le \eta+ \Pr_{\pi \sim \mathcal U_{\lgK, d}}[|A_0\cap \pi(A_1)|\ge 1+\epsilon d].
    \end{align*}
    \else
    \[\Pr_{\pi \sim \Pi_{\lgK, d, \eta}}[|A_0\cap \pi(A_1)|\ge \epsilon d]\le \eta+ \Pr_{\pi \sim \mathcal U_{\lgK, d}}[|A_0\cap \pi(A_1)|\ge 1+\epsilon d].\]
    \fi

    For any $a\in A_1$, let $X_a, X_a'$ be the indicator random variables for $\pi(a)\in A_0$ where $\pi \sim \mathcal U_{\lgK, d}$ and $\mathcal U_{\lgK, d}'$ respectively. Additionally, let $X=\sum X_a$ and $X'=\sum X_a'$. Observe that for any $A\subset A_0$, 
    \ifFOCS
    \begin{align*}
    \EE[\Pi_{a\in A} X_a] &= \Pr[\forall a\in A, X_a=1] \\
    &= \prod_{j=0}^{|A|-1} \frac{|A_0|-j}{2^{\lgK}-j} \\
    &\le \left(\frac{|A_0|}{2^{\lgK}}\right)^{|A|}\le \EE[\Pi_{a\in A} X_a'].
    \end{align*}
    \else
    \[\EE[\Pi_{a\in A} X_a] = \Pr[\forall a\in A, X_a=1] = \prod_{j=0}^{|A|-1} \frac{|A_0|-j}{2^{\lgK}-j} \le \left(\frac{|A_0|}{2^{\lgK}}\right)^{|A|}\le \EE[\Pi_{a\in A} X_a'].\] 
    \fi

    Thus, for any integral $c\ge 0$, we have that $\EE[X^c]\ge \EE[X'^c]$. Now, since we get Chernoff bound by applying Markov Inequality on the moment generating function, we must have that \[\Pr[X\ge \mu(1+\frac{\epsilon d}{\mu}-1)]\le e^{(1-\epsilon d)\cdot \frac{\frac{\epsilon d}{\mu}-1}{1+\frac{\epsilon d}{\mu}}}\le e^{-\epsilon d/6}\] where $\mu = \EE[X]=\frac{|A_0|\cdot |A_1|}{2^{\lgK}}\le \frac{d^2}{2^{\lgK}}\le 1$.
\end{proof}

\subsection{Unambiguous Interactive Proofs of Proximity (\IPP)}

Defined in \cite{EKR04,STOC:RotVadWig13}, Interactive Proofs of Proximity ($\IPP$) for the \pval problem are a key ingredient in several prior works (for example, \cite{STOC:RotVadWig13,RRR18,TCC:RotRot20}). 
They are also used in our protocols. 
Here, we define a slight variant of the standard IPP with the added unambiguity property.

\begin{definition}[\cite{EKR04,STOC:RotVadWig13} (Unambiguous) Interactive Proof of Proximity]
\label{def:IPP}
    Let $d \in \NN$.
    A pair of interacting Turing machines $(\cP, \cV)$, 
    where $\cP$ is deterministic and takes as input a string $x \in \bin^n$ and $\cV$ has oracle access to $x$,
    is an \emph{unambiguous interactive proof of proximity with distance $d$} ($d$-\IPP) for a language $\cL$, 
    under distance $\mathsf{dist}$, with perfect completeness and unambiguity error $\eps = \eps(n)$ if:
    \begin{itemize}
        \item \textbf{Prescribed Completeness:} For all $x \in \bin^n$, the verifier accepts in the interaction $(\cP(x), \cV^x(1^n))$ iff $\mathsf{dist}(x, \cL) < d$.
        \item \textbf{$\eps$-Unambiguity:} For all $x \in \bin^n$, and all provers $\cP^\ast$ that deviates from $\cP$ in some round $j^*$, 
        \[
        \Pr[\cV \text{ accepts the interaction with }\cP^*] \le \eps(n),
        \]
        where the probability is taken over the verifier's random coins after round $j^*$.
    \end{itemize}
    The query complexity $q(n)$ is the number of queries the verifier makes to $x$, the communication and round complexity of the protocol are as usual. 
\end{definition}
\section{Our \bUIP}
Let $\batchK \in \NN$.
For a language $\cL$,
let its batched version $\cL^{\otimes \batchK}$ be 
\ifFOCS
\begin{align*}
\cL^{\otimes \batchK} \coloneqq 
\begin{Bmatrix}
\begin{aligned}
&(x_1,\ldots, x_\batchK) : \\
&\forall i \in [\batchK], x_i \in \cL \text{ and } \abs{x_1} = \ldots = \abs{x_\batchK}
\end{aligned}
\end{Bmatrix}.
\end{align*}
\else
\[
\quad\cL^{\otimes \batchK} \coloneqq \set{(x_1,\ldots, x_\batchK) : \forall i \in [\batchK], x_i \in \cL \text{ and } \abs{x_1} = \ldots = \abs{x_\batchK}}.
\]
\fi
If $\cL$ has a \UIP (\Cref{def:UIP}), then $\cL^{\otimes k}$ has a $\UIP$ obtained by running the underlying $\UIP$ for $\cL$ for each of the $\batchK$ instances in parallel.
This introduces a factor of $\batchK$ overhead.
Our main result is a better batching protocol.
For simplicity, assume the alphabet is $\Sigma = \bin$.
\begin{theorem}[Formal version of \Cref{thm:informal-bUIP}]
    \label{thm:basic-bUIP}
    Let $\cL$ be a language with a (public-coin) $\epsilon$-unambiguous\\
    ($\ell$, $a$, $b$, $\Ptime$, $\Vtime$)-\UIP $\prot{}$,
    where the verifier verdict $\cV$ is a log-space uniform boolean circuit 
    with size $S$ and depth $D$. 
    We also make the simplifying assumptions that the verifier never rejects in the middle of the \UIP, 
    and that $b \le a$.
    
    Let $\bm x = (x_1,\ldots, x_\batchK) \in (\bin^n)^\batchK$ be a batch of statements in $\cL^{\otimes \batchK}$ to be verified, let $\secpar = \secpar(\batchK)\in\mathbb{N}$ be an unambiguity parameter.
    Then there exists an $\epsilon'$-unambiguous $(\ell', a', b', \Ptime', \Vtime', \bin)$-\UIP for the language $\cL^{\otimes \batchK}$, 
    denoted by $\protU{'}$ (\Cref{alg:basic-bUIP}),
    with the following properties, where $\tO$ hides $\poly(\secpar) \cdot\polylog(\batchK, n, a, \ell)$ factors:
    \begin{itemize}
        \item The unambiguity error is $\epsilon' = \tO(\epsilon + 2^{-\secpar})$.
        \item The round complexity is $\ell' = \tO(\ell \cdot D \log S)$.
        \item The prover and verifier message length per round is $a' = b' = \tO(a \ell + \poly(\ell))$.  
        \item The prover runtime is $\Ptime' = \tO(\batchK \cdot \ell \cdot \Ptime + \poly(\batchK, a, \ell, S))$.
        \item The verifier runtime is $\Vtime' = \tO(\ell \cdot \Vtime + a\ell^2(\batchK n + D\log S) + a^2 \cdot \poly(\ell))$.
    \end{itemize}
    Furthermore, 
    the verifier's verdict circuit is log-space uniform and has the following properties:
    \begin{itemize}
        \item $\size(\cV') = \tO(\ell\cdot S + a\ell^2(\batchK n + D\log S) + a^2\cdot \poly(\ell))$.
        \item
        $\depth(\cV') =  D + \tO(1)$.
    \end{itemize}
\end{theorem}
The unambiguity parameter $\secpar$ controls the unambiguity error $\epsilon'$,
which we show to be concretely upper-bounded by $2\log \batchK \cdot(\epsilon + 2^{-\secpar})$ in \Cref{sec:pf-bUIP}.
Increasing $\secpar$ imposes an overhead factor of $O(2^\secpar)$ in the prover's message length and the prover's and verifier's runtime.

If $D\log S = \polylog(\batchK, n)$, then \Cref{thm:informal-bUIP} follows by choosing $\secpar = \log \frac{1}{\epsilon}$. 
With $\tO$ hiding $\polylog(\secpar^{-1}, \batchK, n, a, \ell)$ factors,
\begin{itemize}
    \item $\epsilon' = O(\epsilon \log \batchK)$.\footnote{We require $\epsilon = O(1/\log\batchK)$ for the unambiguity error to be non-trivially small.}
    \item $\ell' = \tO(\ell)$,
    \item $a' = b' = \tO(a \ell)$,
    \item $\Ptime' = \tO(\batchK \cdot \ell \cdot \Ptime + \poly(\batchK, a, \ell, S))$,
    \item $\Vtime' = \tO(\ell \cdot \Vtime +\batchK n a\ell^2 + a^2 \cdot \poly(\ell))$. Note that \Cref{thm:informal-bUIP} uses the looser bound $\batchK n a \poly(\ell)$ for the second and third term.
\end{itemize}

\subsection{Our Ingredients}  
\label{sec:iter}

\paragraph{Our Construction}

Our \bUIP protocol consists of two main ingredients: a \emph{Distance Generation Protocol} and an  \emph{Instance Reduction Protocol}. 
It iteratively applies these two protocols, where in each iteration, 
the \emph{Distance Generation Protocol} $\prot{\Distance}$ is applied,
followed by the application of \emph{Instance Reduction Protocol} $\prot{\Reduce}$.

The Distance Generation Protocol (\Cref{clm:phase1}) ensures that if the $k$-fold input instance was not in the language,
then after running it,
the prover and verifier obtain a related instance that is far from the language.
While the Reduction Protocol (\Cref{clm:phase2}) has the guarantee that as long as the given $k$-fold input instance is far from the language,
it produces a significantly smaller $k'$-fold instance not in the language 
(but without any distance guarantees).
Therefore, as long as the overall shrinkage is significant, say larger than a factor of $2$,
the parties can repeat the two sub-protocols $\nround = O(\log \batchK)$ times,
and end up having roughly $\batchK \cdot 2^{-\nround} = \polylog(\batchK)$ instances that can be explicitly checked.

In the initial stage,
the verifier samples a random $\bm q \gets \bin^{b\ell}$,
and sends the entire $\bm q$ to the prover.
This $\bm q$, along with the statements $(x_1,\ldots, x_\batchK)$,
implicitly define a set of $\batchK$ ``prescribed answers'':
For each $i$, denote by $\transc^{x_i, \bm q} = \cP(x_i, \bm q) \in \bin^{a\ell}$ the answers the base \UIP's prescribed prover $\cP$ would send,
on input statement $x_i$ upon receiving $\bm q$ from the verifier.

We denote by $\C^{(0)}$ the boolean circuit that has $\bm q$ and $(x_1,\ldots,x_k)$ hardwired. It takes the input $\transcMat^{\bm q} \coloneqq \set{\transcMat^{x_i, \bm q}}_{i \in [\batchK]}$
and outputs $1$ iff $\cV(x_i, \bm q, \transcMat^{x_i, \bm q}) = 1$ for all $i \in [\batchK]$.
We treat $\transcMat^{\bf q}$ as a matrix in $\bin^{\batchK \times (a\ell)}$ whose $i$'th row is $\transc^{x_i, \bm q}$,
and call $\transcMat^{\bm q}$ the \emph{prescribed matrix} corresponding to $\bm q$.

Let $\cL'_{\bm x}$  be the \emph{associated language} that consists of all instances $(\cS, \C)$,
where 
\[\cS = \set{(i^1, \bm q^1), \ldots, (i^g, \bm q^g)}\]
is a set of $g$ pairs in $[\batchK] \times \bin^{b\ell}$,
such that the prescribed matrix $\transcMat^{\cS} \in \bin^{g \times (a\ell)}$, consisting of rows $\set{\transc^{x_{i^j}, \bm q^j}}_{j \in [g]}$ where $\transc^{x_{i^j}, \bm q^j}=\cP(x_{i^j},\bm q^j)$,  satisfies the constraint $\C$.
Note that the language $\cL'_{\bm x}$ is not necessarily in $\NP$,
but as we shall see in \Cref{sec:final-check},
there is a \UIP for explicitly checking it.
We mention that 
the prescribed matrix $\transcMat^{\cS}$ replaces the role of the implicit vector of unique witnesses in the case of \bUP.

It follows from the prescribed completeness of the base \UIP that for $\cS^{(0)} \coloneqq \set{(i, \bm q)}_{i \in [\batchK]}$, 
\[
(\cS^{(0)},\C^{(0)}) \in \cL'_{\bm x} ~\mbox{ iff }~\bm x \in \cL^{\otimes \batchK}.
\]
In each iteration,
the following two sub-protocols are run:
\begin{enumerate}
    \item The Distance Generation Protocol is applied to $(\cS^{(0)}, \C^{(0)}) \in \cL'_{\bm x}$ with distance parameter $d = \ell \cdot \polylog(\batchK, \ell) \in \NN$.  At the end of this protocol, the verifier obtains a set $\cSmid$ containing $\K = |\cS^{(0)}| \cdot \ell$ pairs, where initially $|\cS^{(0)}|=|\{(i,\bm q)\}_{i\in\batchK}|=\batchK$,
        as well as a predicate $\Cmid$ that is evaluated over the prescribed matrix $\transcMat^{\cSmid}$.
        The guarantee of our Distance Generation Protocol is that when $(\cS^{(0)}, \C^{(0)}) \notin \cL'_{\bm x}$,
        the prescribed matrix $\transcMat^{\cSmid}$ is $\Delta_c$-$d$-far (\Cref{def:Deltac-dist}) from satisfying the predicate $\Cmid$.

    \item Then, the Instance Reduction protocol with shrinkage parameter $d = \ell \cdot \polylog(\batchK, \ell) \in \NN$ transforms the claim $(\cSmid, \Cmid) \in \cL'_{\bm x}$ to a new claim $(\cS^{(1)}, \C^{(1)}) \in \cL'_{\bm x}$,
        where $\cS^{(1)}$ is a subset of $\cSmid$ with size $\K' \le \K/d$.
        The Instance Reduction Protocol guarantees that $(\cS^{(1)}, \C^{(1)}) \notin \cL'_{\bm x}$ if the prescribed matrix $\transcMat^{\cSmid}$ is $\Delta_c$-$d$-far (\Cref{def:Deltac-dist}) from satisfying $\Cmid$.
        Importantly, the verifier never needs to explicitly access the prescribed matrix $\ippinput$.
\end{enumerate}

We set the distance/shrinkage parameter to be $d > 2 \ell$ to offset the intermediate factor of $\ell$.
In other words,
\[\abs{\cSmid}\leq \ell\cdot\abs{\cS^{(0)}}~~\mbox{ and }~~\abs{\cS^{(1)}}\leq \frac{\abs{\cSmid}}{d} < \frac{\abs{\cS^{(0)}}}{2},
\]
and thus each iteration shrinks the number of instances at least by a factor of~$2$, which implies that the parties need to iterate at most $\nround = \floor{\log \batchK}$ times,
until they can check the final claim $(\cS^{(\nround)}, \C^{(\nround)}) \in \cL'_{\bm x}$,
defined over $\abs{\cS^{(\nround)}} \le \batchK/2^\nround$ pairs.
The final check can be done by explicitly running the \UIP for verifying $(\cS^{(\nround)}, \C^{(\nround)}) \in \cL'_{\bm x}$ (as formally proven in \Cref{lem:final-check} below).

\paragraph{The associated language $\cL'_{\bm x}$}
We formally set up the notations to define the language $\cL'_{\bm x}$.
\begin{definition}[Prescribed transcript relation $\cR$]
    \label{def:prescribed}
    Let $\cR$ be the binary relation containing pairs $((x, \bm q), \transc) \in (\bin^n \times \bin^{b\ell}) \times \bin^{a\ell}$, such that 
    \begin{itemize}
        \item $\cV(x, \bm q, \transc) = 1$, and 
        \item 
        $\transc$ is the \emph{prescribed prover transcript}, i.e. $\transc = \cP(x, \bm q) \in \bin^{a\ell}$ is the prescribed prover messages $(\answer_1^{x, \bm q},\ldots,\answer_\ell^{x, \bm q})$ sent by $\cP$ on input $x$ when interacting with $\cV$ that uses $\bm q$ as its random coins.\footnote{Recall that if $x \notin \cL$, $\transc^{x, \bm q} = \cP(x, \bm q)$ is the default all-zero string, and $\cV(x, \bm q, \transc^{x, \bm q}) = 0$.}
    \end{itemize}
    \textbf{Notations}
    (Prescribed prover messages $\transc^{x, \bm q}$, $\answer_j^{x, \bm q}$, and matrices $\transcMat^{\cS}, \answerMat^{\cS}_j$)
    For any $x \in \bin^n$ and any  $\bm q \in (\bin^{b})^\ell$,
    the shorthand $\transc^{x, \bm q} \in \bin^{a\ell}$ stands for the (unique) prescribed prover messages 
    corresponding to $(x, \bm q)$.  Namely, $\transc^{x, \bm q}:=\cP(x, \bm q)$.
    Similarly,
    for every $j \in [\ell]$,
    the shorthand $\answer_j^{x, \bm q}$ means the $j$-th round prover message in 
    $\transc^{x, \bm q} \coloneqq (\answer_1^{x, \bm q},\ldots,\answer_\ell^{x, \bm q})$.
    More generally, given a batch of statements $\bm x = (x_1,\ldots, x_\batchK) \in (\bin^n)^\batchK$,
    and a subset of $g$ pairs $\cS \subset [\batchK] \times (\bin^b)^\ell$,
    we denote by $\transcMat^{\cS} \coloneqq \left((\transc^{x_i, \bm q})_{(i, \bm q) \in \cS}\right)^\top$ the matrix whose rows are $\transc^{x_i, \bm q} \in \bin^{a\ell}$.
    For every $j \in [\ell]$,
    $\answerMat^{\cS}_j \coloneqq \left((\answer_j^{x_i, \bm q})_{(i, \bm q) \in \cS}\right)^\top$ is the matrix
    whose rows are $\answer_j^{x_i, \bm q} \in \bin^a$,
    the $j$-th round prescribed message in $\transc^{x_i, \bm q}$.
\end{definition}

Note that $\cR$ is not necessary a $\UP$ (or even $\NP$) relation,
but for every $x \in \cL$ and every $\bm q \in \bin^{b\ell}$,
there is a single $\transc$ that satisfies $((x, \bm q), \transc) \in \cR$.
As we shall see in \Cref{lem:Lxq},
when $x \in \cL$, 
there exists a \UIP for certifying that $\transc$ is indeed the (unique) $\transc^{x, \bm q}$ for which $((x, \bm q), \transc) \in \cR$.
Therefore, in this sense,
we can treat $\transc^{x, \bm q}$ as the unique witness for the relation $((x, \bm q), \transc) \in \cR$.

We now formally define the associated language $\cL'_{\bm x}$.
For $(\cS, \C) \in \cL'_{\bm x}$ to hold,
not only should all remaining $x_i$'s be in $\cL$,
a joint constraint $\C$ should also be satisfied on the prescribed matrix $\transcMat^{\cS}$ associated with $\cS$.
\begin{definition}[The associated language $\cL'_{\bm x}$]
    \label{def:int_claim}
    Let $\bm x \in (\bin^n)^\batchK$,
    $\cS \subset [\batchK] \times \bin^{b\ell}$ be a set of $g \coloneqq \abs{\cS}$ pairs,
    and $\C$ be a boolean circuit that takes $\transcMat^{\cS}$ as input.
    Let $\I \coloneqq \set{i: (i, \bm q) \in \cS}$.
    Then 
    \[(\cS, \C) \in \cL'_{\bm x}~~\mbox{ iff }~~ \C(\transcMat^{\cS}) = 1~~\wedge~~\bm x |_{\I} \in \cL^{\otimes \abs{\I}}.
    \]
\end{definition}
    
\subsubsection{The Distance Generation Protocol}
We give the full specification of $\prot{\Distance}$ and its analysis in \Cref{sec:phase1},
which has the guarantees stated in \Cref{clm:phase1} below.

The protocol takes as input $\param{}$,
an unambiguity parameter $\secpar \in \NN$,
a distance parameter $d \in \NN$ and a field parameter $\FF$ with $\abs{\FF} \ge 2^\secpar \cdot \polylog(\batchK, n, a, \ell)$,
and outputs $\parammid{}$ such that $(\cSmid, \Cmid) \in \cL'_{\bm x}$ iff $(\cS, \C) \in \cL'_{\bm x}$.
It has a slightly increased instance size $\abs{\cSmid} = \abs{\cS} \cdot \ell$,
but has the guarantee that if $\C(\bm a^\cS)=0$ then $\ippinput$ is 
$\Delta_c$-$d$-far from satisfying the circuit $\Cmid$,\footnotemark
\footnotetext{Recall that $\ippinput$ is $\Delta_c$-$d$-far from satisfying $\Cmid$, if for every $\bm a^* \in \rowball(\ippinput)$, it holds that $\Cmid(\bm a^*) = 0$.} where $\bm a^\cS$ and $\ippinput$ are the prescribed matrices corresponding to $\cS$ and $\cSmid$, respectively.

\begin{lemma}[$\Delta_c$-Distance Generation Protocol]
    \label{clm:phase1}
    Let $\cL$ be a language with an $\epsilon$-unambiguous ($\ell$, $a$, $b$, $\Ptime$, $\Vtime$)-\UIP $\prot{}$.
    Let $S$ and $D$ be the size and depth of the decision circuit $\cV$.

    There exists a protocol \prot{\Distance} (\Cref{alg:phase1}) that takes as input $\bm x, \param{}$,
    where $\bm x \in (\bin^n)^{\batchK}$,
    $\cS \subset [\batchK] \times \bin^{b\ell}$,
    and $\C : \bin^{\abs{\cS} \times a \ell} \to \bin$,
    an unambiguity parameter $\secpar \in \NN$,
    a distance parameter $d \in \NN$,
    a field $\FF$ with $\abs{\FF} \ge 2^\secpar \cdot \polylog(\batchK, n, a, \ell)$,
    and outputs \itermidone{} such that $\abs{\cSmid} = |\cS| \cdot \ell$,
    and the following holds.
    \begin{itemize}[label=-]
        \item \textbf{Prescribed Completeness:}
        
        In an honest execution, $(\cSmid, \Cmid) \in \cL'_{\bm x}$ iff $(\cS, \C) \in \cL'_{\bm x}$.
        \item \textbf{$(\epsilon + 2^{-\secpar})$-Unambiguous Generation of $\Delta_c$-Distance:} 
        For every cheating prover $\cP^*$ that deviates from $\cP_\Distance$ first in round $j^*$,
        conditioned on the first $j^*$ rounds of the protocol,
        \ifFOCS
        \begin{align*}
        &\Pr\left[\exists \transcMat^* \in \rowball(\ippinput)\text{ s.t. } \Cmid(\transcMat^*) = 1\right] \\
        &\le \epsilon + 2^{-\secpar},
        \end{align*}
        \else
        \[
        \Pr\left[\exists \transcMat^* \in \rowball(\ippinput)\text{ s.t. } \Cmid(\transcMat^*) = 1\right] \le \epsilon + 2^{-\secpar},
        \]
        \fi
        where $\parammid{}$ denotes the output of the interaction between $\cV_\Distance$ and $\cP^*$, and where the probability is over the remaining coins of $\cV_\Distance$.

        Furthermore,
        for any (honest or cheating) prover $\cP^*$, 
        $\ippinput$ is likely the only input in $\rowball(\ippinput)$ that satisfies $\Cmid$.
        Formally, 
        \ifFOCS
        \begin{align*}
        &\Pr\begin{bmatrix}
            \begin{aligned}
            &\exists \transcMat^* \in \rowball(\ippinput):~ \transcMat^* \ne \ippinput\\
            &\wedge \GKRcirc(\transcMat^*) = 1
            \end{aligned}
            \end{bmatrix} \\
        &\le \epsilon + 2^{-\secpar},
        \end{align*}
        \else
        \[
        \Pr\left[\exists \transcMat^* \in \rowball(\ippinput):~ \transcMat^* \ne \ippinput ~\wedge~ \GKRcirc(\transcMat^*) = 1\right] \le \epsilon + 2^{-\secpar},
        \]
        \fi
        where the probability is over the random coins of $\cV_\Distance$.
    \end{itemize}
    Let $\tO$ be hiding $\polylog(\abs{\FF}, \batchK, n, a, \ell)$ factors.
    Denote by $|\cS|=g$,
    then
    \begin{itemize}
        \item $\ell_\Distance = O(\ell)$.
        \item $a_\Distance = \tO(d a)$.
        \item $b_\Distance = \tO(\max(b, d)) = \tO(d a)$, with the simplifying assumption that $b \le a$.
        \item $\Ptime_\Distance = \tO(g \ell \cdot \Ptime + dga\ell^2 + \abs{\lrag \C} + \abs{\lrag \cS} + \batchK \cdot n)$.
        \item $\Vtime_\Distance = \tO(da\ell + \abs{\lrag\C} + \abs{\lrag{\cS}} + \batchK \cdot n)$.
    \end{itemize}
    The bit lengths of the descriptions are $\abs{\lrag{\cSmid}} = \tO(\abs{\lrag\cS} + a\ell)$, $\abs{\lrag{\Cmid}} = \tO(\abs{\lrag\C} + \abs{\lrag{\cS}} + da\ell + \batchK \cdot n)$.
    Let $G(i, \lrag \cS)$ be the circuit that returns the $i$-th element in $\cS$,
    and $C(x, \lrag \C)$ be the circuit that returns $\C(x)$,
    then the new circuits $G_\Distance(i, \lrag \cSmid)$ and $C_\Distance(x, \lrag \Cmid)$,
    satisfy the following.
    \begin{itemize}
        \item $\size(G_\Distance) = \size(G) + b\ell + \tO(1)$.
        \item $\depth(G_\Distance) = \depth(G) + \tO(1)$.
        \item $\size(C_\Distance) = \tO(g \ell \cdot \size(G) + \size(C) + g \cdot S + gda\ell^2)$.
        \item $\depth(C_\Distance) \le \max(\depth(C), \depth(G) + D) + \tO(1)$.
    \end{itemize}
    The verifier never rejects this protocol,
    i.e. $\cV_\Distance$ always outputs 1.
\end{lemma}

\subsubsection{The Instance Reduction Protocol}
The full specification of $\prot{\Reduce}$ and its analysis are in \Cref{sec:phase2}.
It has the following properties.

The protocol takes as input $\parammid$,
as well the parameters $\secpar, d \in \NN$ and the field $\FF$ (i.e. the same parameters as $\prot{\Distance}$),
and outputs $\param{'}$ such that $\abs{\cS'} \approx \frac{\abs{\cSmid}}{d}$,
with the guarantee that $(\cS', \C') \in \cL'_{\bm x}$ only if $\ippinput$ is $\Delta_c$-$d$-close to satisfying $\Cmid$.
Importantly, $\cV_\Reduce$ never accesses $\ippinput$ explicitly during this protocol.

\begin{lemma}[Instance Reduction Protocol for $\Delta_c$-Distance]
    \label{clm:phase2}
    Let $\cL$ be a language with an \\$\epsilon$-unambiguous $(\ell, a, b, \Ptime, \Vtime)$-\UIP $\prot{}$,
    where the verifier's decision $\cV$ is a log-space uniform boolean circuit of size $S$ and depth $D$.

    There exists a constant $C$,
    and a protocol \prot{\Reduce} (\Cref{alg:Reduce}) that takes as input $\bm x, \parammid$,
    where $\bm x \in (\bin^n)^{\batchK}$,
    $\cSmid \subset [\batchK] \times \bin^{b \ell}$,
    $\Cmid : \bin^{\abs{\cSmid} \times a \ell} \to \bin$,
    an unambiguity parameter $\secpar \in \NN$,
    a distance parameter $d \in \NN$,
    a field $\FF$ that satisfies $\abs{\FF} \ge \FboundconcreteT$,
    and outputs $\param{'}$,
    with $\abs{\cS'} = \tO(\frac{\abs{\cSmid}}{d})$ (where $\tO$ hides $\polylog(\abs{\FF}, \batchK, n, a, \ell)$ factors).
    Let $\K \coloneqq \abs{\cSmid} \le \batchK \cdot \ell$,
    and suppose $\K = 2^{\lgK}$ for $\lgK \in \NN$.
    If additionally, the technical condition $d \ge \dboundm$ holds,
    then the protocol has the following guarantees.
    \begin{itemize}[label=-]
        \item \textbf{Prescribed Completeness:} \\If $\cV_\Reduce$ interacts with $\cP_\Reduce$,
        then $(\cS', \C') \in \cL'_{\bm x}$ iff $(\cSmid, \Cmid) \in \cL'_{\bm x}$.
        \item \textbf{$2^{-\secpar}$-Unambiguous Distance Preservation:} 
        Suppose for every $\transcMat^* \in \rowball(\ippinput) \setminus \set{\ippinput}$, $\Cmid(\transcMat^*) = 0$, 
        then for every cheating prover $\cP^*$ that deviates from $\cP_\Reduce$ first in round $j^*$,
        let the output of the interaction between $\cV_\Reduce$ and $\cP^*$ be $\param{'}$,
        then
        \[
        \Pr[\cV_\Reduce \text{ accepts } \wedge~ \C'(\transcMat^{\cS'}) = 1] \le 2^{-\secpar}, 
        \]
        where the probability is over $\cV_\Reduce$'s remaining coins. 
    \end{itemize}
    The protocol has the following complexities.
    \begin{itemize}
        \item $\ell_\Reduce = \tO(\ell \cdot \depth(C_\Distance) \log \size(C_\Distance))$.
        \item $a_\Reduce = b_\Reduce = \tO(da + \poly(d))$.
        \ifFOCS
        \item $\begin{aligned}[t]
            \Ptime_\Reduce = \poly(&\Flog, \batchK, a, \ell, d, \size(C_\Distance), \\
            &\abs{\lrag{\cS_\Distance}},\abs{\lrag{\C_\Distance}}, n).
        \end{aligned}$
        \else
        \item $\Ptime_\Reduce = \poly(\Flog, \batchK, a, \ell, d, \size(C_\Distance), \abs{\lrag{\cS_\Distance}},\abs{\lrag{\C_\Distance}}, n)$. 
        \fi
        \ifFOCS
        \item $\begin{aligned}[t]
            \Vtime_\Reduce = \tO(&d a\ell \cdot (\depth(C_\Distance) \log \size(C_\Distance)\\
            &+ \abs{\lrag{\Cmid}}) + \poly(d) \\
            &+ \abs{\lrag{\cSmid}} + \batchK \cdot n).
        \end{aligned}$
        \else
        \item $\Vtime_\Reduce = \tO(d a\ell \cdot (\depth(C_\Distance) \log \size(C_\Distance) + \abs{\lrag{\Cmid}}) + \poly(d) + \abs{\lrag{\cSmid}} + \batchK \cdot n)$.
        \fi
    \end{itemize}

    The bit lengths of the descriptions are $\abs{\lrag{\cS'}} = \tO(\abs{\lrag\cSmid} + \poly(d))$, $\abs{\lrag{\C'}} = \tO(da\ell + \poly(d))$.
    Furthermore, suppose:
    \begin{itemize}
        \item $G_\Distance(i, \lrag {\cS_\Distance})$ is the circuit that returns the $i$-th element in $\cS_\Distance$, and
        \item $C_\Distance(x, \lrag {\C_\Distance})$ is the circuit that returns $\C_\Distance(x)$.
    \end{itemize}
    The new circuits $G'(i, \lrag {\cS'})$ and $C'(x, \lrag {\C'})$ satisfy the following.
    \begin{itemize}
        \item $\size(G') = \size(G_\Distance) + \tO(\poly(d))$.
        \item $\depth(G') = \depth(G_\Distance) + \tO(1)$.
        \item $\size(C') = \tO(\frac{\K}{d} \cdot a \ell)$.
        \item $\depth(C') = \tO(1)$.
    \end{itemize}
    And the verifier's verdict circuit (which outputs 0 iff it rejects amidst the protocol) satisfies:
    \begin{itemize}
        \ifFOCS
        \item $\begin{aligned}[t]
            \size(\cV_\Reduce) = \tO(&d a\ell \cdot (\depth(C_\Distance) \log \size(C_\Distance)\\
            &+ \abs{\lrag{\Cmid}}) + \poly(d) \\
            &+ \abs{\lrag{\cSmid}} + \batchK \cdot n).
        \end{aligned}$
        \else
        \item $\size(\cV_\Reduce) = \tO(d a\ell \cdot (\depth(C_\Distance) \log \size(C_\Distance) + \abs{\lrag{\Cmid}}) + \poly(d) + \abs{\lrag{\cSmid}} + \batchK \cdot n)$.
        \fi
        \item $\depth(\cV_\Reduce) = \tO(1)$.
    \end{itemize}
\end{lemma}
The number of instances is concretely $\abs{\cS'} \le \ceil{8\secpar\frac{\abs{\cSmid}}{d}}$.

\subsubsection{The Explicit \UIP for the associated language}
We give the full specification of the \UIP $\prot{\cL'_{\bm x}}$ that explicitly checks $(\cS, \C) \in \cL'_{\bm x}$ in \Cref{sec:final-check}.
This protocol has a cost proportional to $\abs{\cS}$.
Therefore,
it is only run at the end of the \bUIP protocol,
when $\abs{\cS}$ is small.
\begin{lemma}[The \UIP for $\cL'_{\bm x}$]
    \label{lem:final-check}
    Let $\cL$ be a language with an $\epsilon$-unambiguous $(\ell, a, b, \Ptime, \Vtime)$-\UIP $\prot{}$,
    where the verifier verdict $\cV$ is a log-space uniform boolean circuit of size $S$ and depth $D$.

    There exists an $\epsilon$-unambiguous ($\ell'$, $a'$, $b'$, $\Ptime'$, $\Vtime'$)-\UIP $\prot{\cL'_{\bm x}}$ for verifying $(\cS, \C) \in \cL'_{\bm x}$.
    Both parties have input $\bm x$ and $\param{}$,
    where $\abs{\cS} = g$,
    and in the end $\cV_{\cL'_{\bm x}}$ either accepts or rejects.

    Let $G(i, \lrag \cS)$ be the circuit that returns the $i$-th element in $\cS$,
    and $C(x, \lrag \C)$ be the circuit that returns $\C(x)$.
    The protocol has the following complexities.
    \begin{itemize}
        \item $\ell_{\cL'_{\bm x}} = \ell + 1$.
        \item $a_{\cL'_{\bm x}} = g \ell \cdot a$.
        \item $b_{\cL'_{\bm x}} = b$.
        \item $\Ptime_{\cL'_{\bm x}} = O(g\ell \cdot \Ptime)$.
        \item $\Vtime_{\cL'_{\bm x}} = O(g\ell \cdot \Vtime + \size(C) + g\size(G))$.
    \end{itemize}
    Let $\tO$ hide $\polylog(g, \ell)$ factors.
    The verifier's decision circuit $\cV_{\cL'_{\bm x}}$ has the following complexities:
    \begin{itemize}
        \item $\size(\cV_{\cL'_{\bm x}}) = \tO(g \ell \cdot S + \size(C) + g\size(G))$.
        \item $\depth(\cV_{\cL'_{\bm x}}) = \max(\depth(G) + D, \depth(C)) + \tO(1)$.
    \end{itemize}
\end{lemma}

\subsection{Our Construction of \bUIP}
\label{sec:pf-bUIP}
We show how to prove \Cref{thm:basic-bUIP} assuming \Cref{clm:phase1,clm:phase2,lem:final-check}.
\begin{proof}[Proof of \Cref{thm:basic-bUIP}]
Denote the sub-protocols as $\prot{\Distance}$, $\prot{\Reduce}$, and $\prot{\cL'_{\bm x}}$.
The full batching protocol $\protU{'}$ is given in \Cref{alg:basic-bUIP}.

\begin{algorithm}
  \setstretch{1.2}
  \caption{Our \bUIP $\protU{'}$.}
  \label{alg:basic-bUIP}
    \textbf{Input Parameters:} $\secpar \in \NN$.\\
    \textbf{Input:} $\batchK \in \NN$, and $\bm x = (x_1,\ldots,x_{\batchK})$ are statements.\\
    \textbf{Other Parameters:} $d = \dbound$ is the distance parameter,
    $\FF$ is a characteristic-2 finite field of size $\abs{\FF} \ge \Fboundconcrete$,
    (where $C_0$ is some universal constant),
    and $\nround = \floor{\log \batchK}$ is the number of iterations.\\
    \textbf{Ingredients:} 
    \begin{itemize}[topsep=-0.02em,itemsep=-0.2em]
        \item The \emph{Distance Generation Protocol}, $\prot{\Distance}$ (\Cref{clm:phase1}).

        It takes in a claim $\param{}$ over $\abs{\cS}$ active pairs and output an intermediate claim $\parammid$ over $\abs{\cSmid} = \K = g \ell$ active pairs.

        \item The \emph{Instance Reduction Protocol}, ($\cP_\Reduce$, $\cV_\Reduce$) (\Cref{clm:phase2}).

        It takes in the claim $\parammid$ over $\K$ active pairs and output a reduced claim $\param{'}$ over $\abs{\cS'} \le \K/d$ active pairs.
        \item The $\UIP$ $\protU{'}$ for explicitly checking the final claim $(\cS^{(\nround)}, \C^{(\nround)}) \in \cL'_{\bm x}$
        (\Cref{lem:final-check}).
    \end{itemize}
  \begin{algorithmic}[1]
    \State $\cV'$ samples a random $\bm q \gets \bin^{b\ell}$, and send it entirely to $\cP'$.
    \State Both parties let $\cS^{(0)} \gets \set{(i, \bm q)}_{i \in [\batchK]}$, 
    and $\C^{(0)}$ be the predicate that checks every row $\transc^{x_i, \bm q}$ in $\transcMat^{\cS^{(0)}}$ are accepted, i.e. $\cV(x_i, \bm q, \transc^{x_i, \bm q}) = 1$ for every $i \in [\batchK]$.
    \State Define $\lrag{\cS^{(0)}} = (\bm q)$,
    and $G^{(0)}(\lrag{\cS^{(0)}})$ which enumerates $\cS^{(0)}$ is the circuit that outputs $(i, \bm q)$ for every $i \in [\batchK]$.
    \State Define $\lrag{\C^{(0)}} = (\bm x, \bm q)$, and $C^{(0)}(x, \lrag{\C^{(0)}})$ be the circuit that checks $\C^{(0)}(x) = 1$.
    \For{$j = 1,\ldots, \nround$}
        \State Run $\prot{\Distance}$ on $\paramold$ with parameters $\secpar$, $d$, and $\FF$,
        obtaining $\lrag{\cSmid^\old}$, $\lrag{\Cmid^\old}$.
        \State Run ($\cP_\Reduce$, $\cV_\Reduce$) on $\lrag{\cSmid^\old}$ and $\lrag{\Cmid^\old}$ with parameters $\secpar$, $d$, and $\FF$,
        obtaining $\paramnew$.
    \EndFor
    \State Run the \UIP $\prot{\cL'_{\bm x}}$ explicitly on $\bm x$ and $\paramC{\nround}$ to check $(\cS^{(\nround)}, \C^{(\nround)}) \in \cL'_{\bm x}$.
    \State $\cV'$ accepts iff $\cV_{\cL'_{\bm x}}$ accepts.
    \end{algorithmic}
\end{algorithm}

\paragraph{Prescribed Completeness.}
Note that $(\cS^{(0)}, \C^{(0)}) \in \cL'_{\bm x}$ iff $\C^{(0)}(\transcMat^{\cS^{(0)}}) = 1$ and $\bm x \in \cL^{\otimes \batchK}$.
The former condition is true iff $\cV(x_i, \bm q, \transc^{x_i, \bm q}) = 1$ for every $(i, \bm q) \in \cS^{(0)}$, which is true iff $\bm x \in \cL^{\otimes \batchK}$ because of the prescribed completeness of $\prot{}$.
By invoking the prescribed completeness of $\prot{\Distance}$ and $\prot{\Reduce}$ and $\prot{\cL'_{\bm x}}$,
we get that $\cV'$ accepts iff $\bm x \in \cL^{\otimes \batchK}$.

\paragraph{Unambiguity.}
Let $\cP^*$ be a prover strategy that deviates from $\cP'$ in some round,
then it must be deviating in one of the $\nround$ applications of $\prot{\Distance}$,
one of the $\nround$ applications of $\prot{\Reduce}$, 
or in the final $\prot{\cL'_{\bm x}}$.
The probability that $\cV'$ accepts is upper bounded by the probability that unambiguity breaks in any of the subsequent rounds,
which,
by a union bound,
is at most $\epsilon' = \nround \cdot (\epsilon + 2^{-\secpar} + 2^{-\secpar}) + \epsilon < (\nround + 1)\epsilon + 2\cdot \nround 2^{-\nround} < 2 \log \batchK \cdot (\epsilon + 2^{-\secpar})$.

\paragraph{Number of Intermediate Instances}
We show by induction that $\abs{\cS^{(j)}} \le \ceil{\batchK \cdot 2^{-j}}$,
for $j \in [\nround]$.
For the base case, $\abs{\cS^{(0)}} = \batchK$ by definition.
Assuming the inductive hypothesis,
by the guarantees of $\prot{\Distance}$,
$\abs{\cSmid^{(j)}} = \abs{\cS^{(j)}} \cdot \ell \le \ceil{\batchK \cdot 2^{-j}} \cdot \ell$.
We conclude by the guarantees of $\prot{\Reduce}$,
$\abs{\cS^{(j+1)}} \le \ceil{8\secpar\abs{\cSmid^{(j)}} / d} \le \frac{\ceil{8\secpar \cdot \batchK \cdot 2^{-j} \cdot \ell}}{\dbound} < \ceil{\batchK \cdot 2^{-j-1}}$.

\paragraph{Description Lengths}
We first bound the bit lengths of all descriptions.\\
With $\tO$ hiding $\polylog(\abs{\FF}, \batchK, n, a, \ell)$ factors,
for all $1 \le j \le \nround < \log \batchK = \tO(1)$,
and
\begin{itemize}
    \ifFOCS
    \item $\begin{aligned}[t]
        \abs{\lrag{\cS^{\new}}} &= \tO(\abs{\lrag{\cSmid^{\old}}} + \poly(d))\\
        &= \tO(\abs{\lrag\cS^{\old}} + a \ell + \poly(d)) \\
        &= \ldots = \tO(j \cdot (a\ell + \poly(d))) \\
        &= \tO(a\ell + \poly(d)), \text{ because }j = \tO(1).
    \end{aligned}$ 
    \else
    \item $\begin{aligned}[t]
        \abs{\lrag{\cS^{\new}}} &= \tO(\abs{\lrag{\cSmid^{\old}}} + \poly(d))\\
        &= \tO(\abs{\lrag\cS^{\old}} + a \ell + \poly(d)) \\
        &= \ldots = \tO(j \cdot (a\ell + \poly(d))) = \tO(a\ell + \poly(d)), \text{ because }j = \tO(1).
    \end{aligned}$ 
    \fi
    \item $\abs{\lrag{\C^{\new}}} = \tO(a\ell +  \poly(d))$.
    \item $\abs{\lrag{\Cmid^{\new}}} = \tO(\abs{\lrag\C^{\new}} + \abs{\lrag{\cS^{\new}}} + da\ell + \batchK \cdot n) = \tO(da\ell + \poly(d) + \batchK \cdot n)$.
\end{itemize}
Note that $\abs{\lrag{\C^{(0)}}} = \tO(\batchK \cdot n + a\ell)$,
because $\lrag{\C^{(0)}} = (\bm x, \bm q)$.

\paragraph{Intermediate Circuit Sizes and Depths}
For all $1 \le j \le \nround \le \log \batchK = \tO(1)$,
let $G^{\new}(s, \lrag{\cSnew})$ be the circuit that returns the $s$-th element in $\cSnew$,
and let $C^{\new}(x, \lrag{\Cnew})$ be the circuit that returns $\Cnew(x)$.
For the base case,
we already have $\size(G^{(0)}) = \tO(\batchK + a\ell)$ and $\depth(G^{(0)}) = O(1)$,
and $\size(C^{(0)}) = \tO(\batchK \cdot S)$ and $\depth(C^{(0)}) = D + \tO(1)$.
By the guarantees of \Cref{clm:phase1,clm:phase2},
we calculate the following:
\begin{itemize}
    \item $\size(G^{\new}) = \size(G^{\old}_{\Distance}) + \tO(\poly(d)) = \size(G^{\old}) + \tO(a \ell + \poly(d)) = \ldots = \tO(\batchK + a\ell + \poly(d))$.
    \item $\depth(G^{\new}) = \depth(G^{\old}_{\Distance}) + \tO(1) = \ldots = \tO(1)$.
    \item $\size(C^{\new}) = \tO(\abs{\cS^{\new}} \cdot a \ell)$.
    \item $\depth(C^{\new}) = \tO(1)$.
    \ifFOCS
    \item $\begin{aligned}[t]
        \size(C^{\new}_{\Distance}) &= \tO\left(\abs{\cS^{\new}} \cdot \ell \cdot \size(G^{\new}) + \size(C^{\new}) \right.\\
        &\quad \left. \vphantom{\abs{\cS^{\new}}} + \abs{\cS^{\new}} \cdot S + gda\ell^2\right)\\
        &= \tO\left(\abs{\cS^{\new}} \cdot \ell\cdot(\batchK + a \ell + \poly(d)) \right.\\
        &\quad \left. \vphantom{\abs{\cS^{\new}}} + \abs{\cS^{\new}} \cdot (a \ell + S)+ gda\ell^2\right)\\
        &= \tO(\batchK^2\ell + \batchK a \ell^2 + \batchK a \ell \\
        &\quad + \batchK \cdot \ell\cdot \poly(d) + \batchK S + gda\ell^2).
    \end{aligned}$
    \else
    \item $\begin{aligned}[t]
        \size(C^{\new}_{\Distance}) &= \tO(\abs{\cS^{\new}} \cdot \ell \cdot \size(G^{\new}) + \size(C^{\new}) + \abs{\cS^{\new}} \cdot S + gda\ell^2)\\
        &= \tO(\abs{\cS^{\new}} \cdot \ell\cdot(\batchK + a \ell + \poly(d)) + \abs{\cS^{\new}} \cdot (a \ell + S)+ gda\ell^2)\\
        &= \tO(\batchK^2\ell + \batchK a \ell^2 + \batchK a \ell + \batchK \cdot \ell\cdot \poly(d) + \batchK S + gda\ell^2).
    \end{aligned}$
    \fi
    \ifFOCS
    \item $\begin{aligned}[t]
        \depth(C^{\new}_{\Distance}) &= \max(\depth(C^{\old}), \depth(G^{\old}) \\
        &\quad + D) + \tO(1) = \ldots = D + \tO(1).
    \end{aligned}$
    \else
    \item $\depth(C^{\new}_{\Distance}) = \max(\depth(C^{\old}), \depth(G^{\old}) + D) + \tO(1) = \ldots = D + \tO(1)$.
    \fi
\end{itemize}
Therefore,
\ifFOCS
$\depth(C^{\new}_{\Distance})\log(\size(C^{\new}_{\Distance})=(D + \tO(1))(\log S + \tO(1))$ $= D\log S + \tO(1)$.
\else
$\depth(C^{\new}_{\Distance})\log(\size(C^{\new}_{\Distance})=$ $(D + \tO(1))(\log S + \tO(1)) = D\log S + \tO(1)$.
\fi

The technical conditions stated in \Cref{clm:phase2} are satisfied by the choice of parameters:
\begin{enumerate}
    \item 
    Since we choose $\abs{\FF}$ to be at least $\Fboundconcrete$, 
    it satisfies the field size requirement of $\abs{\FF} \ge \FboundconcreteT$ as long as $C_0 \ge 2C$,
    where $C$ is the constant in \Cref{clm:phase2}.

    \item
    Fix $j \in [\nround]$,
    by the above analysis, 
    $\abs{\cSmid^{\old}} = \abs{\cSold} \cdot \ell \le \ceil{\batchK \cdot 2^{1-j} \cdot \ell}$.
    To apply $\prot{\Reduce}$, the condition $d \ge\dboundm$ must hold,
    where $m \coloneqq \log \abs{\cSmid^{\old}} \le \log (\batchK \ell)$.
    Indeed,
    $d = \dbound > \dboundm$.
\end{enumerate}

We rely on $d = \dbound = \tO(\ell)$ and $\nround \le \log \batchK = \tO(1)$ to simplify the complexities as follows.

The overall cost equals the cost of running the two sub-protocols $\Distance$ and $\Reduce$ for $\nround$ times,
plus the cost of running the final \UIP $\prot{\cL'_{\bm x}}$ once.
Note that the size of the final $\abs{\cS^{(\nround)}}$ is at most $\batchK \cdot 2^{-\nround} = \tO(1)$.
We use notations such as $\Ptime_{\Distance, j}$ to denote the time complexity of running $\Distance$ in the $j$-th round (and similarly for $\Reduce$ and $\cL'_{\bm x}$).
\begin{itemize}
    \item $\ell' = \sum_{j = 1}^\nround(\ell_{\Distance, j} + \ell_{\Reduce, j}) + \ell_{\cL'_{\bm x}}) = \tO(\nround \cdot \ell D \log S) = \tO(\ell D \log S)$.
    \item $a' = b' = \max(a_\Distance, a_\Reduce, a_{\cL'_{\bm x}}) = \tO(\max(d a, da + \poly(d), \abs{\cS^{(\nround)}}\ell \cdot a)) = \tO(a \ell + \poly(\ell))$.
    \ifFOCS
    \item $\begin{aligned}[t]
        \Ptime' &= \tO(\sum_{j = 1}^\nround (\Ptime_{\Distance, j} + \Ptime_{\Reduce, j}) + \Ptime_{\cL'_{\bm x}})\\
        &= \tO\left(\sum_{j = 1}^\nround((\batchK \cdot 2^{1-j}) \cdot \ell \cdot \Ptime + d (\batchK \cdot 2^{1-j})a \ell^2\right.\\
        &\quad \left. \vphantom{\sum_{j = 1}^\nround} + \abs{\lrag{\Cnew}} + \abs{\lrag{\cSmid}} + \batchK \cdot n)\right) \\
        &\quad + \tO\left(\sum_{j = 1}^\nround\poly(d, \batchK \cdot 2^{1-j}, a, \ell, \size(C_{\Distance}), \right.\\
        &\quad \left. \vphantom{\sum_{j = 1}^\nround} \abs{\lrag{\Cnew}}, \abs{\lrag{\cSmid}}), n)\right)\\
        &\quad + \tO(\ell \cdot \Ptime)\\
        &=  \tO(\batchK \cdot \ell \cdot \Ptime + \poly(\batchK, a, \ell, S)).
    \end{aligned}$
    \else
    \item $\begin{aligned}[t]
        \Ptime' &= \tO(\sum_{j = 1}^\nround((\batchK \cdot 2^{1-j}) \cdot \ell \cdot \Ptime + d (\batchK \cdot 2^{1-j})a \ell^2 + \abs{\lrag{\Cnew}} + \abs{\lrag{\cSmid}} + \batchK \cdot n)) \\
        &\quad+ \tO(\sum_{j = 1}^\nround\poly(d, \batchK \cdot 2^{1-j}, a, \ell, \size(C_{\Distance}), \abs{\lrag{\Cnew}}, \abs{\lrag{\cSmid}}), n)\\
        &\quad + \tO(\ell \cdot \Ptime)\\
        &=  \tO(\batchK \cdot \ell \cdot \Ptime + \poly(\batchK, a, \ell, S)).
    \end{aligned}$
    \fi
    \ifFOCS
    \item $\begin{aligned}[t]
        \Vtime' &= \sum_{j = 1}^{\nround} \Vtime_{\Distance,j} + \sum_{j = 1}^{\nround} \Vtime_{\Reduce,j}\\
        &\quad + \Vtime_{\cL'_{\bm x}}\\
        &= \tO(\nround \cdot d a \ell + 
        \abs{\lrag{\cSnew}} + \abs{\lrag{\Cnew}} + \batchK \cdot n)\\
        &\quad + \tO\left(\sum_{j = 0}^{\nround - 1}\left(\vphantom{\abs{\lrag{\Cmid^{\new}}}}da\ell(\depth(C_\Distance) \right.\right.\\
        &\quad\quad\quad \quad \quad \left.\cdot \log \size(C_\Distance) + \abs{\lrag{\Cmid^{\new}}}\right)\\
        &\quad\quad\quad + \poly(d) + \left. \vphantom{\sum_{j = 0}^{\nround - 1}}\abs{\lrag{\cSmid^{\new}}} + \batchK \cdot n)\right)\\
        &\quad + \tO\left(\abs{\cS^{(\nround)}} \ell \cdot \Vtime + \size(\C^{(\nround)}) +\right.\\
        &\quad\quad\quad\left. \vphantom{\abs{\cS^{(\nround)}}}\abs{\cS^{(\nround)}} \size(G^{(\nround)})\right)\\
        &=\tO(\ell \cdot \Vtime + a \ell^2(\batchK n + D\log S) \\
        &\quad \quad \quad+ a^2\cdot \poly(\ell)). 
    \end{aligned}$
    \else
    \item $\begin{aligned}[t]
        \Vtime' &= \sum_{j = 1}^\nround \Vtime_{\Distance,j} + \sum_{j = 1}^\nround \Vtime_{\Reduce,j} + \Vtime_{\cL'_{\bm x}}\\
        &= \tO(\nround \cdot d a \ell + 
        \abs{\lrag{\cSnew}} + \abs{\lrag{\Cnew}} + \batchK \cdot n)\\
        &\quad + \tO\left(\sum_{j = 0}^{\nround - 1}(da\ell(\depth(C_\Distance) \log \size(C_\Distance) + \abs{\lrag{\Cmid^\new}}) + \poly(d) + \abs{\lrag{\cSmid^\new}}\right.\\
        &\quad\quad\quad+\left.\vphantom{\sum_{j = 0}^{\nround - 1}}\batchK \cdot n\right)\\
        &\quad + \tO(\abs{\cS^{(\nround)}} \ell \cdot \Vtime + \size(\C^{(\nround)}) + \abs{\cS^{(\nround)}} \size(G^{(\nround)}))\\
        &= \tO\left(\nround d a \ell + \batchK \cdot n + \nround da\ell(D\log S + da\ell + \poly(d) + \batchK \cdot n) + \ell\Vtime + \batchK + a \ell\right)\\
        &=\tO(\ell \cdot \Vtime + a \ell^2(\batchK n + D\log S) + a^2\cdot \poly(\ell)). 
    \end{aligned}$
    \fi
\end{itemize}
Finally,
the verifier's verdict circuit $\cV'$ rejects iff any of the sub-protocols rejects.
(Note that $\cV_\Distance$ never rejects, so there are no terms corresponding to these sub-protocols.)
\begin{itemize}
    \ifFOCS
    \item $\begin{aligned}[t]
    \size(\cV') &\le \sum_{j = 1}^{\nround}\size(\cV_{\Reduce,j}) + \size(\cV_{\cL'_{\bm x}})\\
    &= \tO((\ell \cdot S + a\ell^2(\batchK n + D\log S)) + \\
    &\quad \quad \quad a^2\cdot \poly(\ell)).
    \end{aligned}$
    \else
    \item $\begin{aligned}[t]
        \size(\cV') &\le \sum_{j = 1}^{\nround}\size(\cV_{\Reduce,j}) + \size(\cV_{\cL'_{\bm x}})\\
        &= \tO((\ell \cdot S + a\ell^2(\batchK n + D\log S)) + a^2\cdot \poly(\ell)).
    \end{aligned}$
    \fi
    \ifFOCS
    \item $\begin{aligned}[t]
        \depth(\cV') &= \max\left(\set{\depth(\cV_{\Reduce,j})}_{j \in [\nround]}, \right.\\
        &\quad \quad \quad \left. \vphantom{\depth(\cV_{\Reduce,j})} \depth(\cV_{\cL'_{\bm x}})\right) + O(\log \nround)\\
        &= \max\left(\set{\tO(1)}_{j \in [\nround]}, \right.\\
        &\quad \quad \quad \max(\depth(G^{(\nround)}) + D, \\
        &\quad \quad \quad \left.\vphantom{\depth(\cV_{\Reduce,j})} \depth(C^{(\nround)}))\right) + O(\log \nround)\\
        &= D + \tO(1).
    \end{aligned}$
    \else
    \item $\begin{aligned}[t]
        \depth(\cV') &= \max(\set{\depth(\cV_{\Reduce,j})}_{j \in [\nround]}, \depth(\cV_{\cL'_{\bm x}})) + O(\log \nround)\\
        &= \max(\set{\tO(1)}_{j \in [\nround]}, \max(\depth(G^{(\nround)}) + D, \depth(C^{(\nround)}))) + O(\log \nround)\\
        &= D + \tO(1).
    \end{aligned}$
    \fi
\end{itemize}
And moreover,
$\cV'$ is log-space uniform because $\cV'$ is simply the \AND's of the $\nround$ $\cV_\Reduce$ and the final $\cV_{\cL'_{\bm x}}$, and all of these are themselves log-space uniform.
\end{proof}

\subsection{Proof of \Cref{clm:phase1}: \protphaseone}
\label{sec:phase1}

\subsubsection{Additional Ingredients}
\paragraph{A helper \UIP for verifying $\transc = \transc^{x,\bm q}$}
Recall that $\transc^{x,\bm q}$ is the prescribed transcript of $\prot{}$ on the statement $x$ when verifier uses random coin $\bm q$ (\Cref{def:prescribed}).
There exists a \UIP, 
which we call a \emph{Transcript Checker} protocol,
that allows a verifier to certify that some transcript $\transc \in \bin^{a\ell}$ indeed equals $\transc^{x, \bm q}$.

The \emph{Random Continuation Protocol} is such a Transcript Checker.
In this protocol,
the verifier samples a \emph{fresh} random coin sequence $\bm q' = (q'_1, \ldots, q'_{\ell}) \in (\bin^b)^\ell$,
and challenges the prover to succeed in answering $\ell$ hybrid runs of $\prot{}$ on the input $x_i$,
where the $j$-th hybrid run uses the hybrid coin sequence $\bm q^{\hyb}_j \coloneqq (\bm q_{\le j}, \bm q_{> j}')$.
The proof of \Cref{lem:Lxq} is straightforward and deferred to \Cref{sec:part2}.
\begin{proposition}[Transaction Checker; c.f. \emph{Random Continuation} in Lemma 6.3 of \cite{STOC:ReiRotRot16}]
    \label{def:tauchecker}
    \label{lem:Lxq}
    Let $\cL$ be a language with an $\epsilon$-unambiguous $(\ell, a, b, \Ptime, \Vtime)$-\UIP $\prot{}$.
    There exists an $\epsilon$-unambiguous $(\ell, \ell \cdot a, b, \ell\cdot\Ptime, \ell\cdot\Vtime)$-\UIP $\protU{^\tc}$,
    where both parties take $((x, \bm q), \transc)$ as input,
    and the verifier decides whether to accept or reject.
    We call $\protU{^\tc}$ the \emph{Transcript Checker},
    and it satisfies the following properties:
    \begin{itemize}[label=-]
        \item As a \UIP for checking $((x, \bm q), \transc) \in \cR$, it satisfies prescribed completeness and $\epsilon$-unambiguity.
        \item \textbf{Transcript Structure.} 
        The messages of $\cP^\tc$ can be exactly broken down into $\ell$ prescribed messages of the base \UIP $\prot{}$.
        Namely, the $r$-th round message of $\cP^\tc$ on input $((x, \bm q), \transc)$ is $\bm m_r = \left((\answerMat^{x, \bm q^\hyb_j}_r)_{j \in [\ell]}\right)^\top$,
        where $\answerMat^{x, \bm q^\hyb_j}_r = (\cP(x, (\bm q^\hyb_j)_{< r}))$ is the $r$-th round message of $\cP$.
        Therefore, 
        the entire transcript of $\cP^\tc$ equals $(\bm m_1, \ldots, \bm m_{\ell}) = \left((\answerMat^{x, \bm q^\hyb_j})_{j \in [\ell]}\right)^\top$.

        \item \textbf{Check Structure.} 
        The verifier $\cV^\tc$ on a given transcript $\transcMat = \left((\transcMat^{x, \bm q^\hyb_j})_{j \in [\ell]}\right)^\top$
        outputs 1 iff (1) for every $j \in [\ell]$,
        the corresponding $\cV(x, \bm q^\hyb_j, \transcMat^{x, \bm q^\hyb_j}) = 1$,
        and (2) for each $j \in [\ell]$,
        the $(j \cdot a)$-th prefix of $\transcMat^{x, \bm q^\hyb_j}$ is the same as the $(j \cdot a)$-th prefix of $\transcMat^{x, \bm q}$.
    \end{itemize}
\end{proposition}
\begin{remark}
    \label{rem:tauchecker}
    When a batch of inputs $\bm x = (x_1, \ldots, x_\batchK) \in (\bin^n)^{\batchK}$ is given,
    for any sequence of $g$ pairs $\cS = ((i^1, \bm q^1), \ldots, (i^g, \bm q^g)) \subset [\batchK] \times (\bin^b)^\ell$,
    and $g$ transcripts $\transcMat^1, \ldots, \transcMat^g \in \bin^{\ell \times (a\ell)}$,
    we consider a batch run of $g$ instances of $\protU{^\tc}$ on $((x_{i^1}, \bm q^1), \transcMat^1), \ldots, ((x_{i^g}, \bm q^g), \transcMat^g)$ where $\cV^\tc$ uses a shared random coin $\bm q'$ for all the $g$ instances.
    Let 
    \ifFOCS
    \begin{align*}
    \cS^\hyb &\coloneqq \left((i^1, (\bm q^1)^\hyb_1), \ldots, (i^1, (\bm q^1)^\hyb_\ell),\right. \\
    &\quad \quad \quad \left.\ldots, (i^g, (\bm q^g)^\hyb_1), \ldots, (i^g, (\bm q^g)^\hyb_\ell)\right)
    \end{align*}
    \else
    \[\cS^\hyb \coloneqq ((i^1, (\bm q^1)^\hyb_1), \ldots, (i^1, (\bm q^1)^\hyb_\ell), \ldots, (i^g, (\bm q^g)^\hyb_1), \ldots, (i^g, (\bm q^g)^\hyb_\ell))\]
    \fi
    be the set of $g\ell$ hybrid pairs constructed from every pair in $\cS$ using the new $\bm q'$.
    By the transcript structure of $\protU{^\tc}$,
    the $r$-th round messages of the $g$ instances of $\protU{^\tc}$ constitute exactly the rows of $\transcMat^{\cS^\hyb}_r = \left((\transcMat^{x_i, \bm q}_r)_{(i, \bm q) \in \cS^\hyb}\right)^\top$,
    and the full transcripts of the $g$ instances is $\transcMat^{\cS^\hyb}$.
\end{remark}

\paragraph{Checksum}
Let $\K = 2^\lgK$ for some $\lgK \in \NN$,
and $\FF$ be a constructible field ensemble of characteristic 2,
such that $\Flog = \tO(\secpar)$.

Let $d \in \NN$ be some \emph{distance parameter} and $\cksumT \in \NN$ be some parameter be specified.
We utilize the checksum of a systematic error correcting code with message space $\FF^\K$ of distance $2d$,
which we denoted by $\cksum_{\bm \cksumU}(\cdot)$,
where the parameter $\bm \cksumU = (\cksumU_1,\ldots, \cksumU_{\cksumT}) \in (\FF^\lgK)^{\cksumT}$ is a sequence of points in $\FF^{\lgK}$.
\begin{proposition}[\LDE-instantiated checksum (c.f. Proposition 6.6 of \cite{STOC:ReiRotRot16})]
\label{lem:unique-decoding}
Let $\K = 2^\lgK$ and $\lgK \in \NN$, $\FF$ be a field of characteristic 2.
Let $\secpar \in \NN$ be an unambiguity parameter,
and $\cksumT \coloneqq 2d(\lgK + \Flog) + \secpar$.
There exists a function $\cksum_{\bm \cksumU}: \FF^\K \to \FF^\cksumT$, 
such that with probability $1 - 2^{-\secpar}$ over uniform $\bm \cksumU = (\cksumU_1,\ldots, \cksumU_\cksumT)$,
$\cksum_{\bm \cksumU}(\cdot)$ is the checksum of a systematic linear error-correcting code of distance $2d$ on the message space $\FF^\K$.
In other words,
for any $\bm m \in \FF^\K$,
for any $\bm m', \bm m'' \in \FF^\K$ that are both $d$-close to $\bm m$ and $\bm m' \neq \bm m''$.
$\cksum_{\bm \cksumU}(\bm m') \neq \cksum_{\bm \cksumU}(\bm m'')$.

This allows \emph{unique decoding} in the following sense:
if we know $\bm m'$ is $d$-close to some (fixed) $\bm m$,
then we can uniquely determine $\bm m'$ given $\cksum_{\bm \cksumU}(\bm m')$.
\end{proposition}
\begin{proof}
    Recall that $\hat{\bm m}: \FF^\lgK \to \FF^\cksumT$ is the multilinear extension of $\bm m$ (\Cref{sec:LDE}).
    We can let $\cksum_{\bm \cksumU}(\bm m) \coloneqq \hat{\bm m}(\bm \cksumU) = (\hat {\bm m}(\cksumU_1), \ldots, \hat{\bm m}(\cksumU_{\cksumT})) \in \FF^{\cksumT}$.

    To show unique decoding,
    assume that on the contrary, $\cksum_{\bm \cksumU}(\bm m') = \cksum_{\bm \cksumU}(\bm m'')$,
    Then $\widehat{\bm m' - \bm m''}(\bm \cksumU) = 0$ (by linearity).
    Note $\bm m' - \bm m'' \in \Delta(\pval(\bm \pvalU, \bm 0))$ has Hamming weight at most $2d$,
    so $\Delta(\pval(\bm \cksumU, \bm 0)) \le 2d$,
    but this happens with probability at most $2^{-\secpar}$ by \Cref{lem:pval}.
\end{proof}
Abusing the notation,
given a matrix $\bm M \in \FF^{\K \times a}$ with columns $(\bm m_1,\ldots,\bm m_a)$,
we use \\
$\cksum_{\bm \cksumU}(\bm M)$ to denote $\bm S = (\cksum_{\bm \cksumU}(\bm m_1),\ldots,\cksum_{\bm \cksumU}(\bm m_a)) \in \FF^{\cksumT \times a}$.

\subsubsection{Construction of the \Distance{} Protocol}
\paragraph{Road map of our Construction.}
Recall that the input claim is $(\cS, \C) \in \cL'_{\bm x}$,
and given the parameters $\secpar, d \in \NN$ and the field $\FF$ where $\Flog = \tO(\secpar)$,
the goal is to define $(\cSmid, \Cmid)$ such that $\Cmid$ rejects everything in the $\rowball(\ippinput)$ around $\ippinput$ if $(\cS, \C) \notin \cL'_{\bm x}$,
except with probability at most $(\epsilon + 2^{-\secpar})$.

To this end,
both parties interact for $\ell$ rounds,
to \emph{implicitly} run a batch of $\protU{^\tc}$ in parallel.
These protocols are run implicitly in the following sense.
\begin{itemize}
    \item $\cV_\Distance$ samples $\cksumT = \tO(\secpar d)$ random points that defines the checksum function $\cksum_\cksumU: \FF^{g\ell} \to \FF^\cksumT$.
    The checksums enable unique decoding except with probability $2^{-\secpar}$ over the choice of $\bm \cksumU$ (\Cref{lem:unique-decoding}).
    \item $\cV_\Distance$ samples a new random coin sequence $\bm q' \in (\bin^b)^\ell$, and send the coin to the prover round by round,
    \item $\cP_\Distance$ only responds by the checksums of the underlying messages.
\end{itemize}
Specifically,
in round $r \in [\ell]$,
the prover computes checksums $\cksumv_r \in \FF^\cksumT$ on the columns of the matrix $\answerMat^{\cS^{\hyb}}_r \in \bin^{(g\ell) \times a}$,
which consists of the $(g\ell)$ base-\UIP $r$-th round answers that the prover would have sent if the protocols $\protU{^\tc}$ were run explicitly.
Note that we would pick $\cksumT  = \tO(\secpar d) \ll {g\ell}$, 
so indeed the checksums can be represented using very few bits ($\tO(\secpar d \Flog) = \tO(\secpar^2 d)$ bits).

Finally, they let the resulting claim be $(\cSmid, \Cmid) \in \cL'_{\bm x}$ where $\cSmid = \cS^{\hyb}$ and $\Cmid$ is a circuit that takes in $\transcMat$ (purported to be $\transcMat^{\cS^{\hyb}}$) as input.
In additional to checking that rows (transcripts) in $\transcMat$ are indeed accepted by the base \UIP
and satisfy the original constraint $\C$,
the circuit also recomputes the checksums of the columns of $\transcMat$ to make sure they are equal to the checksums committed by the prover.
Due to the unique decoding guarantee of $\cksum$ (\Cref{lem:unique-decoding}),
if there exists $\transcMat^*$ in $\rowball(\transcMat^{\cSmid})$ with $\GKRcirc(\transcMat^*) = 1$,
the cheating prover would be implicitly committing to this $\transcMat^*$.
Therefore,
one can actually extract (via unique decoding) an accepting but deviating prover strategy for one of the base \UIP.
The base \UIP's unambiguity guarantees that any deviating strategy should only be accepted with probability at most $\epsilon$.
Therefore, with high probability,
there is no $\transcMat^* \in \rowball(\ippinput)$ that would make $\GKRcirc(\transcMat^*) = 1$,
i.e. distance is generated.

\paragraph{The Protocol.}
The protocol is given in \Cref{alg:phase1}.
\begin{algorithm}
    \setstretch{1.1}
    \caption{\prot{\Distance} \protphaseone.}
    \label{alg:phase1}
    
    \small
    \textbf{Input Parameters:} $\secpar, d \in \NN$ and a field $\FF$.\\
    \iterinput.\\
    \textbf{Other Parameters:} $g = \abs{\cS}$, $\K = g \ell = 2^\lgK$, and  $\cksumT = 2d(\lgK + \secpar) + \secpar \in \NN$.
    \\
    \textbf{Ingredients:} 
    \begin{itemize}[topsep=-0.02em,itemsep=-0.2em]
        \item \emph{A Transcript Checker \UIP, $\protU{^\tc}$ (\Cref{lem:Lxq}), for certifying $((x_i, \bm q), \transc) \in \cR$.}

        \item \emph{A Checksum Function, $\cksum_{\bm \cksumU}(\cdot)$, computing the checksum of a systematic error-correcting code with distance $2d$},
        where $\bm \cksumU = (\cksumU_1,\ldots, \cksumU_{\cksumT}) \in (\FF^\lgK)^{\cksumT}$ is a sequence of points in $\FF^{\lgK}$.
    \end{itemize}
    \textbf{Output:} 
    \itermidone{}.
    
    \begin{algorithmic}[1]
        \State $\cV_\Distance$ samples $\bm q' = (q'_1,\ldots,q'_\ell) \in (\bin^b)^\ell$.
        \State $\cV_\Distance$ samples $\bm \cksumU \gets_R (\FF^\lgK)^{\cksumT}$ and sends it to $\cP_\Distance$. 
        \For{$r=1, \ldots,\ell$}
        \State $\cV_\Distance$ sends $q_{r}'$ to $\cP_\Distance$. \Comment{Random coins are shared.}
        \State $\cP_\Distance$ receives $q_{r}'$ and does:
            \ifFOCS
            \State Compute 
            $\begin{aligned}[t]
                &\bm m_r^{i, \bm q} \\
                &\coloneqq \cP^\tc((x_i, \bm q), (\bm m_1^{i, \bm q},\ldots, \bm m_{r - 1}^{i, \bm q}))
            \end{aligned}$
            for all $(i,\bm q)\in \cS$,
            \else
            \State \hspace{1em} 1. Compute $\bm m_r^{i, \bm q} \coloneqq \cP^\tc((x_i, \bm q), (\bm m_1^{i, \bm q},\ldots, \bm m_{r - 1}^{i, \bm q}))$ for all $(i,\bm q)\in \cS$,
            \fi
            \State \hspace{1em} 2. By \Cref{rem:tauchecker},
            $\set{\bm m^{i, \b q}_r}_{(i, \bm q) \in \cS}$ exactly forms the rows of the matrix $\answerMat_r \coloneqq \answerMat_r^{\cS^\hyb} \in \bin^{(g\ell) \times a}$.
            \State \hspace{1em} 3. Compute $\cksumv_r\gets \cksum_{\bm \cksumU}(\answerMat_{r}) \in \bin^{\cksumT \times a}$, and send it to $\cV_\Distance$.
        \EndFor
        \State Define an arithmetic circuit $\GKRcirc(\transcMat)$,
        where $\transcMat \in \FF^{\K \times (a\ell)}$ is purported to be $\transcMat^{\cS^{\hyb}}$,
        and the circuit checks the following set of conditions.
        \Comment{Note that $\abs{\cS^\hyb} = \K = g \ell$.}
        \printchecksumchecks
        \State \Return $\lrag{\cS^{\hyb}}, \lrag{\GKRcirc}$.
    \end{algorithmic}
\end{algorithm}

\subsubsection{Analysis of the \Distance{} Protocol}
Recall that the input claim $(\cS, \C) \in \cL'_{\bm x}$ is true iff $(\C(\transcMat^{\cS}) = 1) \wedge (\bm x|_\I \in \cL^{\otimes \abs{\I}})$ where $\cS$ consists of $g \le \batchK$ pairs of the form $(i, \bm q) \in [\batchK] \times \bin^{b \ell}$, 
$\Flog = \tO(\secpar)$,
and $\C$ is a log-space uniform boolean circuit of size $\tO(g a\ell)$ and depth $\tO(1)$.

\paragraph{Prescribed Completeness.}
The honest prover $\cP_{\Distance}$ answers the random coin $\bm q'$ sampled by the verifier by running the Random Continuation \UIP $\prot{i, \bm q}$ in parallel on all $g = \abs{\cS}$ instances.
It honestly sends the checksums $\bm S$ that are computed on the columns of $\ippinput$,
where $\cS^\hyb$ is the set of hybrid pairs between $\cS$ and $q'$.
Recall that $\GKRcirc(\ippinput) = 1$ iff

\printchecksumchecks

The third condition is always true because the checksums are computed honestly on the columns of $\ippinput$.
On the other hand, the first condition is true iff $\C(\transcMat^{\cS}) = 1$,
while the second condition is true iff $\bm x|_\I \in \cL^{\otimes \I}$,
so, $\GKRcirc(\ippinput) = 1$ iff $(\cS, \C) \in \cL'_{\bm x}$.

\paragraph{Unambiguous Generation of $\Delta_c$-Distance.}
The verifier $\cV_\Distance$ sends the fresh random coins $\bm q' = (q_1', q_2', \ldots, q_\ell')$ round by round.
Let $\cP^*$ be a cheating strategy that deviates from $\cP_{\Distance}$.
For each $r \in [\ell]$, $\cP^*$'s behavior in the $r$-th round is fully determined by $\bm q_{\le r}' = (q_1',\ldots,q_r')$,
so let $\cksumv_r^*$ be the checksum produced by $\cP^*$ on $\bm q_{\le r}'$.
Let $\cksumv^* = (\cksumv_1^*,\ldots,\cksumv_\ell^*)$ be all the checksum sent.

Fix some $\bm q_{\le r^*}'$ such that $r^* \in [\ell]$ is indeed the first round when $\cP^*$ deviates,
i.e. when it first produces some $\cksumv_{r^*}^* \neq \cksum(\answerMat_{r^*})$.
As $r^*$ is minimal, $\cksumv_{r^*}^* \neq \cksum(\answerMat_{r^*})$ but $\cksumv_r^* = \cksum(\answerMat_r)$ for all $r < r^*$.
(Recall that $\answerMat_r$ is the $r$-th round chunks of $\ippinput$).
Note that this implies $\Cmid(\ippinput) = 0$ because $\ippinput$ does not have the checksum $\cksumv_{r^*}^*$ committed by the prover in round $r^*$.

Let $E$ be the event that after $\bm q'_{> r^*}$ is sampled, there exists $\transcMat^* \in \rowball(\ippinput) \setminus \set{\ippinput}$, 
such that $\GKRcirc(\transcMat^*) = 1$.
Given that we already showed $\Cmid(\ippinput) = 0$,
to prove generation of $\Delta_c$-distance 
(with probability at least $1-(\epsilon + 2^{-\secpar})$, for every $\transcMat^* \in \rowball(\ippinput)$, $\Cmid(\transcMat^*) = 0$)
we only need to show that $\Pr[E] \le \epsilon + 2^{-\secpar}$.
If $E$ is impossible on this $\bm q_{\le r^*}'$, then we are done,
so we consider the case when $\Pr[E] > 0$.

Since $\Pr[E] > 0$,
there exists some continuation $\bm q_{> r^*}'$ and $\transcMat^* \in \rowball(\ippinput) \setminus \set{\ippinput}$ such that $\GKRcirc(\transcMat^*) = 1$. 
In this case, for every $r \le r^*$,
if we let $\answerMat_r^*$ be the corresponding chunk of $\transcMat^*$ used in round $r$,
we have $\cksum_{\bm \cksumU}(\answerMat_r^*) =\bm \cksumv_r^*$ because of Condition 3 of $\GKRcirc$.

Furthermore, since $\cP^*$ is first cheating on round $r^*$,
we have $\cksumv_{r^*}^* = \cksum(\answerMat_{r^*}^*)\neq \cksum(\answerMat_{r^*})$,
and thus $\answerMat_{r^*}^* \neq \answerMat_{r^*}$. 
This means there exists a row-index $(i, \bm q^\hyb) \in \cS^\hyb$ such that $\answerMat_{r^*}^*[i, \bm q^\hyb] \neq \answerMat_{r^*}[i, \bm q^\hyb]$.
By \Cref{rem:tauchecker},
$(i, \bm q^\hyb)$ is some hybrid derived from some $(i, \bm q) \in \cS$.
Let $\cS^{i, \hyb} = \set{(i, \bm q^\hyb_1), \ldots, (i, \bm q^\hyb_\ell)}$ be the set of hybrid pairs derived from $(i, \bm q)$.

We construct a cheating prover $\cP^\tc_{*}$, that breaks the unambiguous soundness of the UIP $\protU{^\tc}$ as follows.
For each round $r \in [\ell]$, it does:
\begin{enumerate}
    \item Sample $\bm \cksumU \gets_R \bin^{\cksumT}$ uniformly at random and send it to $\cP^*$.
    \item Let $\bm q_{\le r}' = (q_1',\ldots,q_r')$ be the random coins sent by $\cV^\tc$ in the previous rounds,
    \item Apply $\cP^*$ on $\bm x|_\I$ and $\bm q_{\le r}'$ to obtain $\bm \cksumv_r^*$,
    \item Try to infer $\answerMat_r^*$, committed by $\cP^*$,
    by finding at most one $\answerMat_r^*$ that is $\Delta_c$-$d$-close to $\answerMat_r$ with checksum $\cksumv_r^*$.
    \item If such an $\answerMat_r^*$ exists and $\answerMat_r^* \in \bin^{\K \times a}$,
    send the message $\answerMat_{r}^*[\cS^{i, \hyb}, :]$ to $\cV^\tc$.
\end{enumerate}
\begin{claim}
Conditioned on unique decoding\footnote{Recall that unique decoding means for any pair $\answerMat',\answerMat'' \in \rowball(\ippinput)$, if $\cksum_{\bm \cksumU}(\answerMat') = \cksum_{\bm \cksumU}(\answerMat'')$ then $\answerMat' = \answerMat''$.} holding for the sampled $\bm \cksumU$,
when $\cV^\tc$ samples the random coin prefix $\bm q'_{\le r^*}$,
$\cP_*^\tc$ first deviates on the $r^*$-th round and convinces $\cV^\tc$ to accept with probability at least $\Pr[E]$.
\label{fact:break-soundness}
\end{claim}
\begin{proof}
First note that by the minimality of $r^*$ and the uniqueness of $\answerMat^*_r$ in each round,
$\cP^\tc_*$ deviates from the prescribed $\cP^\tc$ the first time exactly in round $r^*$.

Whenever $E$ occurs, 
there exists some completion $\transcMat' \in \rowball(\transcMat)$ such that $\GKRcirc(\transcMat') = 1$.
Let $(\answerMat'_1,\ldots,\answerMat'_\ncol)$ be the columns of $\transcMat'$.
Again by unique decoding,
$\cP^\tc_*$ are \emph{exactly} sending those $\answerMat_r'$ ($\answerMat'_r = \answerMat_r^*$) for all $r \in [\ell]$,
because their checksums are agreeing with $\cksumv^*$.
Furthermore,
$(\answer'_1[\cS^{i, \hyb}, :],\ldots,\answer'_\ell[\cS^{i, \hyb}])$ satisfies Conditions 2 and 4 of $\Cmid$,
so it is also accepted by $\cV^\tc$ (\Cref{lem:Lxq}).
\end{proof}
By the $\epsilon$-unambiguity of $\protU{^\tc}$,
the probability of $\cP_*^\tc$ breaking the unambiguity of $\protU{^\tc}$ is upper-bounded by $\epsilon$.
Therefore, $\Pr[E] - 2^{-\secpar} \le \epsilon$.
Note the loss of $2^{-\secpar}$ comes from the (rare) event that unique decoding does not hold for the initially sampled $\bm \cksumU$ (\Cref{lem:unique-decoding}).

The \emph{furthermore} part follows from that when $E$ does not happen, 
$\ippinput$ is the only candidate solution to $\GKRcirc(\transcMat) = 1$ in $\rowball(\ippinput)$.
This holds for the prescribed $\cP_\Distance$ as well.

\paragraph{Complexities.}
With $\tO$ hiding $\poly(\secpar) \cdot \polylog(\batchK, n, a, \ell)$ factors,
Note that we rely on the assumption that $\Flog = \tO(\secpar)$ to simplify the expressions.
\begin{itemize}
    \item $\ell_\Distance = O(\ell)$,
    \item $a_\Distance = \cksumT \cdot a \cdot \Flog = \tO(da)$.
    \item $b_\Distance = \max(b, \cksumT\Flog) = \tO(\max(b, \secpar^2d)) = \tO(d a)$, assuming $b \le a$.
    \item $\Ptime_\Distance = \tO(g \ell \Ptime + \cksumT (g \ell) (a\ell)\Flog + \abs{\lrag{\cS}} + \abs{\lrag{\C}} + \batchK \cdot n) = \tO(g \ell \Ptime + dga\ell^2 + \abs{\lrag{\cS}} + \abs{\lrag{\C}} + \batchK \cdot n)$.
    \item $\Vtime_\Distance = \tO(\cksumT \cdot a \cdot \Flog \cdot \ell + \abs{\lrag{\cS}} + \abs{\lrag{\C}} + \batchK \cdot n) = \tO(da\ell + \abs{\lrag{\cS}} + \abs{\lrag{\C}} + \batchK \cdot n)$.
\end{itemize}
The runtimes are less immediate, 
so we explain them below.
\begin{itemize}
    \item 
    For the prover runtime,
    note that is has to run $g = \abs{\cS}$ instances of the \UIP $\protU{^\tc}$ in parallel,
    which is in turn $g\ell$ instances of the base $\prot{}$,
    so that takes $O(g \ell \Ptime)$.
    It also computes the checksums (which can be done in $\tO(\cksumT (g \ell) (a\ell)\Flog)$ time using FFT),
    The term $\abs{\lrag{\cS}} + \abs{\lrag{\C}} + \batchK \cdot n$ comes from reading the input.
    \item
    The verifier does not perform any check in this sub-protocol,
    but it has to store the $\cksumT \cdot a$ points it receives every round,
    which requires $\tO(\cksumT \cdot a\cdot \Flog \ell)$ time.
    In the end, the time to compute $\lrag{\cSmid}$ is $O(\abs{\lrag{\cSmid}}) = O(\abs{\lrag S} + b \ell)$,
    and the time to compute $\lrag{\Cmid}$ is $O(\abs{\lrag{\Cmid}}) = \tO(\abs{\lrag \C} + \abs{\lrag \cS} + da\ell + \batchK \cdot n)$,
    since both of these consist of strings the verifier already has access to.
\end{itemize}

\paragraph{Circuit $G_{\Distance}$ for generating $\cSmid$.}
Let $\bm q'$ be the fresh random coin generated in \Cref{alg:phase1}.
Let $p \coloneqq \ceil{\log\batchK} + b\cdot \ell$ be the bit length of a pair $(i, \bm q) \in \bin^{\ceil{\log\batchK}} \times \bin^{b\cdot \ell}$.
We can represent $\cSmid$ as a subset of $\bin^p$.
\begin{claim}
    The set $\cSmid$ has a description $\lrag{\cSmid} = (\lrag{\cS}, \bm q')$, and $\abs{\lrag{\cSmid}} = \abs{\lrag{\cS}} + b\ell + O(1)$.
    \label{clm:gen}
\end{claim}
\begin{proof}
    Assume $\cS = \set{\cS_1, \ldots, \cS_g}$,
    where $\cS_s = (i^{s}, \bm q^{s})$ for every $s \in [g]$,
    then by the definition of the hybrid coins in \Cref{alg:phase1},
    \ifFOCS
    \begin{align*}
    \cSmid = \cS^\hyb = \begin{Bmatrix}
    \begin{aligned}
    &(i^{1}, (\bm q^{1})^\hyb_1), \ldots, (i^{1}, (\bm q^{1})^\hyb_\ell), \\
    &\ldots, (i^{g}, (\bm q^{g})^\hyb_1), \ldots, (i^{g}, (\bm q^{g})^\hyb_\ell)
    \end{aligned}
    \end{Bmatrix},
    \end{align*}
    \else
    \[\cSmid = \cS^\hyb = \set{(i^{1}, (\bm q^{1})^\hyb_1), \ldots, (i^{1}, (\bm q^{1})^\hyb_\ell), \ldots, (i^{g}, (\bm q^{g})^\hyb_1), \ldots, (i^{g}, (\bm q^{g})^\hyb_\ell)},\]
    \fi
    where $(\bm q^{s})^\hyb_r \coloneqq (\bm q^{s}_{\le r}, \bm q'_{> r})$ for every $s \in [g], r \in [\ell]$.

    Let $\K = \abs{\cSmid} = g \ell$,
    and $\lrag{\cSmid} = (\lrag{\cS}, \bm q')$.
    We define a Turing machine $\cM$ that takes in $(\K, 1^{\abs{\lrag{\cSmid}}}, 1^p)$ as input,
    and outputs the circuit $G_\Distance: [\K] \times \bin^{\abs{\lrag{\cSmid}}} \to \bin^p$.
    Since $\lrag{\cS}$ is a succinct description of $\cS$,
    there exists a Turing machine $\cM'$ that takes in $(g, 1^{\abs{\lrag{\cS}}}, 1^p)$ as input,
    and outputs the circuit $G: [g] \times \bin^{\abs{\lrag{\cS}}} \to \bin^p$,
    which, given $(s, \lrag{\cS})$ as input, outputs $\cS_s = (i^s, \bm q^s)$.

    $\cM$ first runs $\cM'$ to get $G$, 
    then moves onto defining $G_\Distance$, which, given $(s', \lrag{\cSmid}) \in [\K] \times \bin^{\abs{\lrag{\cSmid}}}$ as input,
    \begin{enumerate}
        \item First, computes $(g, r)$ such that $s' = s \cdot g + r$ where $s \in [g], r \in [\ell]$, e.g. by long division.
        \item Then, evaluates $G(s, \lrag \cS)$ to get $(i^s, \bm q^s)$.
        \item Finally, computes the corresponding hybrid coin $(\bm q^{s})^\hyb_r$ using $\bm q^s$ and $\bm q'$ (which is just copying the first $r$ bits of $\bm q^s$ and appending the remaining bits of $\bm q'$),
        and outputs $(i^{s}, (\bm q^{s})^\hyb_r)$.
    \end{enumerate}
    Indeed,
    \begin{itemize}
        \item $\size(G_\Distance) = \size(G) + b \ell + \tO(1)$.
        \item $\depth(G_\Distance) = \depth(G) + \tO(1)$.
    \end{itemize}
    The Turing machine $\cM$ uses $O(p) = O(\log \batchK + b \ell)$ space because it only needs to simulate $\cM'$ (which uses $O(p)$ space) and maintain pointers into the circuit $G_\Distance$ (which uses $O(p) + O(\log\batchK + \log b + \log \ell)$ space).
\end{proof}
\paragraph{Circuit $C_\Distance$ for $\Cmid$.}
By construction, $\GKRcirc(\transcMat)$ checks four conditions, 
whose circuit complexities are:
\begin{enumerate}
    \item Checking $\C$ on $\transcMat[\cS, :]$, requires $\size(\C)$ gates and $\depth(\C)$ depth.
    \item \distcirctwo

    This requires $O(\K \cdot S + \K\ncol) = \tO(g\ell S + dga\ell^2)$ gates and $D + O(\log(\K \ncol)) = D + \tO(1)$ depth (by a binary tree of \AND gates).
    \item \distcircthree

    This requires re-computing the checksums for every column of $\transcMat \in \bin^{\K \times \ncol}$.
    To compute one checksum, $\cksumT$ \LDE{} evaluation are needed.
    One \LDE evaluation requires $\tO(\K) = \tO(g\ell)$ gates,
    and $\tO(1)$ depth.
    Therefore, 
    we need $\tO(\cksumT \K \ncol) = \tO(dga\ell^2)$ gates and $\tO(1)$ depth to compute the checksums.
    Finally, we implement the consistency checks by equality tests.
    This requires an additional $\tO(\cksumT \ncol) = \tO(da\ell)$ gates and $\tO(1)$ depth.
    \item \distcircfour

    To check $\transcMat \in \GF(2)^{\K \times \ncol}$, 
    we just need to test if every element is in $\GF(2)$,
    and compute an \AND over them.
    This requires $\tO(\K\ncol) = \tO(g\ell)$ gates and $\tO(1)$ depth,
    because the test can be implemented by evaluating the polynomial $X^2 + X$ over $\FF$ and checking if the result is $0$.
\end{enumerate}
In order to perform the second check,
$C_\Distance$ also needs to first enumerate the pairs $\cS^{\hyb}$ explicitly using $\lrag{\cS^{\hyb}}$.
This incurs an additional overhead of $\depth(G_\Distance) +\tO(1) = \depth(G) + \tO(1)$ and $g \ell \cdot \size(G_\Distance) = \tO(g \ell \cdot \size(G) + gb\ell^2)$ (see \Cref{clm:gen}).

Therefore, we have the following overall bounds.
\begin{itemize}
    \item $\size(C_\Distance) = \tO(g\ell \cdot \size(G) + \size(C) + ga\ell S + dga\ell^2)$.
    \ifFOCS
    \item $\begin{aligned}[t]\depth(C_\Distance) &= \max(\depth(C), \\
        &\quad\quad\quad \depth(G) + D + \tO(1), \\
        &\quad\quad\quad \tO(1))+ \tO(1) \\
        &= \max(\depth(C), \\
        &\quad\quad\quad \depth(G) + D) + \tO(1).\end{aligned}$
    \else
    \item $\begin{aligned}[t]\depth(C_\Distance) &= \max(\depth(C), \depth(G) + D + \tO(1), \tO(1)) + \tO(1) \\
        &= \max(\depth(C), \depth(G) + D) + \tO(1).\end{aligned}$
    \fi
\end{itemize}

Similarly to \Cref{clm:gen},
the circuit $\Cmid$ also has a succinct description,
given by \[\lrag{\Cmid} = (\lrag{\cSmid}, \lrag{\C}, \cksumU, (\cksumv_1,\ldots, \cksumv_\ell), \bm x),\]
and it is indeed log-space uniform,
because its components (the four checks, as well as the circuit $C(\lrag{\cS^\hyb})$ that generates $\cS^\hyb$) are,
and the additional binary trees of $\AND$ as well as the equality checks are simple to describe.
The bit length $\abs{\lrag{\Cmid}} = \tO(\abs{\lrag\C} + \abs{\lrag{\cS}} + da\ell + \batchK \cdot n)$.

\subsection{Proof of \Cref{clm:phase2}: \protphasetwo}
\label{sec:phase2}
\subsubsection{Additional Ingredients}
\paragraph{Road map of our Construction.}
Recall that the intermediate claim is $\Cmid(\ippinput) = 1$.
We also have the \emph{distance guarantee} from \Cref{clm:phase1}: If $\Cmid(\ippinput) = 0$ then every matrix $\transcMat \in \rowball(\ippinput)$ is rejected by $\Cmid$.
Our goal is to output some reduced claim $(\cS', \C')$ such that $\abs{\cS'} \le \abs{\cSmid}/d$.

The \RR protocol provides a similar instance reduction, but only with respect to the Hamming distance.
\begin{theorem}[The original \RR Instance Reduction Protocol, informally rephrased; \cite{TCC:RotRot20}]
    \label{thm:RR-informal}
    Suppose $\K \in \NN$.
    There exists a protocol $\prot{\RR}$,
    where both the prover and verifier gets the description $\lrag{\C}$ of a log-space uniform circuit $\C : \bin^\K \to \bin$,
    and the prover gets the auxiliary input $\bm w \in \bin^{\K}$,
    such that if both parties gets the parameter $d = d(n) \in \NN$,
    then the protocol outputs a succinct description $\lrag{\RRS}$ of a set $\RRS \subset [\K]$ and a new circuit $\C' : \bin^{\abs{\RRS}} \to \bin$,
    such that $\abs{\RRS} \le \frac{\abs{\K}}{d}$.
    The protocol is sound in the following sense: if $\C(\bm w) = 0$ for every $\bm w$ that is (Hamming) $d$-close to $\bm w$, 
    then $\C'(\bm w|_\RRS) = 0$.
    It has $\tO(1)$ rounds and communication complexity $\tO(d)$.
    The prover runtime is $\poly(\K)$ and the verifier runtime is $\tO(d)$.
\end{theorem}

As explained in \Cref{sec:overview:RR},
$\rowball(\ippinput)$ is a strict subset of the Hamming ball.
The soundness of \Cref{thm:RR-informal} applies only when \emph{every Hamming-close string} is rejected,
which is not always true in our case because our notion of distance is different. 

\paragraph{The \RR Instance Reduction Template}
\Cref{thm:RR-informal} is originally proven by the following two steps:
\begin{enumerate}
    \item Use the \GKR protocol (\Cref{lem:GKR}) to transform the input claim $\C(\bm w) = 1$ into the claim $\bm w \in \pval(\pvalU, \bm \pvalv)$.
    \label{step:RR1}
    The parameters are chosen to ensure that \emph{distance is preserved}:
    if every $\bm w'$ that is $d$-close to $\bm w$ is rejected by $\C$,
    then $\bm w$ is $d$-far from $\pval(\pvalU, \bm \pvalv)$.

    \item Apply a special-purpose protocol for reducing $\pval$ instances,
    on the input $\bm w \in \pval(\pvalU, \bm \pvalv)$.
    This outputs a subset of coordinates $\RRS \subset [\K]$ that specifies the positions of the coordinates of $\bm w$ to read,
    and $\abs{\RRS} \le \abs{\K}/d$,
    along with a predicate $\C'$,
    such that $\C'(\bm w|_\RRS) = 1$ only if $\bm w$ is Hamming-$d$-close to $\pval(\pvalU, \bm \pvalv)$.
    \label{step:RR2}
\end{enumerate}
It turns out that the first step can be modified to work for $\Delta_c$-distance quite straightforwardly (\Cref{lem:RVW}).
However, for the second step,
we have to make lower level modifications in order to make it work for our notion of distance.
This discussion appears in \Cref{sec:DcRR}.

In the rest of this section,
we formally recall the guarantees of the two sub-protocols (\Cref{lem:RVW,lem:DcRR}),
corresponding to the two steps mentioned above,
and use them to prove \Cref{clm:phase2}.

\paragraph{The \GKR Protocol is Distance-Preserving}
Let $\ncol \coloneqq a \ell$.
To preserve the distance,
we choose the parameter $\pvalT = \tilde \Theta(d\ncol)$ (see \Cref{lem:RVW}).
In the \GKR protocol,
both parties take as input the description $\lrag \Cmid$,
and the prover additionally has access to $\ippinput$.
After running the \GKR protocol with parameter $\pvalT$,
the verifier obtains $\bm \pvalU \in (\FF^\lgK)^\pvalT, \bm \pvalv \in \FF^\pvalT$ as output,
with the guarantee that $\pval(\bm \pvalU, \bm \pvalv)$ is $\Delta_c$-$d$-close to $\Cmid(\ippinput)$ iff $\ippinput$ is $\Delta_c$-$d$-close to $\Cmid(\ippinput)$.\footnote{Recall that a matrix $\transcMat$ is $\Delta_c$-$d$-close to $\ippinput$ iff $\transcMat \in \rowball(\ippinput)$ (\Cref{def:Deltac-dist}).}

Recall that $D_\Distance \coloneqq \depth(\Cmid)$, $S_\Distance \coloneqq \size(\Cmid)$.
The complexity of the protocol is as follows, where $\tO$ hides $\polylog(\Flog, \K, \ncol)$ factors:
\begin{itemize}
    \ifFOCS
    \item $\epsilon_\GKR = \frac{O(D_\Distance\log S_\Distance)}{\abs{\FF}} + \left(O(\frac{D_\Distance\log S_\Distance}{\abs{\FF}})\right)^\pvalT \left(\binom{\K}{d} \abs{\FF}^{d}\right)^\ncol < 2^{-\secpar + 2}$, so long as $O(\frac{D_\Distance\log S_\Distance}{\abs{\FF}}) < 2^{-\secpar-3}$.
    \else
    \item $\epsilon_\GKR = \frac{O(D_\Distance\log S_\Distance)}{\abs{\FF}} + \left(O(\frac{D_\Distance\log S_\Distance}{\abs{\FF}})\right)^\pvalT \left(\binom{\K}{d} \abs{\FF}^{d}\right)^\ncol < 2^{-\secpar + 2}$, so long as $O(\frac{D_\Distance\log S_\Distance}{\abs{\FF}}) < 2^{-\secpar-3}$.
    \fi
    \item $\ell_\GKR = O(D_\Distance \log S_\Distance)$.
    \item $a_\GKR = b_\GKR = \tO(\pvalT\Flog) = \tO(da\ell)$.
    \item $\Ptime_\GKR = \tO(\pvalT \cdot \poly(S_\Distance) \cdot \Fbits) = \tO(da\ell) \cdot \poly(S_\Distance)$.
    \item $\Vtime_\GKR = \tO(\pvalT \cdot D_\Distance \log S_\Distance \cdot \Flog + \pvalT \cdot \abs{\lrag{\Cmid}}\Flog) = \tO(da\ell \cdot (D_\Distance \log S_\Distance + \abs{\lrag{\Cmid}}))$.
\end{itemize}
And the verifier's verdict circuit satisfies:
\begin{itemize}
    \item $\size(\cV_\GKR) = \tO(\pvalT \cdot D_\Distance \log S_\Distance + \pvalT \cdot \abs{\lrag{\Cmid}}\Flog) = \tO(da\ell \cdot (D_\Distance \log S_\Distance + \abs{\lrag{\Cmid}}))$.
    \item $\depth(\cV_\GKR) = \tO(1)$.
\end{itemize}

\paragraph{The Instance Reduction Protocol for \pval with $\Delta_c$-distance} 
$\prot{\RRrow}$ is a special-purpose protocol for reducing $\pval$ instances with $\Delta_c$ distances.
Both parties in the protocol take as input the parameters $\pvalU \in (\FF^\lgK)^{\pvalT}, \bm \pvalv \in \FF^\pvalT$ that specify $\pval(\pvalU, \bm \pvalv)$,
and the prover additionally has the input $\ippinput$.
The protocol outputs $\param{'}$ such that $\abs{\cS'} \le \abs{\cSmid}/d$,
with the guarantee that $(\C', \cS') \in \cL'_{\bm x}$ iff $\ippinput$ is $\Delta_c$-$d$-close to $\pval(\pvalU, \bm \pvalv)$.

The following is a restatement of \Cref{thm:DcRR-main}.
Its proof appears in \Cref{sec:DcRR}.
\begin{lemma}[The Instance Reduction Protocol for \pval with $\Delta_c$-distance (\Cref{thm:DcRR-main} in \Cref{sec:DcRR})]
    \label{lem:DcRR}
    \DCMainStmt
\end{lemma}

\subsubsection{Construction of the \Reduce{} Protocol}
We give the construction in \Cref{alg:Reduce}.

\begin{algorithm}
    \setstretch{1.1}
    \caption{\prot{\Reduce} \protphasetwo{}}
    \label{alg:Reduce}
    
    \small
    \textbf{Input Parameters:} $\secpar, d \in \NN$, and a field $\FF$ with $\abs{\FF} \ge \FboundconcreteT$ (where $C$ is some constant).\\
    \textbf{Input:} \itermidone{}.\\
    \textbf{Other Parameters:} $\K, \ncol \in \NN$, $\pvalT = \Tbound$.\\
    \textbf{Ingredients:} 
    \begin{itemize}[topsep=-0.02em,itemsep=-0.2em]
        \item \emph{The \GKR protocol} $\prot{\GKR}$ (\Cref{lem:GKR}) for checking $\C(\transcMat) = 1$.

        The prover and verifier gets the description $\lrag \C$,
        while the prover additionally gets the input $\transcMat$.
        Both parties additionally takes in the parameter $\pvalT \in \NN$,
        and in the end outputs $\pvalU \in (\FF^{\lgK + \log \ncol})^\pvalT, \pvalv \in \FF^\pvalT$.

        \item \emph{The Instance Reduction Protocol for \pval} $\prot{\RRrow}$ (\Cref{lem:DcRR}) for checking that $\transcMat$ is $\Delta_c$-$d$-close to $\pval(\bm \pvalU, \bm \pvalv)$.

        Both the prover and verifier gets the input $\pvalU, \bm \pvalv$,
        and the prover additionally gets the input $\transcMat$.
        The protocol outputs \itermidthree{}.
    \end{itemize}
    \textbf{Output:} \iteroutput.
    
    \begin{algorithmic}[1]
        \State $\cP_{\Reduce}$ compute the prescribed $\ippinput$.
        \State 
        Both parties run $\prot{\GKR}$ (\Cref{lem:GKR}) with $\pvalT = \Tbound$ on the claim \[\Cmid(\ippinput) = 1.\]
        The verifier obtains $\pvalU, \bm \pvalv$.

        \State Both parties run the Instance Reduction Protocol for \pval,
        $\prot{\DcRR}$ (\Cref{lem:DcRR}), 
        on the claim \[\ippinput \in \pval(\pvalU, \bm \pvalv),\]
        obtaining \itermidthree{}.

        \State Let $\lrag{\cS'} = (\lrag{\cSmid}, \lrag{\RRS})$.
        \Comment{Implicitly define $\cS' \coloneqq \cSmid[\RRS, :]$ and thus $\abs{\cS'} = \abs{\RRS} \le \abs{\cSmid}/d$.}

        \State \Return $\param{'}$.
    \end{algorithmic}
\end{algorithm}

\subsubsection{Analysis of the \Reduce{} Protocol}
We first verify that the preconditions for $\prot{\RRrow}$ are satisfied. 
Indeed, $d > \dboundm$, and $\abs{\FF} \ge \FboundT$ are assumptions of \Cref{clm:phase2}.
By \Cref{lem:DcRR}, $\abs{\RRS} \le \ceil{8\secpar \cdot \frac{\K}{d}} = O(\secpar\frac{\K}{d})$.

\paragraph{Prescribed Completeness.}
By the prescribed completeness of $\prot{\GKR}$,
$\ippinput \in \pval(\pvalU, \bm \pvalv)$ iff $\Cmid(\ippinput) = 1$ (which is equivalent to $(\cSmid, \Cmid) \in \cL'_{\bm x}$ by the construction of $\Cmid$ in \Cref{clm:phase1}),
and then by the prescribed completeness of $\prot{\RRrow}$,
$\C'(\ippinput[\RRS, :]) = 1$ iff $\ippinput \in \pval(\pvalU, \bm \pvalv)$.
Finally,
$\cS' = \cSmid[\RRS, :]$,
so $(\cS', \C') \in \cL'_{\bm x}$ iff $(\cSmid, \Cmid) \in \cL'_{\bm x}$.

\paragraph{Unambiguity.}
With the same proof as in \Cref{lem:pval}, 
with all but $2^{-\secpar-1}$ probability, 
the following claim holds,
so the unambiguity of $\prot{\RRrow}$ applies. 
\begin{claim}
With probability at least $1 - 2^{-\secpar-1}$, $\Delta(\pval(\bm \pvalU,\bm 0)) \geq 4d.$
\end{claim}
\begin{proof}
    \prot{\GKR} always outputs $\pvalT$ truly random points in $(\KK^{\lgK + \log \ncol})^{\pvalT}$.
    The number of points with $\Delta_c$ distance at most $4d$ from the origin is at most $(\binom{\K}{4d} \cdot \abs{\FF}^{4d})^\ncol$.
    Fixing such a point $x \in \bin^{\K \times \ncol}$,
    since every $\pvalU$ is uniformly random,
    by \Cref{lem:SZ},
    it is a root of the \LDE of $x$ with probability at most $\frac{m}{\abs{\FF}}$.
    Therefore,
    the probability $x \in \pval(\bm \pvalU, \bm \pvalv)$ is at most $\left(\frac{m}{\abs{\FF}}\right)^{\pvalT}$.
    Taking a union bound over all such points,
    and using that $\pvalT \ge \Tbound$,
    \ifFOCS
    \begin{align*}
    &\Pr[\Delta_c(\pval(\bm \pvalU, \bm \pvalv)) \ge 4d] \\
    &\le \left(\binom{\K}{4d} \cdot \abs{\FF}^{4d}\right)^\ncol \cdot \left(\frac{m}{\abs{\FF}}\right)^{\pvalT} \le 2^{-\secpar-1}.\qedhere
    \end{align*}
    \else
    \[
    \Pr[\Delta_c(\pval(\bm \pvalU, \bm \pvalv)) \ge 4d] \le \left(\binom{\K}{4d} \cdot \abs{\FF}^{4d}\right)^\ncol \cdot \left(\frac{m}{\abs{\FF}}\right)^{\pvalT} \le 2^{-\secpar-1}.\qedhere
    \]
    \fi
\end{proof}
Let $\cP^*$ be a cheating prover that deviates from $\cP_{\Reduce}$,
then it must be deviating in either of the two protocols.
By the unambiguity of $\prot{\GKR}$ and $\prot{\RRrow}$ and the last claim,
the probability that the prover produces $\param{'}$ such that $\C'(\cS') = 1$ is at most $\epsilon_\GKR + \epsilon_\RRrow + 2^{-\secpar-1} \le 2^{-\secpar}$.

\paragraph{Complexities.}
The \GKR protocol has $\ell_\GKR = O(D_\Distance \log S_\Distance)$ rounds.
We can reduce the message length per-round ($a_\GKR$) by increasing the number of rounds to $\ell_{\GKR}' = \tO(\ell \cdot D_\Distance \log S_\Distance)$,
which changes its complexity to $a_{\GKR}' = b'_{\GKR} = \tO(da)$.
Similarly,
for $\prot{\RRrow}$,
we increase its number of rounds to $\ell_{\RRrow}' = \tO(\ell)$,
and its message lengths reduce to $a_{\RRrow}' = b_{\RRrow}' = \tO(da + \poly(d))$.
This is a simplifying assumption to make per-round messages have comparable lengths as in $\prot{\Distance}$.
\begin{itemize}
    \item $\ell_{\Reduce} = \ell_{\GKR}' + \ell_{\RRrow}' = \tO(\ell \cdot D_\Distance \log S_\Distance).$
    \item $a_{\Reduce} = b_{\Reduce} = \max(a_{\GKR}', a_{\RRrow}') = \tO(da + \poly(d))$.
    \ifFOCS
    \item $\begin{aligned}[t]
        \Ptime_{\Reduce} &= \Ptime_{\GKR} + \Ptime_{\RRrow} \\
        &= \tO(d \ncol) \cdot \poly(S_\Distance) \\
        &\quad + \poly(\K \ncol, \pvalT \cdot \Flog)\\
        &= \poly(\Flog, \K, \ncol, d, S_\Distance).
    \end{aligned}$
    \item $\begin{aligned}[t]
    \Vtime_{\Reduce} &= \Vtime_{\GKR} + \Vtime_{\RRrow}\\
    &\quad + \tO(\batchK \cdot n + \abs{\lrag{\cSmid}} + \abs{\lrag{\Cmid}})\\
    &\le \tO(d a\ell \cdot (D_\Distance \log S_\Distance +\abs{\lrag{\Cmid}}) \\
    &\quad + \poly(d) + \abs{\lrag{\cSmid}} + \batchK \cdot n).
    \end{aligned}$
    \else
    \item $\begin{aligned}[t]
    \Ptime_{\Reduce} &= \Ptime_{\GKR} + \Ptime_{\RRrow}\\
    &= \tO(d \ncol) \cdot \poly(S_\Distance) + \poly(\K \ncol, \pvalT \cdot \Flog)\\
    &= \poly(\Flog, \K, \ncol, d, S_\Distance).
    \end{aligned}$
    \item $\begin{aligned}[t]
    \Vtime_{\Reduce} &= \Vtime_{\GKR} + \Vtime_{\RRrow} + \tO(\batchK \cdot n + \abs{\lrag{\cSmid}} + \abs{\lrag{\Cmid}})\\
    &\le \tO(d a\ell \cdot (D_\Distance \log S_\Distance +\abs{\lrag{\Cmid}}) + \poly(d) + \abs{\lrag{\cSmid}} + \batchK \cdot n).
    \end{aligned}$
    \fi
    Note that the third term in the first summation is due to the overhead of reading the input $(\cSmid,\bm x)$.
\end{itemize}

\paragraph{Circuit $G'$ for generating $\cS'$.}
The sub-protocol \prot{\RRrow} outputs the description $\lrag \RRS$.
Our new description $\lrag {\cS'} = (\lrag{\cSmid}, \lrag{\RRS})$.
We represent every pair $(i, \bm q)$ as elements in $\bin^p$,
where $p = \ceil{\log \batchK} + b \cdot \ell$.
Let $g' = \abs{\cS'}$ be the number of pairs to output.
The Turing machine $\cM$ takes in $(g', 1^{\abs{\lrag{\cS'}}}, 1^p)$,
and outputs the circuit $G':[g'] \times \bin^{\abs{\lrag{\cS'}}} \to \bin^p$.
$G'(\lrag{\cS'})$ enumerates $\cS'$ as follows:
for every $s \in [g']$, 
\begin{itemize}
    \item Compute the $s$-th element in $\RRS \in [\K]$ (using $G_\RRS(s, \lrag{\RRS})$), and let it be $\RRS_s \in [\K]$.
    \item Output the $\RRS_s$-th pair in $\cSmid$ (using $G_\Distance(\RRS_s, \lrag{\cSmid})$).
\end{itemize}
The circuit satisfies:
\begin{itemize}
    \item $\size(G') = \size(G_\RRS) + \size(G_\Distance) = \size(G_\Distance) + \tO(\poly(d))$.
    \item $\depth(G') = \depth(G_\Distance) + \tO(1)$.
\end{itemize}
The bit length satisfies $\abs{\lrag{\cS'}} = \tO(\poly(d) + \abs{\lrag{\cSmid}})$.
$G'$ is $p$-space uniform because $G_{\RRS}$ is $p$-space uniform and $\cM$ only needs to maintain additional pointers into the circuit $G'$.

\paragraph{Circuit $C'$ for checking $\C'$.}
Note that the sub-protocol \prot{\RRrow} already outputs the description $\lrag{\C'}$.
By \Cref{lem:DcRR},
$\abs{\lrag{\C'}} = \tO(da\ell + \poly(d))$,
and the complexity of $C'(\transcMat, \lrag{\C'}) = \C'(\transcMat)$ is the following:
\begin{itemize}
    \item $\size(C') = \tO(\frac{\K}{d} \cdot a\ell)$,
    \item $\depth(C') = \tO(1)$.
\end{itemize}
\paragraph{Verdict Circuit $\cV_{\Reduce}$.}
The verdict circuit $\cV_{\Reduce}$ does not reject iff neither $\cV_\GKR$ nor $\cV_{\RRrow}$ rejects.
Therefore,
the circuit that implements this check would satisfy:
\begin{itemize}
    \ifFOCS
    \item $\begin{aligned}[t]
   \size(\cV_{\Reduce}) &= \tO(\size(\cV_\GKR) + \size(\cV_{\RRrow}) \\
   &\quad\quad\quad +\abs{\lrag{\cSmid}} + \abs{\lrag{\Cmid}} + \batchK \cdot n) \\
   & = \tO(d a \ell(D_\Distance \log S_\Distance + \abs{\lrag{\Cmid}}) \\
   &\quad\quad\quad + \poly(d) + \abs{\lrag{\cSmid}} + \batchK \cdot n).
   \end{aligned}$
    \else
    \item $\begin{aligned}[t]
   \size(\cV_{\Reduce}) &= \tO(\size(\cV_\GKR) + \size(\cV_{\RRrow}) + \abs{\lrag{\cSmid}} + \abs{\lrag{\Cmid}} + \batchK \cdot n) \\
   & = \tO(d a \ell(D_\Distance \log S_\Distance + \abs{\lrag{\Cmid}}) + \poly(d) + \abs{\lrag{\cSmid}} + \batchK \cdot n).
   \end{aligned}$
    \fi
    \item $\begin{aligned}[t]
        \depth(\cV_{\Reduce}) &= \max(\depth(\cV_\GKR), \depth(\cV_{\RRrow}))\\
        &= \tO(1).
    \end{aligned}$
\end{itemize}

\subsection{Analysis of Other Sub-protocols}
\label{sec:other-analysis}
We present of the deferred sub-protocols and analyze their properties in this section. 

\subsubsection{Proof of \Cref{lem:Lxq}: the Random Continuation Protocol}
\label{sec:part2}
The protocol $\protU{^\tc}$ is formally given in \Cref{alg:iq}.

\paragraph{Prescribed Completeness.} This is immediate from the prescribed completeness of $\prot{}$.

\paragraph{Unambiguity.}
We consider the following two cases.
\begin{itemize}
    \item If the input $\transcMat = \transcMat^{x, \bm q}$,
    let $r^* \in [\ell]$ be the first round when $\cP^*$ deviates from $\cP^{\tc}$.
    Then $\cP^*$ sends a non-prescribed message $\tilde \answer^{x, \bm q^\hyb_j}_{r^*} \neq \answer^{x, \bm q^\hyb_j}_{r^*}$ for some $j \in [\ell]$.
    With probability at most $\epsilon$ over the remaining interaction,
    $\cV(x_i, \bm q^\hyb_j, \tilde \transc^{x, \bm q^\hyb_j}) = 1$
    by the unambiguity of $\prot{}$,
    so $\cP^*$ would be rejected by $\cV^\tc$ with probability at least $1-\epsilon$.
    
    \item If the input $\transcMat \neq \transcMat^{x, \bm q}$,
    let $j^* \in [\ell]$ be the first round for which $\transcMat_{j^*} \neq \transcMat^{x, \bm q}_{j^*}$.
    Since $\cV^\tc$ checks that $\transcMat$ and $\transcMat^{x, \bm q^\hyb_{j^*}}$ share a prefix of length $j^* \cdot a$,
    $\cV^\tc$ would reject at once if it is not the case that $\transcMat_{j^*} = \transcMat^{x, \bm q^\hyb_{j^*}}_{r}$ for every $1 \le r \le j^*$.
    Therefore, we assume $\transcMat_{j^*} = \transcMat^{x, \bm q^\hyb_{j^*}}_{r}$ for every $1 \le r \le j^*$.
    This means that $\transcMat^{x, \bm q^\hyb_{j^*}}$ is the transcript of an interaction where $\cP^*$ deviates firstly in round $j^*$ of the base protocol.
    By the unambiguity of $\prot{}$, 
    with probability at most $\epsilon$ over the remaining interaction starting from the $j^*$-th round,
    $\transcMat^{x, \bm q^\hyb_{j^*}}$ would be rejected,
    which means $\cP^*$ would be rejected by $\cV^\tc$ with probability at least $1-\epsilon$.
\end{itemize}
\paragraph{Complexities.}
The number of rounds is $\ell$.
The prover message is $a^\tc = \ell \cdot a$,
and the verifier message is $b^\tc = b$.
The prover time is $\Ptime^\tc = \ell \cdot \Ptime$,
and the verifier time is $\Vtime^\tc = \Vtime + \ell \cdot \size(\cV)$.

\paragraph{Transcript Check Structures.}
By inspection, indeed the $j$-th message of $\cP^\tc$ is \[
\bm m_r = \left((\answerMat^{x, \bm q^\hyb_j})_{r \in [\ell]}\right)^\top,
\] and the verifier $\cV^\tc$ simply checks all $\transcMat^{x, \bm q^\hyb_j}$ are accepted,
and that $\transcMat^{x, \bm q}$ and $\transcMat^{x, \bm q^\hyb_j}$ share a prefix of length $j \cdot a$.

\begin{algorithm}
    \setstretch{1.1}
    \caption{$\protU{^\tc}$ Random Continuation.}
    \label{alg:iq}
    
    \small
    \textbf{Input:} $x \in \bin^n$, $\bm q \in (\bin^b)^\ell$, and a transcript $\transc \in \bin^{a \times \ell}$ purported to be $\transc^{x, \bm q}$.  \\
    \textbf{Output:} the verifier decides to accept or reject the claim that $\transc$ is indeed $\transc^{x, \bm q}$, and that $x \in \cL$.
    \begin{algorithmic}[1]
        \State $\cV^\tc$ samples $\bm q' = (q'_1,\ldots,q'_\ell) \in (\bin^b)^\ell$.
        \For{$r=1, \ldots,\ell$}
            \State $\cV^\tc$ sends $q_{r}'$ to $\cP^\tc$.
            \For{$j = 1, \ldots, \ell$}
                \State $\cP^\tc$ let $\answer_r^{x, \bm q^\hyb_j} \gets \cP(x, (\bm q^\hyb_j)_{\le r})$, where $\bm q^\hyb_j = (q_1,\ldots,q_{j}, q'_{j + 1}, \ldots, q'_{\ell})$.
            \EndFor
            \State $\cP^\tc$ sends $\set{\answer_r^{x, \bm q^\hyb_j}}_{j \in [r]}$ to $\cV^\tc$.
        \EndFor
        \State $\cV^\tc$ let $\transc^{x, \bm q^\hyb_j} \gets (\answer_1^{x, \bm q^\hyb_j}, \ldots, \answer_\ell^{x, \bm q^\hyb_j})$ for every $j \in [\ell]$.
        \State $\cV^\tc$ accepts iff $\bigwedge_{j \in [\ell]}\left[\cV(x, \bm q^\hyb_j, \transc^{x, \bm q^\hyb_j}) = 1\right]$ and that for every $j \in [\ell]$, $\transc$ and $\transc^{x, \bm q^\hyb_{j}}$ share a prefix of length $j \cdot a$.
    \end{algorithmic}
\end{algorithm}

\subsubsection{Proof of \Cref{lem:final-check}: the Explicit \UIP for the associated language}
\label{sec:final-check}
We construct $\prot{\cL'_{\bm x}}$, an explicit \UIP for the associated language $\cL'_{\bm x}$,
which takes as input $\bm x$ and $\param{}$
and verifies the claim $(\cS, \C) \in \cL'_{\bm x}$.
It consists of the following two steps:
\begin{enumerate}
    \item $\cP_{\cL'_{\bm x}}$ sends the prescribed transcript matrix $\transcMat^\cS \in \bin^{\abs{\cS} \times (a\ell)}$ to $\cV_{\cL'_{\bm x}}$.
    \item $\cV_{\cL'_{\bm x}}$ enumerates $G(s, \lrag \cS)$ for every $s \in [\abs{\cS}]$.
    \item Both parties execute $\protU{^\tc}$ explicitly on the input $((i, \bm q), \transc^\cS[(i, \bm q), :])$, for every $(i, \bm q) \in \cS$ in parallel.

    By \Cref{rem:tauchecker},
    the messages sent by $\cP_{\cL'_{\bm x}}$ constitute the transcript matrix $\transcMat = \transcMat^{\cS^\hyb}$,
    where $\cS^\hyb$ is the set of hybrid pairs in constructed between $\cS$ and $\bm q' \in \bin^{b\ell}$, the random coin sampled by $\cV_{\cL'_{\bm x}}$.

    \item $\cV_{\cL'_{\bm x}}$ accepts iff every execution of $\protU{^\tc}$ in $\transcMat$ is accepted, 
    $G(\transcMat^\cS, \lrag \C) = 1$ and $\transcMat[\cS,:] = \transcMat^\cS$.
\end{enumerate}
Let $\I \coloneqq \set{(i, \bm q): (i, \bm q) \in \cS}$.
Recall that $(\cS, \C) \in \cL'_{\bm x}$ iff $(\C(\transcMat^\cS) = 1) \wedge (\bm x|_{\I} \in \cL^{\otimes \abs{\cS}})$.
\paragraph{Prescribed Completeness.}
$\cP_{\cL'_{\bm x}}$ sends the verifier the prescribed $\transcMat^\cS$ in the first step,
and executes $\cP_{\cL'_{\bm x}}$ honestly for every $(i, \bm q) \in \cS$ in parallel,
so $\transcMat = \transcMat^{\cS^\hyb}$ (and thus $\transcMat[\cS,:] = \transcMat^\cS$).
Since $\transcMat^\cS$ is the prescribed transcript matrix,
all executions of $\protU{^\tc}$ in $\transcMat$ are accepted iff $\bm x|_{\I} \in \cL^{\otimes \abs{\cS}}$.
Therefore,
the verifier's check passes exactly when $(\cS, \C) \in \cL'_{\bm x}$.

\paragraph{Unambiguity.}
Suppose that $\cP^*$ deviates from $\cP_{\cL'_{\bm x}}$ in the first step,
i.e. it sends some $\transcMat^* \neq \transcMat^\cS$ to $\cV_{\cL'_{\bm x}}$.
This means for some row-index $(i, \bm q) \in \cS$, $\transcMat^*[(i, \bm q)] \neq \transcMat^{i, \bm q}$.
If $\cP^*$ does not deviate in the rest of the protocol,
then by prescribed completeness of $\protU{^\tc}$ (\Cref{def:tauchecker}),
$\cV_{\cL'_{\bm x}}$ rejects on the input $((i, \bm q), \transcMat^*[(i, \bm q)])$,
and thus $\cV_{\cL'_{\bm x}}$ rejects.
Otherwise, $\cP^*$ deviates in one of the $\ell$ rounds while executing $\protU{^\tc}$ in parallel.
By the $\epsilon$-unambiguity of $\protU{^\tc}$ (\Cref{def:tauchecker}),
that execution of $\protU{^\tc}$ is accepted with probability at most $\epsilon$,
so $\cV_{\cL'_{\bm x}}$ accepts with probability at most $\epsilon$.

\paragraph{Complexities.}
Let $g = \abs{\cS}$.
The number of rounds is $\ell_{\cL'_{\bm x}} = \ell + 1$,
The prover message (per round) is of length $a_{\cL'_{\bm x}} = g\ell \cdot a$.
This can be seen by noting that the message in the first step is of length $g \cdot a\ell$ and the rest are of length $g \cdot a\ell$ per round as well (since $g$ instances of the protocol in \Cref{def:tauchecker} are in parallel).
The verifier message (per round) is of length $b_{\cL'_{\bm x}} = b$.
The prover runtime is $\Ptime_{\cL'_{\bm x}} = O(g\ell \cdot \Ptime)$,
because it simply runs $g$ instances $\cP^\tc$ in parallel (each of which has runtime $\ell \cdot \Ptime$),
and the verifier runtime is $\Vtime_{\cL'_{\bm x}} = O(g\ell \cdot \Vtime + g\size(G) + \size(C))$.
This is because it first computes $G(s, \lrag \cS)$ for every $s \in [g]$,
then runs $g$ instances of $\cV^\tc$ in parallel (each of which has runtime $\ell \cdot \Vtime$),
and finally evaluates $C(\transcMat^\cS, \lrag \C)$ (which takes time $\size(C)$).

The verifier's decision circuit $\cV_{\cL'_{\bm x}}$ has the following complexities:
\begin{itemize}
    \item $\size(\cV_{\cL'_{\bm x}}) = \tO(g \cdot \ell \cdot \size(\cV) + g\size(G) + \size(C))$.
    \item $\depth(\cV_{\cL'_{\bm x}}) = \max(\depth(G) + D, \depth(C)) + \tO(1)$.
\end{itemize}

\section{Doubly-Efficient Interactive Proofs}
\label{sec:app}
\subsection{A New \UIP for Verifying Space-Bounded Computations}
We demonstrate the power of \bUIP by constructing a doubly-efficient interactive proof for delegating space-bounded computations.
Denote by $\TS$ the class of languages decidable by a Turing machine in time $T(n)$ and space $S(n)$.
\begin{theorem}[Doubly-efficient Interactive Proof for $\TS$]
    \label{thm:TS}
    For any $T(n) \le n^{O\left(\left(\frac{\log n}{\log \log n}\right)^{1/2}\right)}$ and $S(n) = \poly(n)$,
    any $\secpar=\secpar(n)= \polylog(n)$,
    any $\delta=\delta(n) \in (0, 1)$,
    and any language $L \in \TS$,
    there exists a \UIP with the following properties:
    \begin{itemize}
        \item The soundness error is $2^{-\secpar} \cdot (1/\delta)\cdot (\log n)^{O(1/\delta)}$.
        \item The number of rounds is $(\log n)^{O(1/\delta)}$.
        \item The prover and verifier message length per-round is at most $T^\delta \cdot S(n) \cdot (\log n)^{O(1/\delta^2)}$.
        \item The prover runtime is $T(n) \cdot \poly(n) \cdot (\log n)^{O(1/\delta^2)}$.
        \item The verifier runtime is $T^{O(\delta)} \cdot S(n)^2 \cdot \poly(n) \cdot (\log n)^{O(1/\delta^2)}$.
    \end{itemize}
\end{theorem}
\noindent Setting $\delta = \sqrt{\frac{\log \log n}{\log n}}$ (so that $(\log n)^{O(1/\delta^2)} = \poly(n)$),
and $\secpar = \Theta(\sqrt{\log n \cdot \log \log n})$,
we have the following corollary.
\begin{corollary}[Formal statement of \Cref{thm:informal-deIP}]
    \label{cor:TS}
    Let language $L \in \TS$ where $T(n) = n^{O\left(\left(\frac{\log n}{\log \log n}\right)^{1/2}\right)}$ and $S(n) = \poly(n)$,
    there exists a \UIP with the following properties:
    \begin{itemize}
        \item The soundness error is $o(1)$.
        \item The number of rounds is $o(n)$.
        \item The prover and verifier message length per-round is $S(n) \cdot \poly(n)$.
        \item The prover runtime is $T(n) \cdot \poly(n)$.
        \item The verifier runtime is $S(n)^2 \cdot \poly(n)$.
    \end{itemize} 
\end{corollary}
\subsubsection{Augmenting Interactive Proofs of Transition of Turing Machines}
Our approach is similar to that in \cite{STOC:ReiRotRot16}.
Fix $\cM$ to be some Turing machine with $T(n)$ time and $S(n)$ space.
For every $t \le T(n)$,
let the language $\cL_{t}$ consist of all tuples $(x, w_1, w_2)\in \bin^n \times \bin^{O(S(n))} \times \bin^{O(S(n))}$ such that $\cM$ transitions from configuration $w_1$ to configuration $w_2$ in exactly $t$ steps when the input is $x$.

The crucial observation is captured in the following \emph{Augmentation Lemma} (\Cref{lem:aug}):
let $\batchK = \batchK(n) \in \poly(n)$ be some batch size parameter.
The language $\cL_{\batchK \cdot t}$ can be expressed as a batch of $\cL_{t}$.
Therefore,
an efficient batching protocol such as \Cref{thm:basic-bUIP} allows the verifier to verify $(\batchK \cdot t)$-step transitions by batch-verifying $t$-step transitions,
at the cost of a small increase in complexity.
By repeated application of \Cref{lem:aug},
the verifier gradually increase the transition length $t$ that it can verify,
until it reaches $\cL_{T(n)}$.
\begin{lemma}[Augmentation Lemma]
    Let $\cM$ be a Turing machine with $T(n)$ time and $S(n) \ge n$ space.
    Let $t = t(n)$ be some time bound,
    and let $\batchK = \batchK(n) \in \poly(n)$ be some batch size parameter.
    Assume that $\cL_{t}$ has an $\epsilon$-sound $(\ell, a, b, \Ptime, \Vtime)$ \UIP $\prot{}$.
    For every $\secpar \in \NN$, 
    there exists a protocol $\epsilon_A$-unambiguous $(\ell', a', b', \Ptime', \Vtime')$ \UIP $\protU{'}$ for $\cL_{\batchK \cdot t}$.
    With $\tO$ hiding $\poly(\secpar) \cdot \polylog(\batchK, S(n), a, \ell)$ factors,
    it has the following properties:
    \begin{itemize}
        \item $\epsilon' = 2\log \batchK(\epsilon + 2^{-\secpar})$.
        \item $\ell' = \tO(\ell \cdot \depth(\cV) \log \size(\cV))$.
        \item $a' = b' = \tO(a \ell + \poly(\ell))+ O(\batchK \cdot S(n))$.
        \item $\Ptime' = \tO(\batchK \cdot \ell \cdot \Ptime + \poly(\batchK, a, \ell, \size(\cV)))$.
        \item $\Vtime' = \tO(\ell \cdot \Vtime +  a \ell^2(\batchK S(n) + \depth(\cV) \log \size(\cV)) + a^2 \cdot \poly(\ell))$.
    \end{itemize}
    Furthermore,
    \begin{itemize}
        \item $\size(\cV') = \tO(\ell \cdot \size(\cV) + a\ell^2(\batchK S(n) + \depth(\cV) \log \size(\cV)) + a^2 \cdot \poly(\ell))$.
        \item $\depth(\cV') = \depth(\cV) + \tO(1)$.
    \end{itemize}
    \label{lem:aug}
\end{lemma}
\begin{proof}[Proof Sketch]
    The protocol $\protU{'}$ on input $(x, w_0, w_{\batchK \cdot t})$ is constructed as follows.
    \begin{enumerate}
        \item $\cP'$ runs $\cM$ on input $x$ and starting configuration $w_1$ for $\batchK \cdot t$ steps,
        obtaining all intermediate configurations $w_1, w_2, \ldots, w_{\batchK \cdot t}$.
        \emph{Note that if $\cP'$ already has access to $w_1, w_2, \ldots, w_{\batchK \cdot t}$,
        then it can skip this step.}
        \item $\cP'$ sends $w_t, w_{2t},\ldots, w_{\batchK t}$ to $\cV_{\batchK \cdot t}$.
        \item
        Both parties run the protocol $\prot{\bUIP}$ to batch $\prot{}$,
        with security parameter $\secpar$,
        that batch-verifies $\cL_{t}^{\otimes \batchK}$,
        on input $\set{(x, w_t, w_{2t})}_{t = 0}^\batchK$.
        \item The verifier $\cV'$ accepts iff the batched protocol accepts.
    \end{enumerate}

    \paragraph{Prescribed Completeness.} This is immediate.

    \paragraph{$\epsilon_A$-Unambiguity.} A deviating prover can either deviate by sending some incorrect configuration $w_t$,
    which leads to at least one incorrect $t$-step transition and hence a false statement in $\cL_{t}$,
    or by deviating in the batched protocol.
    In either case,
    the unambiguity of $\prot{\bUIP}$ ensures that the verifier accepts with probability at most $\epsilon_A$.

    \paragraph{Complexities.}
    The complexities remain the same as in \Cref{thm:basic-bUIP},
    with the input ($n$ in \Cref{thm:basic-bUIP}) being replaced by $O(S)$,
    except that $a_A = b_A \le \tO(a \ell + \poly(\ell))+ O(\batchK \cdot S)$,
    where the $O(\batchK \cdot S)$ term comes from the prover sending the configurations $w_t, w_{2t},\ldots, w_{\batchK t}$.
\end{proof}
\subsubsection{Our Construction of the Doubly-efficient \UIP for \TS}
For the base case,
we have the following (trivial) proof system for verifying one transition.
On input $(x, w_0, w_1)$,
the prover and verifier does not communicate,
and the verifier simply simulate $\cM$ on $x$ with $w_0$ by one step and verify it actually transitions to $w_1$,
which requires $O(S(n))$ time.
\begin{proposition}[A \UIP for verifying one transition]
    There exists a $(\ell_0, a_0, b_0, \Ptime_0, \Vtime_0)$ \UIP $\prot{0}$ for $\cL_1$ with the following properties:
    \begin{itemize}
        \item The unambiguity error is $\epsilon_0 = 0$.
        \item The number of rounds is $\ell_0 = 0$.
        \item The prover message length is $a_0 = 0$.
        \item The verifier message length is $b_0 = 0$.
        \item The prover does not do anything, so its runtime is 0.
        \item The verifier runs in time $O(S(n))$.
    \end{itemize}
    Furthermore,
    the verifier's decision circuit $\cV_0$ has the following properties:
    \begin{itemize}
        \item The circuit size is $\size(\cV_0) = O(S(n))$.
        \item The circuit depth is $\depth(\cV_0) = O(1)$.
    \end{itemize}
    \label{prop:base}
\end{proposition}
Using \Cref{lem:aug},
we can show the following claim by induction.
\begin{proposition}
    \label{prop:induction}
    Let $\batchK = \batchK(n) \in \poly(n), \gamma = \gamma(n) \in O\left(\sqrt{\frac{\log n}{\log \log n}}\right), \secpar = \secpar(n) \in \polylog(n)$ be parameters.
    There exist universal constants $C_1, C_2, C_3$ such that
    for each $i \le \gamma$,
    there exists a $(\ell_i, a_i, b_i, \Ptime_i, \Vtime_i)$ \UIP $\prot{i}$ for $\cL_{\batchK^i}$,
    and for large $n$,
    the following holds:
    \begin{enumerate}
        \item $\max(\batchK, S(n), a_i, \ell_i, \size(\cV_i))  = O(n^c)$,
            where \[ c = \limsup_{n \to \infty} \ceil{\cbound} \] is a constant independent of $n$ and $i$ (since $\batchK, S(n) \in \poly(n)$).
        \item Regarding the complexities of $\prot{i}$:
        \begin{itemize}
            \item The unambiguity error is $\epsilon_i \le i \cdot (2\log\batchK)^i  2^{-\secpar}$.
            \item The number of rounds is $\ell_i \le \log^{3C_1 i} n$.
            \item The prover and verifier message length per round is $a_i = b_i \le \batchK  S(n) \cdot \log^{4C_1C_3 i^2}$.
            \ifFOCS
            \item The prover runtime is \\$\Ptime_i = \batchK^i \cdot n^{C_2C_3} \log^{3C_1 i^2} n$.
            \item The verifier runs in time \\$\Vtime_i \le (\batchK S(n))^2\log^{8C_1C_3 i^2} n$.
            \else
            \item The prover runtime is $\Ptime_i = \batchK^i \cdot n^{C_2C_3} \log^{3C_1 i^2} n$.
            \item The verifier runs in time $\Vtime_i \le (\batchK S(n))^2\log^{8C_1C_3 i^2} n$.
            \fi
        \end{itemize}
        Furthermore,
        the verifier's decision circuit $\cV_i$ has the following properties:
        \begin{itemize}
            \ifFOCS
            \item The circuit size is \\$\size(\cV_i) \le (\batchK S(n))^2\log^{8C_1C_3 i^2} n$.
            \else
            \item The circuit size is $\size(\cV_i) \le (\batchK S(n))^2\log^{8C_1C_3 i^2} n$.
            \fi
            \item The circuit depth is $\depth(\cV_i) \le i \cdot \log^{C_1} n$.
        \end{itemize}
    \end{enumerate}
\end{proposition}
\noindent To simplify our calculation,
we use the following loose bounds on the constants involved.
\begin{remark}[Determining the constants $C_1$, $C_2$, $C_3$.]
The hidden constants in \Cref{lem:aug} determine the constants $C_1$, $C_2$, $C_3$,
as follows.
\begin{itemize}
    \item 
    Consider the upper bound of all $\poly(\secpar) \cdot \polylog(\batchK, S(n), a_i, \ell_i)$ factors hidden in $\tO$ of \Cref{lem:aug}.
    Assume that indeed $\batchK, S(n), a_i, \ell_i$ are each upper-bounded by $n^c$ for large $n$,
    and that $\secpar = \polylog(n)$,
    then for some constant $c_1$ and large $n$, the following holds,
    \ifFOCS
    \begin{align*}
    &\poly(\secpar)\polylog(\batchK, S(n), a_i, \ell_i) \\
    &\le (c\log n)^{c_1} \\
    &\le \left(\ceil{\cbound}\log n\right)^{c_1}.
    \end{align*}
    \else
    \[
    \poly(\secpar)\polylog(\batchK, S(n), a_i, \ell_i) \le (c\log n)^{c_1} \le \left(\ceil{\cbound}\log n\right)^{c_1}.
    \]
    \fi
    For large $n$, 
    \[\log n \ge \ceil{\cbound},\]
    so the above is upper-bounded by $\log^{2c_1} n$,
    irrespectively of the value of $c$ (or $C_1$).
    Therefore,
    if we take $C_1 \coloneqq 2c_1 + 1$, 
    then
    we can replace all occurrence of $\poly(\secpar) \cdot \polylog(\batchK, S(n), a_i, \ell_i)$ with $\log^{C_1} n$.
    \item 
    There exists a constant $c_2$ such that the term $\poly(\batchK, a_i, \ell_i, \size(\cV_i)) < (\batchK a_i \ell_i \size(\cV_i))^{c_2}$ in the $\Ptime$ recurrence relation.
    If $\max(\batchK, S(n), a_i, \ell_i, \size(\cV_i)) \le n^c$ for large $n$,
    then $(\batchK a_i \ell_i \size(\cV_i))^{c_2} \le n^{4cc_2}$ for large $n$.

    Therefore, we choose $C_2 \coloneqq 4cc_2$.
    \item $C_3$ is simply defined to be the power of $\ell$ in the terms $\poly(\ell)$ appearing in the recurrence relations.
\end{itemize}
\label{rmk:constants}
\end{remark}
\begin{proof}[Proof of \Cref{prop:induction}]
    Let $C_1 \coloneqq 2c_1 + 1$, $C_2 \coloneqq 4cc_2$ be the constants as defined in \Cref{rmk:constants}.
    (Recall that $c = \limsup_{n \to \infty} \ceil{\cbound}$.)

    We prove the following two claims simultaneously by induction on $i$.
    \begin{enumerate}
        \item $\max(\batchK, S(n), a_i, b_i, \size(\cV_i)) \le n^c$ for large $n$.
        \item The bounds about each of $\epsilon_i$, $\ell_i$, $a_i$, $b_i$, $\Ptime_i$, $\Vtime_i$, $\size(\cV_i)$, $\depth(\cV_i)$ hold.
    \end{enumerate}
    The base case, $i = 0$, corresponds to a transition of one step and holds trivially as in \Cref{prop:base}.

    \paragraph{Inductive Step.}
    Assume that we have a \UIP $\prot{i}$ for $\cL_{\batchK^{i}}$ with the stated complexity bounds.
    Note that an important consequence of $i \le \gamma \le \sqrt{\frac{\log n}{\log \log n}}$ is that $i^2\log \log n < \log n$.

    By our second inductive hypothesis, the following hold for large $n$.
    \begin{enumerate}
        \item $\depth(\cV_i) \le i \cdot \log^{C_1} n$.
        \item From $\size(\cV_i) \le O((\batchK S(n))^2 \log^{8C_1C_3i^2} n)$ and $\log (\batchK S(n)) = \log_n(\batchK S(n)) \log n$ we deduce 
        \ifFOCS
            \begin{align*}
                \log \size(\cV_i) &\le O(1) + 2\log (\batchK S(n)) \\
                &\quad + 8C_1C_3i^2 \log \log n\\
                &= O(1) + \\
                &\quad + (2\log_n (\batchK S(n)) + 8C_1C_3i^2)\log n\\
                &< \log^2 n.
            \end{align*}
            \else
            \begin{align*}
                \log \size(\cV_i) &\le O(1) + 2\log (\batchK S(n)) + 8C_1C_3i^2 \log \log n\\
                &= O(1) + (2\log_n (\batchK S(n)) + 8C_1C_3i^2)\log n\\
                &< (2\log_n (\batchK S(n)) + 8C_1C_3i^2 + 1)\log n < \log^2 n.
            \end{align*}
            \fi
        \end{enumerate}
    Therefore, for large $n$, \(\depth(\cV_i) \log \size(\cV_i) < i \cdot \log^{C_1} n \log^2 n< \log^{C_1 + 3}n\).

    \paragraph{Complexities}
    Applying \Cref{lem:aug} with batch size $\batchK$ and security parameter $\secpar$,
    we obtain a \UIP $\prot{i + 1}$ for $\cL_{\batchK^{i + 1}}$.
    As in \Cref{rmk:constants}, the first inductive hypothesis implies that all $\tO(1)$ terms hidden can be upper bounded by $\log^{C_1} n$ when $n$ is large.
    We work out the recurrence relations to obtain the following.
    \begin{itemize}
        \item $\depth(\cV_{i + 1}) = \depth(\cV_{i}) + \tO(1) \le \depth(\cV_{i}) + \log^{C_1} n \le (i + 1) \cdot \log^{C_1} n.$
        \item $\ell_{i + 1} = \ell_{i} \cdot (\depth(\cV_{i}) \log \size(\cV_{i})\log^{C_1} n) \le \log^{3C_1i + (2C_1 + 3)}n < \log^{3C_1(i + 1)}n$, (w.l.o.g. $C_1 > 3$).
        \ifFOCS
        \item $\begin{aligned}[t]
            a_{i + 1} = b_{i + 1} &\le a_{i} \cdot  \ell_{i} \cdot \log^{C_1} n + \ell_i^{C_3} \log^{C_1} n \\
            &\quad + O(\batchK \cdot S(n))\\
            &= \batchK \cdot S(n) \cdot \left(\log^{4C_1C_3i^2 + C_1i + C_1} n \right.\\
            &\quad\quad\quad \left. + \log^{3C_1C_3 i + C_1} n + O(1)\right)\\
            &< \batchK \cdot S(n) \cdot \log^{4C_1C_3(i + 1)^2} n.\\
            &\text{(w.l.o.g. $C_3 \ge 1$)}.
        \end{aligned}$
        \else
        \item $\begin{aligned}[t]
            a_{i + 1} = b_{i + 1} &\le a_{i} \cdot  \ell_{i} \cdot \log^{C_1} n + \ell_i^{C_3} \log^{C_1} n + O(\batchK \cdot S(n))\\
            &= \batchK \cdot S(n) \cdot (\log^{4C_1C_3i^2 + C_1i + C_1} n + \log^{3C_1C_3 i + C_1} n + O(1))\\
            &< \batchK \cdot S(n) \cdot \log^{4C_1C_3(i + 1)^2} n. \text{ (w.l.o.g. $C_3 \ge 1$)}.
        \end{aligned}$
        \fi
        \ifFOCS
        \item $\begin{aligned}[t]
            \Vtime_{i + 1} &\le \Vtime_i \cdot (\ell_i \cdot \log^{C_1} n) \\
            &\quad +\left(a_i\ell_i^2(\batchK S(n) \right.\\
            &\quad\quad\quad\quad\quad + \depth(\cV_i)\log \size(\cV_i))\\
            &\quad\quad\quad \left. + a_i^2\ell_i^{C_3}\right) \log^{C_1} n\\\
            &\le (\batchK S(n))^2 \log^{8C_1C_3i^2 + 3C_1i + C_1}n\\
            &\quad + \batchK S(n)(\batchK S(n) + \log^{C_1 + 3}n)\\
            &\quad \cdot \log^{4C_1C_3i^2 + 6C_1i}n \cdot \log^{C_1} n\\
            &\quad + (\batchK S(n))^2 \cdot  \log^{8C_1C_3i^2} n\\
            &\quad \cdot \log^{C_1C_3 i + C_1} n\\
            &< (\batchK S(n))^2 \cdot\log^{8C_1C_3(i+1)^2}n.
        \end{aligned}$
        \else
        \item $\begin{aligned}[t]
            \Vtime_{i + 1} &\le \Vtime_i \cdot (\ell_i \cdot \log^{C_1} n) + \left(a_i\ell_i^2(\batchK S(n) + \depth(\cV_i)\log \size(\cV_i)) + a_i^2\ell_i^{C_3}\right) \log^{C_1} n\\\
            &\le (\batchK S(n))^2 \log^{8C_1C_3i^2 + 3C_1i + C_1}n\\
            &\quad + \batchK S(n)(\batchK S(n) + \log^{C_1 + 3}n)\log^{4C_1C_3i^2 + 6C_1i}n \cdot \log^{C_1} n\\
            &\quad + (\batchK S(n))^2 \cdot  \log^{8C_1C_3i^2} n \cdot \log^{C_1C_3 i + C_1} n\\
            &< (\batchK S(n))^2 \cdot\log^{8C_1C_3(i+1)^2}n.
        \end{aligned}$
        \fi
        \ifFOCS
        \item Similarly, \[\size(\cV_{i + 1}) < (\batchK S(n))^2 \cdot \log^{8C_1C_3(i+1)^2}n.\]
        \else
        \item Similarly, $\size(\cV_{i + 1}) < (\batchK S(n))^2 \cdot \log^{8C_1C_3(i+1)^2}n.$
        \fi
        \ifFOCS
        \item $\begin{aligned}[t]
            \Ptime_{i + 1} &\le \Ptime_i \cdot (\batchK \cdot \ell_i \cdot \log^{C_1} n) \\
            &\quad + \poly(\batchK, a_i, \ell_i, \size(\cV_i))\\
            &\le \Ptime_i \cdot \batchK \log^{3C_1i^2 + C_1}n \log^{C_1} n \\
            &\quad + n^{C_2C_3}\\
            &< \batchK^{i + 1} n^{C_2C_3}\log^{3C_1(i + 1)^2 + C_1}n.
        \end{aligned}$
        \else
        \item $\begin{aligned}[t]
            \Ptime_{i + 1} &\le \Ptime_i \cdot (\batchK \cdot \ell_i \cdot \log^{C_1} n) + \poly(\batchK, a_i, \ell_i, \size(\cV_i))\\
            &\le \Ptime_i \cdot \batchK \log^{3C_1i^2 + C_1}n \log^{C_1} n + n^{C_2C_3}\\
            &< \batchK^{i + 1} n^{C_2C_3}\log^{3C_1(i + 1)^2 + C_1}n.
        \end{aligned}$
        \fi
    \end{itemize}
    \ifFOCS
    Recall that \[c = \ceil{\limsup_{n \to \infty} \cbound}.\]
    \else
    Recall that $c = \ceil{\limsup_{n \to \infty} \cbound}$.
    \fi
    The condition in the induction follows from that 
    $\max(\batchK, S(n), a_{i + 1}, b_{i + 1}, \size(\cV_{i + 1})) = (\batchK S(n))^2 \cdot \log^{8C_1C_3(i + 1)^2} n \le n^{c}$ for large $n$,
    since $(i + 1)^2 \le \gamma^2 \le \frac{\log n}{\log \log n}$ and $(\log n)^{\frac{\log n}{\log \log n}} = n$.

    Finally, the unambiguity error satisfies $\epsilon_{i + 1} = 2\log \batchK(\epsilon_{i} + 2^{-\secpar}) \le 2 \log \batchK (i (2\log \batchK)^i 2^{-\secpar} + 2^{-\secpar}) < (i + 1) (2\log \batchK)^{i + 1} 2^{-\secpar}$.
\end{proof}

Let $\delta(n) \in (0,1)$ be a parameter.
\Cref{thm:TS} now follows from \Cref{prop:induction} by setting $\batchK = T^{\delta}$ and $\gamma = \frac{1}{\delta}$.
In this case, 
$\cL_{\batchK^\gamma} = \cL$ and it has a \UIP with the following complexities:
\begin{itemize}
    \item $\begin{aligned}[t]
        \epsilon_\gamma \le \gamma (2\log n)^\gamma 2^{-\secpar} = 2^{-\secpar} \cdot (1/\delta) \cdot (\log n)^{O(1/\delta)}.
    \end{aligned}$
    \item $\begin{aligned}[t]
        \ell_\gamma \le \log^{3C_1\gamma} n = (\log n)^{O(1/\delta)}.
    \end{aligned}$
    \ifFOCS
    \item $\begin{aligned}[t]
        a_\gamma = b_\gamma &\le \batchK \cdot S(n) \log^{4C_1C_3\gamma^2}n \\
        &< T^{\delta} \cdot S(n) \cdot (\log n)^{O(1/\delta^2)}.
    \end{aligned}$
    \item $\begin{aligned}[t]
        \Ptime_\gamma &\le \batchK^\gamma n^{C_2C_3}\log^{3C_1\gamma^2}n \\
        &< T^{\delta} \cdot \poly(n) \cdot (\log n)^{O(1/\delta^2)}.
    \end{aligned}$
    \item $\begin{aligned}[t]
        \Vtime_\gamma &= (\batchK S(n))^2\log^{8C_1C_3\gamma^2} n \\
        &< T^{2\delta} \cdot S(n)^2 \cdot (\log n)^{O(1/\delta^2)}.
    \end{aligned}$
    \else
    \item $\begin{aligned}[t]
        a_\gamma = b_\gamma \le \batchK \cdot S(n) \log^{4C_1C_3\gamma^2}n = T^\delta \cdot S(n) \cdot (\log n)^{O(1/\delta^2)}.
    \end{aligned}$
    \item $\begin{aligned}[t]
        \Ptime_\gamma &\le \batchK^\gamma n^{C_2C_3}\log^{3C_1\gamma^2}n = T^\delta \cdot \poly(n) \cdot (\log n)^{O(1/\delta^2)}.
    \end{aligned}$
    \item $\begin{aligned}[t]
        \Vtime_\gamma &= (\batchK S(n))^2\log^{8C_1C_3\gamma^2} n = T^{2\delta} \cdot S(n)^2 \cdot (\log n)^{O(1/\delta^2)}.
    \end{aligned}$
    \fi
\end{itemize}

\section{IPP for \PVAL with Column Distance}
\label{sec:DcRR}

The work of \cite{TCC:RotRot20} constructed an $\IPP$ (same as \Cref{def:IPP} without the unambiguous soundness property) for the $\pval$ problem with respect to the Hamming distance. We use their ideas to construct a $\UIPP$ for a slightly different formulation of the $\pval$ problem which has a different distance measure. The same ideas also result in an instance reduction protocol for \pval. In particular, we use the column distance $\Deltac$ (\Cref{def:Deltac-dist}), as opposed to the Hamming distance.

We impose the additional constraint on our $\IPP$ to be a $\UIP$, as it is used as a building block in our batch $\UIP$ protocol.\footnote{We note that \cite{TCC:RotRot20} did not require their $\IPP$ to be unambiguous.
Their protocol, 
despite being slightly unambiguous, 
does not have an exponentially small unambiguity error.
Given that parallel repetition does not reduce unambiguity error,
we have to modify the protocol.}
Thus, we also prove the unambiguous soundness of our protocol. We now present our main result for $\UIPP$s.

\begin{theorem}[\UIPP for \pval with column distance; Restatement of \Cref{lem:DcRR}]
    \label{thm:DcRR-main}
    \DCMainStmt
\end{theorem}

\subsection{Overview of the \RR \IPP for \pval.}
The setup for the \RR \IPP is as follows:
 \begin{itemize}
     \item The verifier has oracle access to a function $\transcMat:\{0,1\}^\lgK\mapsto \FF$ (note the notation $\transcMat$ often refers to the truth table of the function which is a string in $\FF^{2^\lgK}$), where $\FF$ is a field of characteristic $2$. The verifier also has elements $\bm j\in (\FF^\lgK)^T$ and $\bm v\in \FF^T$.
     \item $\hat \transcMat: \FF^\lgK\mapsto \FF$ is the multilinear extension of $\transcMat$.
     \item The prover and verifier want to run a $\UIP$ protocol to verify whether $\hat \transcMat(\bm j)=\bm v$ or if $\hat{\transcMat}$ is $\delta$ far from $\pval(\bm j, \bm v)$ by making only $\tilde O\left(\frac{1}{\delta}\right)$ oracle queries to the function $\transcMat$. 
 \end{itemize}
 
Since the verifier only cares about the case when the distance is at least $\delta$, it could hope to query $\transcMat$ on just $O(\frac{1}{\delta})$ uniformly random points to get a false claim. 
Unfortunately, the verifier does not begin with a claim about any subset of $\frac{1}{\delta}$ coordinates, but only global constraints imposed by $\bm j$ and $\bm v$. 

The original idea for getting around this difficulty was given in $\cite{STOC:RotVadWig13}$.
Their idea was to break $\transcMat$ into two parts: $\transcMat_0, \transcMat_1$, each corresponding to a function in $\{0,1\}^{\lgK-1}\to\FF$ such that $\transcMat_b(x)=\transcMat(b,x)$ for $b\in \{0,1\}$, and thus
\[\transcMat(x_1, \vec{x})=(1-x_1)\transcMat_0(\vec x)+x_1\transcMat_1(\vec x).
\]
This, together with the fact that $\hat{\transcMat}$ is the multilinear extension of $\transcMat$, implies that \[\hat{\transcMat}(x_1, \vec x)=(1-x_1)\hat{\transcMat}_0(\vec x)+x_1\hat{\transcMat}_1(\vec x)\]

Exploiting this observation, one can ask the prover to reduce the $\pval$ claims about $\hat{\transcMat}$ to $\hat{\transcMat}_0$ and $\hat{\transcMat}_1$. Assuming $\transcMat$ was far from the $\pval$ set, regardless of the prover's answers, we would expect a random linear combination $\transcMat_0+c\transcMat_1$ to be far from the resulting \pval claim. Let's try to understand how far from the claim we can expect $\transcMat_0+c\transcMat_1$ to be.

We know that the prover needs to change $\transcMat$ in $\delta\cdot 2^\lgK$ points to satisfy the \pval instance. Thus, the prover could answer such that it would need to change the answers in both $\transcMat_0$ and $\transcMat_1$ in $\delta\cdot 2^{\lgK-1}$ positions each. 
Therefore, the prover may need to change $\transcMat_0+c\transcMat_1$ in only $\delta\cdot 2^{\lgK-1}$ points to satisfy the $\pval$ instance if the error positions are aligned. Thus, even if the error positions are aligned, we could hope that $\transcMat_0+c\transcMat_1$ remains $(\delta 2^{m-1})/2^{m-1}=\delta$-far from satisfying the resulting $\pval$ claim.

\cite{STOC:RotVadWig13} proved that with high probability, $\transcMat_0+c\transcMat_1$ is at least $\frac{\delta}{2}$-far from the resulting \pval since they did not have any assumption that the base \pval sets themselves have any distance. 
In contrast,
\cite{BKS18} got improved upon \cite{STOC:RotVadWig13} by additionally assuming that every two elements in the \pval set are also far apart. This allowed them to associate a unique point in the set \pval to which $\transcMat_0+c\transcMat_1$ is close.

Now we have the guarantee that the reduced instance has $\delta$-distance from the PVAL set. 
However, observe that if we want to query $\transcMat_0+c\transcMat_1$ on one point then we need to query $\transcMat$ on two points. 
Even though the ``instance size'' is halved,
the ``effective instance size'', measured by the number of needed queries to $\transcMat$, 
remains the same. 
Specifically, if we have to query $\transcMat_0+c\transcMat_1$ on $q$ points, then we need to query $\transcMat$ on $2q$ points. 
To make progress, 
we would in fact like to \emph{double} the fractional distance,
i.e. maintain the absolute number of error positions in the new claim.
If this is achieved,
the number of verifier queries to $\transcMat_0+c\transcMat_1$ can be halved while retaining the desired soundness.
Through this, 
the number of queries to $\transcMat$ remains the same,
and the \pval constraint is indeed only over the folded instance, 
which is half the original size.

The above construction is indeed quite lossy since $1$ error each in $\transcMat_0$ and $\transcMat_1$ can add up to just count as $1$ error in $\transcMat_0+c\transcMat_1$ when they align in the truth tables. 
Thus, \cite{TCC:RotRot20} had the idea of instead reducing not to $\transcMat_0+c\transcMat_1$ but to $\transcMat_0+c(\transcMat_1\circ \pi)$ where $\pi$ is a sufficiently random permutation. 
This step distributes the errors,
and $\transcMat_0 + c(\transcMat_1\circ \pi)$ is indeed about $2\delta$-far from the resulting instance. 

In the next subsection, we show that this approach is compatible with not just the Hamming distance but also column distance.
Essentially, this is the same as running $\ncol$ instances of \RR in parallel with the same randomness and with a \pval claim over the $\ncol$ instances together.
Additionally, since we want to use this protocol as a sub-protocol of a $\UIP$, we prove that the resulting protocol is unambiguous.
This requires $\FF$ to be a large enough field, which is without loss of generality since we can always just replace $\FF$ with a sufficiently large extension.

\subsection{Gap Amplification Lemmas for Column Distance}
In the section above, we gave an overview of how \cite{TCC:RotRot20} used ideas from \cite{STOC:RotVadWig13} and \cite{BKS18}. 
These results are actively studied in the SNARKs literature \cite{BCIKS23,DCC:AHIV23,EC:ACFY25,EPRINT:Zeilberger24,C:ACFY24} and are known as ``divergence'' and ``correlated-agreement'' theorems. We prove the $\Deltac$ analogues of \cite{STOC:RotVadWig13} and \cite{BKS18} here now.

We begin by giving an analogous version of the \cite{STOC:RotVadWig13} batch amplification lemma. The following lemma only shows that if $\transcMat_0$ and $\transcMat_1$ are $\delta$-far from vector spaces $V$ then $\transcMat_0+c\transcMat_1$ is at least $\delta/2$-far from $V$ with high probability.\footnote{Note that absolute distances are used in place of the relative forms.}
This result might seem suboptimal, but it is important to note that we do not have any assumption on the distance of the space $V$ itself.

\begin{lemma}[\cite{STOC:RotVadWig13} gap amplification for $\Deltac$] \label{lem: Batch RVW}
    Let $\FF$ be a field of characteristic $2$ and integers $\lgK, \lncol, T>0$. Let $V \subset \FF^{\{0,1\}^{\lgK+\lncol}}$ be a non-empty $\FF$-linear vector space, and $\transcMat_0, \transcMat_1 : \{0,1\}^{\lgK+\lncol} \rightarrow \FF$ be functions. 
    Suppose \[d_i = \Deltac(\transcMat_i, V).\] Then, \[\Pr_{c\leftarrow \FF^\ast}[\Deltac(\transcMat_0+c \transcMat_1, V)<\max(d_0, d_1)/2]\le \frac{1}{|\FF|-1}.\]
\end{lemma}

\begin{proof}
    We emulate the proof given in \cite{STOC:RotVadWig13}. For the sake of contradiction assume that \[\Pr[\Deltac(\transcMat_0+c \transcMat_1, V)<d_0/2]> \frac{1}{|\FF|}.\]

    Call a $c$ \textit{bad} if $\Deltac(\transcMat_0+c \transcMat_1, V)<d_0/2$. Now, if the theorem is false, then there exist $2$ bad values of $c$, let these be $c'$ and $c''$.

    We would like to use triangle inequality to say that if $g'=\transcMat_0+c'\transcMat_1$, $g''=\transcMat_0+c''\transcMat_1$ are both less than $\frac{d_0}{2}$ far from $V$ then their linear combinations are less than $\delta$-far as well. In particular, for any $r \in \FF$, we have that \[\Deltac(r\cdot g' + (1-r)\cdot g'', V)\le \Deltac(g, V)+\Deltac(g', V)<d_0\]
    
    Now, if $rc'+(1-r)c''=0$ then $r\cdot g'+(1-r)\cdot g'' = \transcMat_0$ which would imply that $\Deltac(\transcMat_0, V)<d_0$ which is a contradiction. Thus, there is no $r\in \FF$ such that $r(c'-c'')=-c''$. But, then we must have $c'=c''$ and there is at most one \textit{bad} value of $c$ as desired.
\end{proof}

\subsubsection{Correlated-Agreement under Column Distance}
\cite{BKS18} improves the \cite{STOC:RotVadWig13} analysis by observing that if we have an additional assumption about distance on the base vector space $V$ then we could get a better result. Their ideas also provide an improved gap amplification lemma which served as the point upon which \cite{TCC:RotRot20} built and improved. We prove the version of \cite{BKS18} in our $\Deltac$ setting. We later use this directly when we prove the \RR gap amplification lemma over $\Deltac$ distance.

\cite{BKS18} observed that if we could assume that if $\transcMat_0$ is close to $v_0$ and $\transcMat_1$ is close to $v_1$ where $v_0, v_1\in V$ where $V$ is a vector space then we could assume that the closest point to $\transcMat_0+c \transcMat_1$ is $v_0+c v_1$ as long as $V$ has large distance. As soon as we have such an assumption, we can immediately improve the \cite{STOC:RotVadWig13} analysis.

Similarly to what is done in \cite{TCC:RotRot20}, 
instead of stating the main \cite{BKS18} result on gap amplification directly, 
we will first prove a ``correlated-agreement'' lemma. 
Since the plan is to fold the functions as $\transcMat_0+c(\transcMat_1\circ \pi)$, we need to understand how the errors in the sum look more carefully.

\cite{BKS18} observed that if $V$ is some vector space with large distance and $\transcMat_0, \transcMat_1$ are different points in the ambient space of $V$, then if many points along the line $\transcMat_0+c \transcMat_1$ are close to $V$ then not only are both $\transcMat_0$ and $\transcMat_1$ close to $V$ but the number of coordinates that both of them simultaneously agree with elements of $V$ on is also large.

\begin{lemma}[Correlated Agreement for $\Deltac$] \label{lem: Batch BKS}
    Let $V\subseteq \FF^{\{0,1\}^{\lgK+\log \ncol}}$ be a non-empty vector space over $\FF$ and $\Deltac(V)=\Upsilon$.  For every $\varepsilon, \delta >0$ and $\transcMat_0, \transcMat_1 \in \FF^{\{0,1\}^{\lgK +\lncol}}$ such that \begin{itemize}
        \item $\delta(1-\varepsilon) \le \Upsilon/3$
        \item $|\{c \mid \Deltac(\transcMat_0+c \transcMat_1, V)<\delta(1-\epsilon)\}|>1/\epsilon$
    \end{itemize}
    \ifFOCS
    there exist $v_0, v_1 \in V$ such that for all $t\in \{0,1\}^{\log \ncol}$, 
    \begin{align*}
    &|\{i\mid (\transcMat_{0}[i,t]=v_{0}[i,t])\wedge (\transcMat_1[i,t] = v_1[i,t])\}|\\
    &\ge (1-\delta)2^\lgK
    \end{align*}
    \else
    there exist $v_0, v_1 \in V$ such that for all $t\in \{0,1\}^{\log \ncol}$, \[|\{i\mid (\transcMat_{0}[i,t]=v_{0}[i,t])\wedge (\transcMat_1[i,t] = v_1[i,t])\}|\ge (1-\delta)2^\lgK\]
    \fi
\end{lemma}

\begin{proof}
    Let $S\subseteq\{c \mid \Deltac(\transcMat_0+c \transcMat_1, V)<\delta(1-\epsilon)\}$ such that $|S|=1+\lfloor1/\varepsilon\rfloor$. Now, for each $c \in S$, let $v^{(c)}\in V$ be such that $\Deltac(\transcMat_0 + c \transcMat_1, v^{(c)})<\delta(1-\epsilon)$.

    Observe that now if all these points $(c, v^{(c)})$ lied on a line then we can get the result. In particular, we can write $v^{(c)}=v_0 + cv_1$ for some fixed $v_0, v_1$ independent of $c$. And, we would get
    \[\Deltac(\transcMat_0+c \transcMat_1, v_0+c v_1)<\delta (1-\varepsilon)\]

     Rearranging, we get for every $c \in S$, \[\Deltac(\transcMat_0-v_0, c(v_1-\transcMat_1))<\delta(1-\epsilon)\]

     Now, observe that for any $i\in \{0,1\}^\lgK, t\in \{0,1\}^{\log \ncol}$ if $\transcMat_0[i,t] \ne v_0[i,t]$ then there is at most one value of $c$ for which $\transcMat_0[i,t]-v_0[i,t] = c(v_1[i,t]-\transcMat_1[i,t])$ and similarly if $\transcMat_1[i,t]\ne v_1[i,t]$. 
     
     Thus, fix any $t\in \{0,1\}^{\log \ncol}$ and let \[T_t = \{i \mid (\transcMat_0[i,t] = v_0[i,t])\wedge (\transcMat_1[i,t] = v_1[i,t])\}\] we get that there exists a $c \in S$ such that 
     \ifFOCS
     \begin{align*}
     &\frac{|\{i\mid (\transcMat_0[i,t]-v_0[i,t]) \ne c(v_1[{i,t}]-\transcMat_{i,t})\}|}{2^\lgK}\\
     &\ge \left(1-\frac{|T_t|}{2^{\lgK}}\right)\left(1-\frac{1}{|S|}\right)
     \end{align*}
     \else
     \[\frac{|\{i\mid (\transcMat_0[i,t]-v_0[i,t]) \ne c(v_1[{i,t}]-\transcMat_{i,t})\}|}{2^\lgK}\ge \left(1-\frac{|T_t|}{2^{\lgK}}\right)\left(1-\frac{1}{|S|}\right)\]  
     \fi

     However, since, for any $T$, we have that: 
     \ifFOCS
     \begin{align*}
     &\Deltac(\transcMat_0-v_0, c(v_1-\transcMat_1))\\
     &\ge \frac{|\{i\mid (\transcMat_0[i,t]-v_0[{i,t}]) \ne c(v_1[{i,t}]-\transcMat_1[{i,t}])\}|}{2^\lgK}
     \end{align*}
     \else
     \[\Deltac(\transcMat_0-v_0, c(v_1-\transcMat_1))\ge \frac{|\{i\mid (\transcMat_0[i,t]-v_0[{i,t}]) \ne c(v_1[{i,t}]-\transcMat_1[{i,t}])\}|}{2^\lgK}\]
     \fi

    Thus, 
    \ifFOCS
    \begin{align*}
    &\left(1-\frac{|T|}{2^\lgK}\right)\left(1-\frac1{|S|}\right)\le \delta(1-\varepsilon)\\
    &\implies 1-\frac{|T|}{2^\lgK} < \delta \implies 2^\lgK(1-\delta)< |T|
    \end{align*}
    \else
    \[\left(1-\frac{|T|}{2^\lgK}\right)\left(1-\frac1{|S|}\right)\le \delta(1-\varepsilon)\implies 1-\frac{|T|}{2^\lgK} < \delta \implies 2^\lgK(1-\delta)< |T|\]
    \fi

     The final result is since $|S|>\frac{1}{\varepsilon}$. 
    
     Thus, we just need to show that indeed the $(c, v^{(c)})$ lie on a line. This is feasible to show because we additionally assume that $V$ has large distance. We will now use the same triangle inequality trick to show that if all $v^{(c)}$ are not on a line then there are two distinct points of $V$ which are $<\Upsilon$ far from each other. 
    
    \noindent\textbf{Claim. } All the points $(c, v^{(c)})$ lie on a line.
    \begin{proof}
    Observe that the above claim follows trivially when $\epsilon >1/2$ since then $|S|=2$. 

    We are interested in the regime when $\varepsilon \le 1/2$. Consider any $c_1, c_2, c_3 \in S$. Then, we can take a linear combination of $(\transcMat_0+c_1\transcMat_1)$ and $(\transcMat_0+c_2\transcMat_1)$ to get $(\transcMat_0+c_3\transcMat_1)$. 
    Observe that 
    \ifFOCS
    \begin{align*}
    &\Deltac\left(\frac{(c_2-c_3)(\transcMat_0+c_1\transcMat_1)+(c_3-c_1)(\transcMat_0+c_2\transcMat_1)}{c_2-c_1}, \right.\\
    &\left. \vphantom{\frac{(c_2-c_3)(\transcMat_0+c_1\transcMat_1)+(c_3-c_1)(\transcMat_0+c_2\transcMat_1)}{c_2-c_1}}\frac{(c_2-c_3)v^{(c_1)}+(c_3-c_1)v^{(c_2)}}{c_2-c_1}\right)\\
    &<2\delta(1-\epsilon)
    \end{align*}
    \else
    \[\Deltac\left(\frac{(c_2-c_3)(\transcMat_0+c_1\transcMat_1)+(c_3-c_1)(\transcMat_0+c_2\transcMat_1)}{c_2-c_1}, \frac{(c_2-c_3)v^{(c_1)}+(c_3-c_1)v^{(c_2)}}{c_2-c_1}\right)<2\delta(1-\epsilon)\]
    \fi
    Let $v'=\frac{(c_2-c_3)v^{(c_1)}+(c_3-c_1)v^{(c_2)}}{c_2-c_1}$. 
    Thus, by the above we have \[\Deltac(v', \transcMat_0+ c_3\transcMat_1)<2\delta(1-\epsilon)\] and $\Deltac(v^{(c_3)}, \transcMat_0+ c_3\transcMat_1)<\delta(1-\epsilon)$. 
    Finally, by the triangle inequality, we have that: \[\Deltac(v', v^{(c_3)})<3\delta(1-\varepsilon)\le \Upsilon\implies v'=v^{(c_3)}.\]
    Hence, $v'=v^{(c_3)}$ so $v^{(c_1)}, v^{(c_2)},$ and $v^{(c_3)}$ are indeed collinear as desired. 
\end{proof}
\end{proof}

\subsubsection{An \RR-Style Gap Amplification}

Now, we are ready to prove the \RR style gap amplification lemma for $\Deltac$.

\cite{TCC:RotRot20}'s idea is to amplify distance using not just $\transcMat_0+c\transcMat_1$ as in \cite{STOC:RotVadWig13} but instead to work with $\transcMat_0+c(\transcMat_1\circ \pi)$ where $\pi$ is an efficient to compute permutation. If we assume that $\pi\sim \Pi$ is $\eta$-close to being $d-$wise independent then we can hope that the errors in $\transcMat_0$ and $\transcMat_1 \circ \pi$ do not align with very high probability. The issue here is that even though we have that $\Deltac(\transcMat_0\|\transcMat_1, \pval(j,v))\ge \delta$, but that does not translate into any immediate claim about distance of $\transcMat_1\circ \pi$ with any immediate new $\pval$ claim. 

Thus, we need to add a \GKR based protocol first to create claims about $\pval(\transcMat_1\circ \tau)$ where $\tau=\pi \times I$\footnote{It implements $\pi$ in each column of the matrix as in \Cref{fig:a-matrix-permuted}.}. The permutation $\pi$ is sampled as $\pi \sim \Pi$ where $\Pi$ is an $\eta$-almost $d$-wise independent distribution as defined in \Cref{thm: random reversible circuits are independent}.

\begin{definition}\label{def: GKR inside RR circuits}
    $\Phi_{\pi, \bm j, v}$ is the circuit that takes in an input $\transcMat$ and computes whether \[\left(\widehat{\transcMat \circ (\pi\times I})\right)(\bm j)= v\]
    Here, we assume that $\pi$ has a succinct description.
\end{definition}

\begin{lemma}[Complexity of $\Phi_{\pi, \bm j, \bm v}$]
    The circuit $\Phi_{\pi, \bm j, \bm v}$ has size $\tilde{O}(T\cdot (2^{\lgK+\lncol})\cdot \log|\FF|)$ and depth $\poly(\lgK+\lncol, \log(T\log|\FF|))$.
\end{lemma}
\begin{proof}
    The permutation is only $1$ layer of the circuit and size linear in the input. Beyond that, it is verifying $T$ distinct multilinear evaluations which can be evaluated by simply taking the sum of the indicator variables. Thus, the result follows. 
\end{proof}

Now, consider the following sub-protocol:
\begin{algorithm}
  \setstretch{1.1}
  \caption{$(\cP_{\Pi}, \cV_{\Pi})$: GKR to get claims about $\transcMat\circ \tau$.}
  \label{alg:GKRinsideRR}
    \small
    \textbf{Input Parameters:} $\lgK$, $\ncol$, $T \in \NN$ and $\FF$ is a field.\\
    \textbf{Input:} $\pi \in \Pi$, a permutation, $\bm j \in \left(\FF^{\lgK+\log \ncol}\right)^T$ and $\bm v\in \FF^T$.\\
    \textbf{Prover Auxiliary Input:} $\transcMat \in \bin^{\lgK\times \ncol} \to \FF$.\\
    \textbf{Ingredients:} 
    \begin{itemize}[topsep=-0.02em,itemsep=-0.2em]
        \item \emph{The \GKR protocol} $\prot{\GKR}$ (\Cref{lem:GKR}) for checking $\C(\transcMat) = 1$.

        The prover and verifier gets the description $\lrag \C$,
        while the prover additionally gets the input $\transcMat$.
        Both parties additionally takes in the parameter $\pvalT \in \NN$,
        and in the end outputs $\pvalU \in (\FF^{\lgK + \log \ncol})^\pvalT, \pvalv \in \FF^\pvalT$.
    \end{itemize}
    \textbf{Output:} $\bm j' \in \left(\FF^{\lgK+\log \ncol}\right)^T, v'\in \FF^T$.

    \begin{algorithmic}[1]
        \State Let $\tau = \pi \times I$. 
            Let $\lrag{\Phi_{\pi, \bm j, \bm v}}$ be the description of the circuit $\Phi_{\pi, \bm j, \bm v}$ that takes in a function $\transcMat$ and checks whether $\widehat{(\transcMat\circ \tau)}(\bm j)=\bm v$.
        \State Run the \GKR protocol $T$ times on $\lrag{\Phi_{\pi, \bm j, \bm v}}$ and prover auxiliary input $\transcMat$ and output the resulting $\bm j', \bm v'$.
    \end{algorithmic}
\end{algorithm}

We can directly use this sub-protocol to create claims about $\transcMat_1\circ \tau$ and thus use \cite{TCC:RotRot20}'s idea of folding the function $\transcMat$ as $\transcMat_0+c(\transcMat_1\circ \tau)$. 

\begin{lemma}[\RR gap amplification for $\Deltac$] \label{lem: Batch RR}
    Let $\lgK, \lncol, T, \secpar, d>0$ be integers and $0<4d\le 2^{\lgK/50}$. Let $\FF$ be a field of characteristic $2$. For $\bm j \in \left(\FF^{\lgK+\lncol}\right)^T$ suppose that has minimum column distance $\ge 4d$.    
    Let $\bm v_0, \bm v_1\in \FF^T$ be vectors such that both $\pval(\bm j, \bm v_0)$ and $\pval(\bm j, \bm v_1)$ are non-empty. Let $\transcMat_0, \transcMat_1 : \{0,1\}^{\lgK+\lncol} \rightarrow \FF$ be functions.  
    Suppose \[d\le\Deltac\left(\begin{bmatrix}
        \transcMat_0\\
        \transcMat_1
    \end{bmatrix}, \begin{bmatrix}
        \pval(\bm j,\bm v_0)\\
        \pval(\bm j,\bm v_1)
    \end{bmatrix}\right)\] Consider a uniformly random $\pi\in \Pi_{\lgK,4d,2^{-\secpar}}$ where $\Pi$ is the set of \textit{permutations} on $\{0,1\}^\lgK$ defined in \Cref{subsec: affine perms}. Let $\tau = \pi \times I$ and $\bm j'\in \left(\FF^{m+\lncol}\right)^T, \bm v'\in \FF^T$ such that $\bm j', \bm v'$ is outputted by the \Cref{alg:GKRinsideRR} protocol on input $\transcMat_1\circ \tau, \pi^{-1},\bm j, v_1$ and the same additional parameters i.e. $\lgK, \lncol, T$. 

    Now, for $c\leftarrow \FF$ define: \[g=\transcMat_0+c(\transcMat_1\circ \tau).\] 
    
    Let $S_{\pi}\subseteq \FF^{2t}$ be the set of pairs of vectors $(u_0, u_1)$ such that the sets $\pval((\bm j, \bm j'), (\bm v_0, u_0))$ and $\pval((\bm j, \bm j', (u_1, \bm v'))$ are not empty. For $(u_0, u_1)\in S_{\pi}$, define 
    \ifFOCS
    \begin{align*}
    &\delta_{\pi, j', v', (u_0, u_1), c}\\
    &=\Deltac(g, \pval((\bm j, \bm j'), (v_0, u_0)+c(u_1, v')))
    \end{align*}
    \else
    \[\delta_{\pi, j', v', (u_0, u_1), c}=\Deltac(g, \pval((\bm j, \bm j'), (v_0, u_0)+c(u_1, v'))\] 
    \fi
    Then for every $\gamma < (2/d, 1)$ and $\epsilon \in (0,1)$, taking \[\delta'= d(1 -\gamma)(1-\varepsilon),\] we have that 
    \ifFOCS
    \begin{align*}
    &\Pr_{\pi \leftarrow \Pi, \Cref{alg:GKRinsideRR}}
    \begin{bmatrix}
    \begin{aligned}
        &\exists (\bm {u_0}, \bm {u_1}) \in S_{\pi} \\
        &\text{ s.t. } \Pr_{c_0, c_1\leftarrow \FF}\left[\delta_{\pi, (u_0, u_1), c, j'}<\delta'\right]\\
        &\quad > \frac{1}{\varepsilon |\FF|}+\frac{1}{|\FF|}
    \end{aligned}
    \end{bmatrix}\\
    &<2^{-\secpar}+\exp(-\gamma d/6)+\err_{\GKR}
    \end{align*}
    \else
    \begin{align*}
    \Pr_{\pi \leftarrow \Pi, \Cref{alg:GKRinsideRR}}\left[\exists (\bm {u_0}, \bm {u_1}) \in S_{\pi} \text{ s.t. } \Pr_{c_0, c_1\leftarrow \FF}\left[\delta_{\pi, (u_0, u_1), c, j'}<\delta'\right]> \frac{1}{\varepsilon |\FF|}+\frac{1}{|\FF|}\right]\\<2^{-\secpar}+\exp(-\gamma d/6)+\err_{\GKR}
    \end{align*}
    \fi
    where $\err_{\GKR}$ refers to the failure probability of \Cref{alg:GKRinsideRR} as described in \Cref{lem:RVW}. 
    Thus, we have that 
    \ifFOCS
    \begin{align*}
    &\err_{\GKR}\le \left(\epsilon_{\GKR}(D, \log S, \abs{\FF}) \vphantom{\binom{\K}{4d} \abs{\FF}^{4d}}\right.\\
    &\quad \left. + (\epsilon_{\GKR}(D, \log S, \abs{\FF}))^T \cdot\left(\binom{\K}{4d} \abs{\FF}^{4d}\right)^\ncol\right)
    \end{align*}
    \else
    \[\err_{\GKR}\le \left(\epsilon_{\GKR}(D, \log S, \abs{\FF}) + (\epsilon_{\GKR}(D, \log S, \abs{\FF}))^T \cdot\left(\binom{\K}{4d} \abs{\FF}^{4d}\right)^\ncol\right)\]
    \fi
    where $D$ and $S$ are the depth and size of $\Phi_{\pi, \bm j, v}$ respectively.
\end{lemma}

\begin{proof}
For the sake of contradiction, assume that the lemma is false. 

Observe that for any $\pi \in \Pi$, with probability $\ge 1 - \err_{\GKR}$, by \Cref{lem:RVW} and \Cref{lem:GKR}, there exists a $g\in \pval(\bm j', v')\cap B_{4d, \FF}(\transcMat_1\circ \tau)$ iff the prover was honest with implicit input $g$ and $\left(g\circ \tau^{-1}\right)\in \pval(\bm j, v_1)$. 

Thus, we condition on this event happening and that accounts for the $\err_{\GKR}$ term in the lemma by union bound. 

Let $\bar\delta = \frac{\delta'}{1-\epsilon}<d<\frac{4d}{3}$. Now, with probability $>p=2^{-\secpar}+\exp(-\gamma\cdot d/6)$ over the choice of $\pi$, there exists $(u_0, u_1)\in S_{\pi}$ and $c_0\ne 0$, such that \[\Pr_{c_1\leftarrow \FF}[\delta_{\pi, (u_0, u_1), c, \bm j'}<\bar\delta(1-\epsilon)]>\frac{1}{\epsilon|\FF|}\]

    Since $\bar{\delta}< 4d/3$, with probability $>p$, over the choice of $\pi\in \Pi$, let the choice of $(u_0, u_1)\in S_{\pi}$ and $c\ne 0$ be as above. Thus, let $V=\pval((\bm j, \bm j'),0)$ and observe that $\Deltac(V)\ge 4d$. Let 
    \ifFOCS
    \begin{align*}
    &h_0\in \pval((\bm j, \bm j'), (v_0, u_0)), \\
    &h_1\in \pval((\bm j, \bm j'), (u_1, v')),
    \end{align*}
    \else
    \[h_0\in \pval((\bm j, \bm j'), (v_0, u_0)), \hspace{1cm} h_1\in \pval((\bm j, \bm j'), (u_1, v')),\]
    \fi

     then 
     \ifFOCS
     \begin{align*}
     &\Pr_{c\leftarrow \FF}
     \begin{bmatrix}
     \begin{aligned}
     &\Deltac(V+h_0+c(V+h_1),\transcMat_0+c(\transcMat_1\circ \tau))\\
     &<\bar\delta(1-\epsilon)
     \end{aligned}
     \end{bmatrix}\\
     &>\frac{1}{\epsilon|\FF|},
     \end{align*}
     \else
     \[\Pr_{c\leftarrow \FF}[\Deltac(V+h_0+c(V+h_1),\transcMat_0+c(\transcMat_1\circ \tau))<\bar\delta(1-\epsilon)]>\frac{1}{\epsilon|\FF|},\] 
     \fi
     and thus 
     \[\Pr_{c\leftarrow\FF}[\Deltac(V, (\transcMat_0- h_0)+ c(\transcMat_1\circ \tau-h_1))<\bar\delta(1-\epsilon)]>\frac{1}{\epsilon|\FF|}.\]  
     Letting $u^\ast = \transcMat_0 - h_0$ and $u = \transcMat_1\circ \tau-h_1$ in \Cref{lem: Batch BKS}, we get that there exist $y_0, y_1 \in V$ such that for all $r\in \{0,1\}^\lncol$ 
     \ifFOCS
     \begin{align*}
     &\left|\begin{Bmatrix}
        \begin{aligned}
        &i\in \{0,1\}^\lgK \mid (\transcMat_0(i,r)-h_0(i,r)=y_0(i,r))\\
        &\quad \wedge (\transcMat_1\circ \tau(i,r)-h_1(i,r)=y_1(i,r))
        \end{aligned}
        \end{Bmatrix}\right|\\
        &\ge 2^\lgK - \bar\delta,
     \end{align*} 
     \else
     \[|\{i\in \{0,1\}^\lgK \mid (\transcMat_0(i,r)-h_0(i,r)=y_0(i,r))\wedge (\transcMat_1\circ \tau(i,r)-h_1(i,r)=y_1(i,r))\}|\ge 2^\lgK - \bar\delta,\] 
     \fi
     or equivalently for all $r\in \{0,1\}^\lncol$, 
     \ifFOCS
     \begin{align*}
     &\left|\begin{Bmatrix}
        \begin{aligned}
        &i\in \{0,1\}^\lgK \mid (\transcMat_0(i,r)=(y_0+h_0)(i,r))\\
        &\quad \wedge (\transcMat_1\circ \tau(i,r)=(y_1+h_1)(i,r))
        \end{aligned}
        \end{Bmatrix}\right|\\
        &\ge 2^\lgK - \bar\delta,
     \end{align*}
     \else
     \[|\{i\in \{0,1\}^\lgK \mid (\transcMat_0(i,r)=(y_0+h_0)(i,r))\wedge (\transcMat_1\circ \tau(i,r)=(y_1+h_1)(i,r))\}|\ge 2^\lgK - \bar\delta\]
     \fi

        Now, observe that since $(y_0+h_0) \in \pval(\bm j, v_0)$, $(y_1+h_1) \in \pval(\bm j', v')$, and \\$\Deltac(\pval((\bm j, \bm j'), 0))>3\bar{\delta}$, we have that there is a unique element in $z_0\in \pval(\bm j, v_0)$ such that $\Deltac(\transcMat_0, z_0)\le \bar\delta$ and similarly, there is a unique element $z_1 \in \pval(\bm j', v')$ such that $\Deltac(\transcMat_1\circ \tau, z_1)\le \bar{\delta}$. Now, observe that since $z_1\in \pval(\bm j', v')$ and $\Deltac(\transcMat_1\circ \tau, z_1)\le 4d$, we must have that the prover was honest with implicit input $z_1$ in mind and $z_1\circ \tau^{-1}\in \pval(\bm j, v_1)$. Observe that $\Deltac(\pval(\bm j, v_1))\ge 4d$. Thus, $z_1'=z_1\circ \tau^{-1}$ is the unique element such that $z_1'\in \pval(\bm j, v_1)$ and $\Deltac(\transcMat_1, z_1')<\bar\delta$.

     Plugging, this back in, we get that with probability $\ge p$ over the choice of $\pi \in \Pi$, we have that for all $r\in \{0,1\}^\lncol$ 
     \ifFOCS
     \begin{align*}
     &\left|\begin{Bmatrix}
        \begin{aligned}
        &i\in \{0,1\}^\lgK \mid (\transcMat_0(i,r)=z_0(i,r))\\
        &\quad \wedge (\transcMat_1\circ \tau(i,r)=z_1'\circ \tau(i,r))
        \end{aligned}
        \end{Bmatrix}\right|\\
        &>2^\lgK - \bar\delta,
     \end{align*}
     \else
     \[|\{i\in \{0,1\}^\lgK \mid (\transcMat_0(i,r)=z_0(i,r))\wedge (\transcMat_1\circ \tau(i,r)=z_1'\circ \tau(i,r))\}|>2^\lgK - \bar\delta\]
     \fi

    Since $\Deltac(\transcMat_0\|\transcMat_1, z_0\|z_1')\ge d$, for each $r\in \{0,1\}^\lncol$, let 
    \ifFOCS
    \begin{align*}
        &\delta_{r, 0}=\Deltac(\transcMat_0(\cdot, r), z_0(\cdot, r))\\
        &\delta_{r,1}=\Deltac(\transcMat_1(\cdot, r), z_1'(\cdot, r)).
    \end{align*}
    \else
    \[\delta_{r, 0}=\Deltac(\transcMat_0(\cdot, r), z_0(\cdot, r)) \hspace{1cm} \delta_{r,1}=\Deltac(\transcMat_1(\cdot, r), z_1'(\cdot, r)).\]
    \fi
    By our assumption, we know that $\max_{r\in \{0,1\}^\lgK}(\delta_{r,0}+\delta_{r,1})=\Deltac(\transcMat_0\|\transcMat_1, z_0\|z_1')\ge d$.
    
    Now, we fix our focus to a fixed $r\in \{0,1\}^\lncol$ such that 
    $\delta_{r,0}+\delta_{r,1}\ge d$. Now, for this $r$, with probability $\ge p$ over the choice of $\pi$, we have that 
    \ifFOCS
    \begin{align*}
    \left|\begin{Bmatrix}
    \begin{aligned}
    &i\in \{0,1\}^\lgK\mid \left(\transcMat_0(i,r)=z_0\right)\\
    &\quad \wedge\left( \transcMat_1\circ\tau(i,r)=z_1'\circ\tau(i,r)\right)
    \end{aligned}
    \end{Bmatrix}\right|
    \ge 2^\lgK - \bar\delta
    \end{align*}
    \else
    \[|\{i\in \{0,1\}^\lgK\mid \left(\transcMat_0(i,r)=z_0\right)\wedge\left( \transcMat_1\circ\tau(i,r)=z_1'\circ\tau(i,r)\right)\}|\ge 2^\lgK - \bar\delta\]
    \fi

    Let 
    \ifFOCS
    $E_0=\{i\mid \left(\transcMat_0(i,r)\ne z_0(i,r)\right)\}$ and $E_1=\{i\mid \left(\transcMat_1(i,r)\ne z_1'(i,r)\right)\}$.
    \else
    \[E_0=\{i\mid \left(\transcMat_0(i,r)\ne z_0(i,r)\right)\} \hspace{1cm} E_1=\{i\mid \left(\transcMat_1(i,r)\ne z_1'(i,r)\right)\}.\]
    \fi
    We know $|E_0|+|E_1|\ge d$. 
    Thus, we can take subsets $A_0\subseteq E_0$ and $E_1\subseteq A_1$ such that $|A_0|+|A_1|\ge d$ and $|A_0|, |A_1|\le d$.
    Therefore,
    probability $\ge p$ over the choice of $\pi$, we have that $|A_0\cap \pi A_1|\ge d-(1-\gamma)d\ge \gamma d$.

    Since $\pi \sim \Pi_{\lgK, 4d, 2^{-\secpar}}$ and $p=2^{-\secpar}+\exp(-\gamma d/6)$, this would contradict \Cref{thm: random set intersection is tiny}. Thus, we have a contradiction.
\end{proof}

\subsection{Efficient \IPP for \pval with Column Distance}

Now, we give an analogous efficient Interactive proof of proximity (IPP) for $\pval$ to \cite{TCC:RotRot20} in the column distance setting. The key differences here from their work is that we add a soundness parameter as we want the protocol to be unambiguous,
and our distance notion is different. Besides these differences which mostly just affect our choices of parameters, the protocol is essentially the same as \RR. 

For the applications that we use the protocol, we do not require it to be an IPP. We mention the $\IPP$ here for completeness and historical reasons.

Before presenting the protocol, we go over the ideas in the protocol first:
\begin{itemize}
    \item We begin with a function $\transcMat:\{0,1\}^{\lgK+\lncol}\rightarrow \FF$ and a $\pval(\bm j,\bm v)$ set. The verifier would like to check whether $\transcMat\in \pval(\bm j,\bm v)$ or if $\Deltac(\transcMat, \pval(\bm j,v))>d$. The verifier would also like the final check to be efficient and succinctly described so that it could be further delegated.
    \item \textbf{Breaking into Multilinear Functions:} The verifier asks the prover to break $\transcMat$ into multilinear functions $\transcMat_0, \transcMat_1:\{0,1\}^{\lgK +\lncol-1}\mapsto \FF$ such that $\hat \transcMat(\chi, x)=(1-\chi)\hat \transcMat_0(x)+\chi \hat \transcMat_1(x)$ for $\chi\in \FF$. This is possible due to the multilinearity of $\hat \transcMat$. The verifier then asks the prover to break $j_i= (\chi_i, j_i')$ and send $\transcMat_0(j_i')$ and$ \transcMat_1(j_i')$. Let the resulting $\pval$ instances be $\pval(\bm j', \bm v_0)$ and $\pval(\bm j',\bm v_1)$.

    The verifier can now use the linearity to check if these claims are consistent.
    \item \textbf{Permuting Errors:} Now, the prover uses \Cref{alg:GKRinsideRR} to convert the claim $\transcMat_1\in \pval(\bm j',\bm v_1)$ into a new claim that $\transcMat_1\circ \tau\in \pval(\bm j',\bm v')$. 
    \item \textbf{Filling in Values and emptiness checks:} The verifier would want to create now a new \pval claim about $\transcMat_0+c(\transcMat_1\circ \tau)$ but does not know $\hat \transcMat_0(\bm j')$ or $\widehat{\transcMat_1\circ \tau}(\bm j')$. Thus, the verifier simply asks the prover to fill in these values. This is all compatible with the \Cref{lem: Batch RR} gap amplification that we have proved above.
    \item \textbf{Curve fitting:} Assuming all this goes well, the verifier has reduced to a claim about a function $\transcMat':\{0,1\}^{\lgK +\lncol -1}\mapsto \FF$ but the resulting $\pval$ instance now has $2T$ evaluations instead of $T$. Thus, the verifier can ask the prover to construct a low degree curve through $j'$ and $j'$ and ask the prover for values of $\transcMat_0$ and $\transcMat_1\circ \tau$ on all the points on the curve.
    \item \textbf{Reducing number of evaluation points:} Now, the verifier just picks $T$ uniformly random points on the curve. Since $j'$ was uniformly random, these $T$ points also end up being uniformly random elements of $\FF^{\lgK+\lncol-1}$ and the distance guarantee continues to hold with high probability.
\end{itemize}

The protocol (\Cref{alg: RRbetter}) is essentially just the above with some parameter balancing to get the appropriate soundness error. We need to do all the above over larger fields than \cite{TCC:RotRot20} did since we also prove unambiguity of the procedure and thus cannot use repetition to amplify the unambiguous soundness error. In subsequent subsections - we describe each of these components of the full protocol. We then combine all of these parts to get \Cref{alg: RRbetter} and the following theorem:
\begin{theorem}[IPP for $\pval$ with column distance]
    \label{thm: DcRR}
    Let $\lgK,\lncol, T, \secpar, d \in \NN$. Let $\FF$ be a constructible field ensembles such that $\FF$ has characteristic $2$.
    If
    \begin{itemize}
        \item $|\FF|\ge \Omega(2^{\secpar}\cdot T^2(\lgK+\lncol)^2)$,
        \item $T>8d\cdot 2^{\lncol}\cdot \lgK$,
        \item $d\ge 16\cdot \lgK\cdot \secpar$
    \end{itemize}

    Then for any $\bm j \in (\FF^{\lgK+\lncol})^T, v\in \FF^T$, the matrix $\pval(\bm j,\bm v)$ has a public coin unambiguous $\IPP$ as long as $\Deltac(\pval(\bm j, 0))>4d$:
    \begin{itemize}
        \item Soundness error: $\lgK\cdot 2^{4-\secpar}$
        \item Prover to Verifier communication complexity: $T\log |\FF|\cdot \poly(\lgK,\lncol)+\ncol\cdot \log|\FF|\cdot \poly(d)$.
        \item Verifier to Prover communication complexity: $T\log |\FF|\cdot \poly(\lgK,\lncol)$ 
        \item Round Complexity: $\poly(\lgK,\lncol, \log(T\log|\FF|))$
        \item Prover runtime: $\poly(2^{\lgK+\lncol},T\log |\FF|)$
        \item Verifier runtime: \ifFOCS\\ \fi$T\log|\FF|\cdot \poly(\lgK,\lncol, \log(T\log|\FF|), \log |\FF|) + \tilde{O}(\ncol\cdot \log|\FF|\cdot \poly(d))$ 
        \item Query complexity : $\abs{\RRS} = \left\lceil\frac{4\secpar\cdot 2^{\lgK}}{d}\right\rceil\cdot \ncol$.
    \end{itemize}
    The final verifier verdict is given by a succinctly described predicate $\Phi : \FF^{\left\lceil\frac{4\secpar\cdot 2^{\lgK}}{d}\right\rceil\cdot2^{\lncol}}\mapsto \{0,1\}$ along with a succinctly described set $Q\subseteq [2^\lgK]$ such that $|Q|=\left\lceil\frac{4\secpar\cdot 2^{\lgK}}{d}\right\rceil$, and the verdict is $\Phi(f|Q\times \{0,1\}^\lncol)$ with $|\left<Q\right>|=\tilde O(\lgK^2d\sigma\cdot +\lgK\log|\FF|+\sigma\cdot \poly(d))$ and $|\left<\Phi\right>|=|\left<Q\right>|+\sigma\cdot \poly(d)\cdot \ncol\cdot \log|\FF|$.
\end{theorem}

\subsubsection{Breaking into Multilinear Functions and Consistency Checks}

\begin{algorithm}
  \setstretch{1.1}
  \caption{Breaking $\transcMat$ into two multilinear functions}
  \label{alg: break into two}
  \textbf{Input Parameters:} $\lgK$, $\ncol$, $T \in \NN$ and $\FF$ is a field.\\
  \textbf{Input:} $\bm j \in (\FF^{\lgK+\lncol})^T$ and $\bm v\in \FF^T$.\\
  \textbf{Prover Auxiliary Input:} $\transcMat: \{0,1\}^{\lgK+\lncol}\rightarrow \FF$.\\
  \textbf{Output:}  $\{j'_i, \zeta^{(0)}_i, \zeta^{(1)}_i\}_{i\in [T]}$.
  \begin{algorithmic}[1]
    \State $\cP$ lets $\bm j_i = (\chi_i, \bm j_i')\in \FF\times \FF^{\lgK-1+\lncol}$. 
    \State $\cP$ sends $\set{\zeta_i^{(b)}=\hat \transcMat(b, \bm j_i')}_{i\in [T], b\in \{0,1\}}$ to $\cV$.
    \If {$\bigwedge_{i\in [T]}\left[v_i=(1-\chi_i)\zeta_{i}^{(0)}+\chi_i\zeta_i^{(1)}\right]$}
    \State $\cV$ outputs $\{\bm j_i', \zeta^{(b)}_{i}\}_{b\in \{0,1\}, i \in [T]}$.
    \Else
    \State $\cV$ halts and rejects.
    \EndIf
\end{algorithmic}
\end{algorithm}

\begin{lemma}\label{lem: break into two parts}
    Let $T, d, \lgK, \lncol \in \NN$ and $\FF$, a finite field of characteristic $2$. $\bm j \in (\FF^{\lgK+\lncol})^T$ and $\bm v\in \FF^T$. If, $\Deltac(\pval(j, 0))\ge 4d$, then 
    \Cref{alg: break into two} has the following properties:
    \begin{itemize}
    \item \textbf{Unambiguous Distance Preservation:}
    \begin{itemize}
        \item If $\Deltac(\transcMat, \pval(\bm j, \bm v))\le d$ and the prover answers according to the prover's prescribed strategy for the unique $\transcMat^\ast \in \pval(\bm j, \bm v)\cap \rowball(\transcMat)$ then $\transcMat^\ast$ remains the unique element in $\rowball(\transcMat) \cap (\pval(\bm j', \bm \zeta^{(0)})\|\pval(\bm j', \zeta^{(1)}))$.  
        \item If $\Deltac(\transcMat, \pval(\bm j, \bm v))> d$ or the prover is not honest with respect to the unique $\transcMat^\ast \in \rowball(\transcMat)\cap \pval(\bm j, \bm v)$ then, $\Deltac(\transcMat, \pval(\bm j', \zeta^{(0)})\| \pval(\bm j', \zeta^{(1)}).)>d$.
    \end{itemize}
        \item \textbf{Round Complexity:} $1$
        \item \textbf{Communication Complexity:} $T\cdot \log \FF\cdot (\lgK+\lncol+1)$.
        \item \textbf{Verifier Time Complexity:} $O(T\poly\log(\FF))$
        \item \textbf{Prover Runtime:} $\poly(2^{\lgK+\lncol}, T, \log \FF)$.
        \item In all cases: $\Deltac(\pval(\bm j', \bm 0))\ge 4d$.
    \end{itemize}

    \begin{proof}
    The round, communication and time complexity follow directly. The final property is since $\pval(\bm j', \bm 0)\|\bm 0 \subseteq \pval(\bm j, 0)$ and the first has distance $\Deltac(\pval(\bm j', \bm 0))$.

    Now, if the verifier doesn't reject and $\transcMat_0\in \pval(\bm j', \mathbf \zeta^{(0)})$ and $\transcMat_1\in \pval(\bm j', \mathbf \zeta^{(1)})$ then, we have that for all $i\in [T]$:
    \ifFOCS
    \begin{align*}
    \widehat{\transcMat_0\|\transcMat_1}(\bm j_i) &= \widehat{\transcMat_0\|\transcMat_1}(\chi_i, \bm j_i')= (1-\chi_i)\widehat{\transcMat_0}(\bm j_i')+\chi_i\widehat{\transcMat_1}(\bm j_i') \\
    &= (1-\chi_i)\zeta_i^{(0)}+\chi_i\cdot \zeta_i^{(1)} = \bm v_i
    \end{align*} 
    \else
    \[\widehat{\transcMat_0\|\transcMat_1}(\bm j_i) = \widehat{\transcMat_0\|\transcMat_1}(\chi_i, \bm j_i')= (1-\chi_i)\widehat{\transcMat_0}(\bm j_i')+\chi_i\widehat{\transcMat_1}(\bm j_i') = (1-\chi_i)\zeta_i^{(0)}+\chi_i\cdot \zeta_i^{(1)} = \bm v_i\] 
    \fi
    by multilinearity. We get that $\transcMat_0\|\transcMat_1 \in \pval(\mathbf{j, v})$. This clearly implies the honest prover case and the $\Deltac(\transcMat, \pval(\bm j, \bm v))>d$ cases.

    If the prover is dishonest with respect to $\transcMat^\ast$, then we know that $\transcMat^\ast \notin \pval(\bm j', \mathbf \zeta^{(0)})\|\pval(\bm j', \mathbf \zeta^{(1)})$. Thus, $\Deltac(\transcMat^\ast, \pval(\bm j', \mathbf \zeta^{(0)})\|\pval(\bm j', \mathbf \zeta^{(1)}))\ge 4d\implies \Deltac(\transcMat, \pval(\bm j', \mathbf \zeta^{(0)})\|\pval(\bm j', \mathbf \zeta^{(1)}))>d$.
    \end{proof}
\end{lemma}

\subsubsection{Permuting Errors and Curve Fitting}

\begin{algorithm}
  \setstretch{1.1}
  \caption{Permuting Errors and Curve Fitting}
  \label{alg: gap amplification}
  \textbf{Input parameters:} $\lgK$, $\ncol$, $T$, $\secpar \in \NN$ and $\FF$ is a field.\\
  \textbf{Input:} $\bm j\in (\FF^{\lgK+\lncol})^T$ and $\bm v_0, \bm v_1\in \FF^T$.\\
  \textbf{Prover Auxiliary Input:} $\transcMat_0, \transcMat_1 : \{0,1\}^{\lgK+\lncol}\rightarrow \FF$.\\
  \textbf{Other parameters:}
  $\mathbf \Lambda=(\Lambda_1, \Lambda_2, \ldots, \Lambda_{2T})$ are distinct fixed canonical points in $\FF$.

  \textbf{Output:} A permutation $\tau: \{0,1\}^{\lgK + \lncol} \to \bin^{\lgK + \lncol}$, a curve $\mathcal C: \FF\rightarrow \FF^{\lgK+\lncol}$, $\bm j' \in (\FF^{\lgK+\lncol})^T$, $\bm v' \in \FF^T$, and $g^{(0)}$ and $g^{(1)} : \FF\rightarrow \FF$ which are univariate polynomials of degree at most $2T(\lgK+\lncol)$.
  \begin{algorithmic}[1]
    \State $\cV$ samples $\pi\leftarrow \Pi_{\lgK, 4d, 2^{-\secpar}}$ and sends it to the prover. They both let $\tau = \pi \times I$. 
    \State Both parties run \Cref{alg:GKRinsideRR} on $\transcMat_1\circ \tau, \pi^{-1}, \bm j, \bm v_1$,
    obtaining $\bm j', \bm v'$.
     \State Let $\mathcal C: \FF\rightarrow \FF^{\lgK+\lncol}$ be the canonical low degree curve passing through $\bm j$ and  $\bm j'$ i.e. the unique degree $\le 2T-1$ curve such that $\mathcal C(\bm \Lambda)=(\bm j, \bm j')$. 
    \State $\cP$ computes the univariate polynomials of degree $\le (2T-1)(\lgK+\lncol)$,
    denoted by $g^{(0)}$ and $g^{(1)}$,
    that satisfy
    \[g^{(0)}(x)=\widehat{\transcMat_0}({\mathcal{C}}(x)) \hspace{1cm} g^{(1)}(x)=\widehat{\transcMat_1\circ \tau}(\mathcal C(x)),\]
    and sends them to $\cV$.
    \State $\cV$ accepts and outputs $\pi, \mathcal{C}, \bm j', \bm v', g^{(0)}, g^{(1)}$ if
    \ifFOCS
    \begin{align*}
    g^{(0)}(\Lambda_1, \Lambda_2, \ldots, \Lambda_T)&=\bm v_0\\
    \wedge \,g^{(1)}(\Lambda_{T+1},\ldots,\Lambda_{2T})&=\bm v'.
    \end{align*}
    \else
    \[g^{(0)}(\Lambda_1, \Lambda_2, \ldots, \Lambda_T)=\bm v_0\hspace{1cm} \wedge \hspace{1cm} g^{(1)}(\Lambda_{T+1},\ldots,\Lambda_{2T})=\bm v'.\]
    \fi 
\end{algorithmic}
\end{algorithm}

\begin{lemma}\label{lem: permuting Errors and Curve Fitting}
    Let $\lgK, \lncol, T, \secpar, r, d>0$ be integers with $r<d/4$, $4d\le 2^{\lgK/50}$. Let $\FF$ be a constructible field ensemble of characteristic $2$. We also have the inputs $\bm {j, v_0, v_1}$ where $\bm j \in (\FF^{\lgK+\lncol})^T$ and $\bm v_0, \bm v_1 \in \FF^T$, with prover having access to functions $\transcMat_0, \transcMat_1: \{0,1\}^{\lgK+\lncol}\rightarrow \FF$. Additionally, if $\Deltac(\pval(\bm j, 0))\ge 4d$ then we have that the protocol either breaks and rejects or outputs univariate polynomials $g_0, g_1: \FF\rightarrow \FF$ of degree at most $2T(\lgK+\lncol)$, a permutation $\pi: \{0,1\}^{\lgK}\rightarrow \{0,1\}^{\lgK}$, and a uniformly random $\bm j' \in (\FF^{\lgK+\lncol})^T$. Then \Cref{alg: gap amplification} has the following properties:
    \begin{itemize}
        \item \textbf{Unambiguous Distance Preservation:}
        Let $\transcMat=\transcMat_0\|\transcMat_1$. Now, with probability, $\ge 1- (2^{-\secpar}+\err_{\GKR}+\exp(-d/12r))$, the following properties hold:
        \begin{itemize}\item If $\Deltac(\transcMat, \pval(\bm j, \bm v_0)\|\pval(\bm j, \bm v_1))\le d$ and the prover answers according to the prover's prescribed strategy for the unique $\transcMat_0^\ast\|\transcMat_1^\ast=\transcMat^\ast \in (\pval(\bm j, \bm v_0)\|\pval(\bm j, \bm v_1))\cap \rowball(\transcMat)$ then:
        \begin{itemize}
            \item For any $c\in \FF$: $\transcMat^\ast_0+c\transcMat^\ast_1\circ \tau$ is the unique element in $\rowball(\transcMat_0+c\transcMat_1\circ \tau)\cap(\pval(\mathcal C (\FF), g_0+cg_1(\FF)))$.
            \ifFOCS
            \item $\begin{aligned}[t]
                &\Pr_c\begin{bmatrix}
                \begin{aligned}
                &\Deltac(\transcMat_0+c\transcMat_1\circ \tau, \transcMat^\ast_0+c\transcMat^\ast_1\circ \tau)\\
                &<(1 - 1/r) \cdot \Deltac(\transcMat, \transcMat^\ast)
                \end{aligned}
            \end{bmatrix}\\
            &<\frac{2r+1}{\FF}.\end{aligned}$
            \else
            \item $\Pr_c\begin{bmatrix}
                \Deltac(\transcMat_0+c\transcMat_1\circ \tau, \transcMat^\ast_0+c\transcMat^\ast_1\circ \tau)
                <\Deltac(\transcMat, \transcMat^\ast)(1-\frac{1}{r})
            \end{bmatrix}
            <\frac{2r+1}{\FF}.$
            \fi
        \end{itemize}
        \item If $\Deltac(\transcMat, \pval(\bm j, \bm v_0)\|\pval(\bm j, \bm v_1))> d$ or the prover is not honest with respect to the unique $\transcMat^\ast \in \Deltac(\transcMat, \pval(\bm j, \bm v_0)\|\pval(\bm j, \bm v_1))$ then either: 
        \begin{itemize}
            \item $\pval((\bm j, \bm j'), g_0(\Lambda))$ is empty; or
            \item $\pval((\bm j, \bm j'), g_1(\Lambda))$ is empty; or
            \ifFOCS
            \item $\begin{aligned}[t]
                &\Pr_c
                \begin{bmatrix}
                \begin{aligned}
                &\Deltac(\transcMat_0+c\transcMat_1\circ \tau,\\
                &\quad \quad \pval(\mathcal C(\FF), g_0+cg_1(\FF))\\
                &<d(1-\frac{1}{r})
                \end{aligned}
                \end{bmatrix}\\
                &< \frac{2r+1}{\FF}.
                \end{aligned}
                $
                \else
                \item 
                $\Pr_c
                \begin{bmatrix}
                \Deltac(\transcMat_0+c\transcMat_1\circ \tau, \pval(\mathcal C(\FF), g_0+cg_1(\FF))
                <d(1-\frac{1}{r})
                \end{bmatrix} < \frac{2r+1}{\FF}.
                $
                \fi
        \end{itemize}
        \end{itemize}
    \item \textbf{Round Complexity:} $O(D\log S)$
    \item \textbf{Message length per round:} $\tilde O(T(\lgK+\lncol)\log F+d\lgK\secpar)$.
    \item \textbf{Verifier Runtime:} $\tilde O(D\log S\cdot \poly\log(\FF) + d\lgK\secpar+T(\lgK+\lncol))$.
    \item \textbf{Prover Runtime:} $\poly(S, \log \FF, T, 2^{\lgK+\lncol})$.
    \end{itemize}
    Here $D$ and $S$ are the depth and size respectively of $\Phi_{\pi, \bm j, \bm v}$ as defined in \Cref{def: GKR inside RR circuits}.
\end{lemma}
\begin{proof}
We again don't focus on round, communication, and time complexities as they follow directly. 
The complexities of $\pi, g_0, g_1$ and \Cref{alg:GKRinsideRR} directly imply the results. 
Focusing on the unambiguous distance preservation property, the result is essentially a restatement of \Cref{lem: Batch RR} with $\epsilon=\gamma=\frac{1}{2r}$. 
In particular:
\begin{itemize}
    \item If $\Deltac(\transcMat, \pval(\bm j, \bm v_0)\|\pval(\bm j, \bm v_1))>d$, then the result follows by $\Cref{lem: Batch RR}$ with $\epsilon=\gamma = \frac{1}{2r}$.
    \item If $\Deltac(\transcMat, \pval(\bm j, \bm v_0)\|\pval(\bm j, \bm v_1))\le d$ and the prover is honest then:
    \begin{itemize}
        \item The first condition follows since $\Deltac(\pval(\mathcal C(\FF), g_0+cg_1(\FF)))\ge \Deltac(\pval(\bm j, 0))\ge 4d$,
        and we know that by the honesty definition that $\transcMat_0^\ast+c\transcMat_1^\ast\circ \tau \in \rowball(\transcMat_0+c\transcMat_1\circ \tau)$. Thus, it must also be unique.
        \item The second part follows from $\Cref{lem: Batch RR}$.
    \end{itemize}
    \item If the prover is dishonest first in:
    \begin{itemize}
        \item The GKR i.e. \Cref{alg:GKRinsideRR} then with probability $\ge 1-\err_{\GKR}$, we have that $\Deltac(\transcMat_1\circ \pi, \pval(\bm j', v'))\ge 4d$. Now, by \Cref{lem: Batch RVW}, the result follows.
        \item In returning $g$: If the checks pass, 
        but the prover lied in the value of $g$,
        then we know that either $\Deltac(\transcMat_0^\ast, \pval(\mathcal C(\FF), g_0(\FF)))\ge 4d$ or $\Deltac(\transcMat_1^\ast\circ \tau, \pval(\mathcal C(\FF), g_1(\FF)))\ge 4d$. Thus, since $\transcMat_0^\ast$ and $\transcMat_1^\ast\circ \tau$ are within $d$ distance of $\transcMat_0$ and $\transcMat_1\circ \tau$ by the GKR correctness and our assumption respectively. Thus, we can again simply conclude by the triangle inequality and \Cref{lem: Batch RVW}.
    \end{itemize}
\end{itemize}

\end{proof}

\subsubsection{Efficient \UIP for \pval emptiness}

Since we have used non-emptiness conditions on the \pval instances defined in \Cref{lem: Batch RR}, we need to be able to check the same. The following is a standard interactive proof for checking whether a given $\pval$ instance is empty. This is based on the observation that checking emptiness of a \pval instance is simply a rank computation, and thus we can use \Cref{lem:GKR} on such a circuit. We will make use of this lemma freely now.

\begin{lemma}[Lemma $5.10$ in \cite{TCC:RotRot20}] \label{lem: IP for pval emptiness}
    Let $T, \lgK, \secpar \in \NN$ and $\FF$, a finite field of characteristic $2$. There is a public coin unambiguous interactive proof for the language \[\mathcal L = \left\{(\bm j,\bm v) \in \left((\FF^m)^T, \FF^T\right) \mid \pval(\bm j, \bm v) \neq \varnothing\right\},\]
    with perfect completeness and the following parameters: \begin{itemize}
        \item Unambiguous Soundness: $2^{-\secpar}$
        \item Communication complexity:\\
        $\poly(m, \log (T\log \FF), \secpar)$.
        \item Round Complexity: $\poly(m, \log (T\log |\FF|),\secpar)$
        \item Verifier running time: $T\cdot \poly(m, \log |\FF|, \secpar)$
        \item Prover running time: $\poly(2^m, T\log |\FF|, 2^{\secpar})$.
    \end{itemize}  
\end{lemma}

\begin{proof}
    Let matrix $M\in \FF^{2^m\times T}$ where \[M_{a,b}=\hat{x}^a(j_b)\] where $x^a = \prod\limits_{a_r = 1} x_r$. Observe that each entry of the matrix can be computed using an arithmetic circuit over $\FF$ of depth $O(\log m)$ and size $m$. Now, since each multiplication and addition over $\FF$ as a boolean circuit is of depth $\log |\FF|$ and size $\log^2|\FF|$. Thus, each coordinate of $M$ can be computed in size $m\log |\FF|^2$ and depth $O(\log|\FF|\log m|)$.

    The rank of a matrix can be computed in depth $O(\log^2(T\cdot 2^m)$ as a circuit over $\FF$ and in size $\poly(2^mT)$. Now, each operation over $|\FF|$ multiplies another $\log^2|\FF|$ to size and another $\log|\FF|$ to the depth. Now, the verifier can simply check using GKR whether $\mathsf{rank}(M)=\mathsf{rank}(M\|\bm v)$. Observe that if these two ranks are equal then $\bm v\in \FF^{T}$ is in the $\FF$ linear span of the columns of $M$ i.e. there is a polynomial whose extension applied to $\bm j$ has the image $\bm v$ and the affine space has dimension $>0$ if the rank is not full. Similarly, we could have also just checked if $\dim \pval(\bm j,\bm v)>0$ instead of just checking if $|\pval(\bm j,\bm v)|>0$. 
\end{proof}

\subsubsection{Reducing Number of Evaluation Points:}

\begin{lemma}\label{lem: eval point reduction}
    
$\lgK, \lncol, T, \secpar, d>0$ are integers. $\FF$ is a constructible field of characteristic $2$. We also have $\bm j \in (\FF^{\lgK+\lncol})^T$ and uniformly random points $\bm j'\in (\FF^{\lgK+\lncol})^T$. $\bm \Lambda=(\Lambda_1,\ldots, \Lambda_{2T})$ are canonical distinct fixed points in $\FF$. Let $\transcMat: \{0,1\}^{\lgK+\lncol}\rightarrow \FF$ and let $g: \FF\rightarrow \FF$ be a univariate polynomial of degree at most $2T(\lgK+\lncol)$. Additionally, assume that $\Deltac(\pval(\bm j, 0))$ is non-empty.

Then, we have the following unambiguous distance preservation property:
\begin{itemize}
    \item If $\Deltac(\transcMat, \pval(\mathcal C(\FF), g(\FF)))\le d$ then:
    \ifFOCS
    \begin{align*}
        &\Pr_{\alpha=(\alpha_1, \ldots, \alpha_T)\in \FF^T}\begin{bmatrix}
        \begin{aligned}[t]
            &\rowball(\transcMat)\cap \pval(\mathcal C(\alpha),g(\alpha))\\
            &=\rowball(\transcMat)\cap \pval(\mathcal C(\FF, g(\FF))
        \end{aligned}
        \end{bmatrix}\\
        &\ge p
    \end{align*}
    \else
    \[\Pr_{\alpha=(\alpha_1, \ldots, \alpha_T)\in \FF^T}[\rowball(\transcMat)\cap \pval(\mathcal C(\alpha),g(\alpha))=\rowball(\transcMat)\cap \pval(\mathcal C(\FF, g(\FF)))] \ge p\]
    \fi
    \item If $\Deltac(\transcMat, \pval(\mathcal C(\FF), g(\FF)))>d$ then: 
    \ifFOCS
    \begin{align*}
        &\Pr_{\alpha=(\alpha_1, \ldots, \alpha_T)\in \FF^T}\begin{bmatrix}
        \begin{aligned}[t]
            \Deltac(\pval(\mathcal C(\alpha), g(\alpha)), \transcMat)
            >d
        \end{aligned}
        \end{bmatrix}\\
        &\ge p,
    \end{align*}
    \else
    \[\Pr_{\alpha=(\alpha_1, \ldots, \alpha_T)\in \FF^T}[\Deltac(\pval(\mathcal C(\alpha), g(\alpha)), \transcMat)>d]\ge p,\]
    \fi
    where $p = 1 - \left(\binom{M}{d}|\FF|^d\right)^{L}\left(\frac{2T(m+\log L)}{\FF}\right)^T$.

    Additionally, 
    \ifFOCS
    \begin{align*}
        &\Pr_{\alpha=(\alpha_1, \ldots, \alpha_T)\in \FF^T}\begin{bmatrix}
        \begin{aligned}[t]
            \Deltac(\pval(\mathcal C(\alpha),0))
            \ge 4d
        \end{aligned}
        \end{bmatrix}\\
        &\ge 1-\left(\binom{M}{4d}|\FF|^{4d}\right)^L\left(\frac{m+\log L}{\FF}\right)^T
    \end{align*}
    \else
    \[\Pr_{\alpha=(\alpha_1, \ldots, \alpha_T)\in \FF^T}[\Deltac(\pval(\mathcal C(\alpha),0))\ge 4d]\ge 1-\left(\binom{M}{4d}|\FF|^{4d}\right)^L\left(\frac{m+\log L}{\FF}\right)^T\]
    \fi
    where $\mathcal C: \FF\rightarrow \FF^{\lgK+\lncol}$ is the canonical $2T-1$ degree curve passing through $\bm j, \bm j'$.
\end{itemize}
\end{lemma}
\begin{proof}
    For the first part, observe that for if $\transcMat' \in \rowball(\transcMat)\setminus \pval(\mathcal C(\FF), g(\FF))$ then we have that $\widehat{\transcMat'} \circ \mathcal C\ne g$. Thus, $\Pr_{\alpha\in \FF^T}[\widehat{\transcMat'} \circ \mathcal C (\alpha) = g(\alpha)]\le \left(\frac{2T(\lgK +\lncol)}{\FF}\right)^T$ by the Schwartz-Zippel Lemma(\Cref{lem:SZ}). Thus, we finish by a union bound over all elements of $\rowball(\transcMat)$.

    For the second part, observe that since $\bm j'$ is uniformly random, thus so is $\mathcal C(\alpha)$. Thus, for any element $ \transcMat'$ such that $0<\Deltac(0, \transcMat')< 4d$, we have that $\Pr_{\alpha \in \FF^T}[\widehat{\transcMat'}(\mathcal C(\alpha))=0]\le \left(\frac{\lgK+\lncol}{|\FF|}\right)^T$ by \Cref{lem:SZ}. We can now finish by the union bound over all elements.

    \textbf{Remark.} In the proof of the second part, we used the fact that $j'$ was uniformly random. If we do not want to use this fact then we can observe that $\widehat{\transcMat'}\circ \mathcal C \ne 0$ for any $\transcMat'$ since $\Deltac(\pval(j, 0))\ge 8d$ and have a loss of a factor of $T$,
    but this will not matter much anyway.
\end{proof}

\subsubsection{Putting together the Reduction Protocol:}

\begin{algorithm}
    \setstretch{1.1}
    \caption{The Reduction Protocol}
    \label{alg: RRfolding}
    \textbf{Input Parameters:} $\lgK$, $\ncol$, $T$, $\secpar \in \NN$ and a field $\FF$. \\
    \textbf{Input:} $\bm j \in (\FF^{\lgK+\lncol})^T$, $\bm v \in \FF^T$. \\
    \textbf{Prover Auxiliary Input:} $\transcMat:\{0,1\}^{\lgK+\lncol}\rightarrow \FF$.\\
    \textbf{Output:} A permutation $\tau \in \{0,1\}^{\lgK - 1+\lncol}\to\bin^{\lgK-1+\lncol}$, a scalar $c \in \FF$, and a smaller instance: $\transcMat_{\text{new}} : \{0,1\}^{\lgK-1+\lncol}\to\FF$,
    $\bm j_{\text{new}} \in (\FF^{\lgK-1+\lncol})^T$ and $\bm v_{\text{new}} \in \FF^T$.
    \begin{algorithmic}[1]
        \State Run \Cref{alg: break into two} on input $\bm j, \bm v$ and $\transcMat$ to get output $\bm j'', \zeta^{(0)}, \zeta^{(1)}$. \State Break $\transcMat$ into $\transcMat_0\|\transcMat_1$.
        \State Run \Cref{alg: gap amplification} on input $\bm j'', \bm v_0=\zeta^{(0)}, \bm v_1 = \zeta^{(1)}$,
        obtaining $\tau, \mathcal{C}, \bm j', \bm v', g_0, g_1$.
        \State \textbf{$\pval$ Emptiness Check:} Both parties run the Protocol in \Cref{lem: IP for pval emptiness} with error parameter $\secpar$ to check the non-emptiness of $\pval((\bm j, \bm j'),g^{(0)}(\Lambda))$ and $\pval{((\bm j, \bm j'), g^{(1)}(\Lambda))}$.
        \State $\cV$ rejects if the check fails.
        \State $\cV$ samples $\alpha_1, \ldots, \alpha_T \leftarrow \FF$ and $c \leftarrow \FF$ and sends them to $\cP$.
        \State $\cV$ outputs $\tau, c, \transcMat_{\text{new}}=\transcMat_0+c(\transcMat_1\circ\tau)$, $\bm j_{\text{new}} = \left(\mathcal C (\alpha_i)\right)_{i = 1}^T$ and $\bm v_{\text{new}}= \left((g_0+cg_1)(\alpha_i)\right)_{i = 1}^T$.
    \end{algorithmic}
\end{algorithm}

\begin{lemma}\label{lem: RR Reduction Lemma}
    Let $\lgK, \lncol, T, \secpar, d, r, \in \mathbb N$ such that $r<d/4$ and $4d<2^{\lgK/50}$ and let $\FF$ be a constructible field. Let $\transcMat: \{0,1\}^{\lgK+\lncol}\rightarrow \FF$ and $\bm j \in (\FF^{\lgK+\lncol})^T, \bm v\in \FF^T$ such that $\Deltac(\pval(\bm j, 0))\ge 4d$.
    Define the following quantity.
    \ifFOCS
    \begin{align*}
        &p\coloneqq 1 - \frac{2r+1}{|\FF|}-2\left(\binom{M}{4d}|\FF|^{4d}\right)^L\left(\frac{m+\log L}{\FF}\right)^T\\
        &-\left(\binom{M}{d}|\FF|^{d}\right)^L\left(\frac{2T(m+\log L)}{\FF}\right)^T\\
        &-\left(2^{2-\secpar}+\err_{\GKR}+\exp(-d/12r)\right).
    \end{align*}
    \else
    \begin{align*}
        p\coloneqq 1 - \left(\frac{2r+1}{|\FF|}+2\left(\binom{M}{4d}|\FF|^{4d}\right)^L\left(\frac{m+\log L}{\FF}\right)^T\right)\\-\left(\left(\binom{M}{d}|\FF|^{d}\right)^L\left(\frac{2T(m+\log L)}{\FF}\right)^T+2^{2-\secpar}+\err_{\GKR}+\exp(-d/12r))\right).
    \end{align*}
    \fi
    If we run \Cref{alg: RRfolding} and get outputs $\transcMat_{\text{new}}, \bm j_{\text{new}}, \bm v_{\text{new}}$ then we have the following properties:
    \begin{itemize}
        \item \textbf{Completeness:} If $\transcMat \in \pval(\bm j, \bm v)$ and the prover was honest then $\transcMat_{\text{new}}\in \pval(\bm j_{\text{new}}, \bm v_{\text{new}})$.
        \item \textbf{Soundness:} If $\transcMat\not\in\pval(\bm j, \bm v)$ with respect to $\transcMat$ then \[\Pr[\Deltac(\transcMat_{\text{new}}, \pval(\bm j, \bm v))\ge d(1-\frac{1}{r})]\ge p.\]
        \item \textbf{Unambiguity:} If $\transcMat\in \pval(\bm j, \bm v)$ but the prover deviates from the honest strategy then \[\Pr[\Deltac(\transcMat_{\text{new}}, \pval(\bm j, \bm v))\ge d(1-\frac{1}{r})]\ge p.\]
        \item With probability at least $1-\left(\binom{M}{4d}|\FF|^{4d}\right)^L\left(\frac{m+\log L}{\FF}\right)^T$, we have that $\Deltac(\pval(\bm j_{\text{new}}, 0))\ge 4d$.
    \end{itemize}

    Additionally, the protocol has the following complexities:
    \begin{itemize}
        \item \textbf{Round Complexity:} $\tilde O(D\log S\cdot \poly(\lgK+\lncol,\log(T\log|\FF|),\secpar))$
        \item \textbf{Communication Complexity:} $(T+d)\cdot \poly(\lgK+\lncol, \log(T\log|\FF|), \secpar)$
        \item \textbf{Verifier Runtime:} $(T+d)D\log S\cdot \poly(\lgK+\lncol, \log(T\log|\FF|), \secpar)$.
        \item \textbf{Prover Runtime:}\ifFOCS \\ \fi $\poly(S, T, 2^{\lgK+\lncol}, \log|\FF|, \secpar)$.
        \item \textbf{Verifier's verification Circuit Size:}\ifFOCS \\ \fi $\tilde{O}(T)\poly(\log \FF, \lgK+\lncol, \secpar)$.
        \item \textbf{Verifier's verification circuit depth:}\ifFOCS \\ \fi $\polylog(T, \lgK+\lncol, \log|\FF|,\sec)$.
    \end{itemize}
        If we further assume that \begin{itemize}
            \item $|\FF|\ge \tilde\Omega(2^{\secpar}\cdot T^2(\lgK+\lncol+\log T)^2)$.
            \item $T>8d\cdot 2^{\lncol}\cdot \lgK$
            \item $d\ge 16\cdot r\cdot \secpar$
            \item $r = o(|\FF|\cdot 2^{-\sigma})$,
        \end{itemize} then 
        \ifFOCS
        \begin{align*}
        &\quad \frac{2r+1}{|\FF|}+3\left(\binom{M}{4d}|\FF|^{4d}\right)^L\left(\frac{m+\log L}{\FF}\right)^T\\
        &+\left(\binom{M}{d}|\FF|^{d}\right)^L\left(\frac{2T(m+\log L)}{\FF}\right)^T\\
        &+2^{2-\secpar}+\err_{\GKR}+\exp(-d/12r)\le 2^{3-\secpar}. 
        \end{align*}
        \else
        \begin{align*}
        \left(\frac{2r+1}{|\FF|}+3\left(\binom{M}{4d}|\FF|^{4d}\right)^L\left(\frac{m+\log L}{\FF}\right)^T\right)\\+\left(\left(\binom{M}{d}|\FF|^{d}\right)^L\left(\frac{2T(m+\log L)}{\FF}\right)^T+2^{2-\secpar}+\err_{\GKR}+\exp(-d/12r))\right)\le 2^{3-\secpar}.
        \end{align*}
        \fi
\end{lemma}
\begin{proof}
    This statement is a combination of \Cref{lem: break into two parts,lem: permuting Errors and Curve Fitting,lem: IP for pval emptiness,lem: eval point reduction}. The only thing left to show is the final verifier's verification circuit and the setting of parameters.

    The verifier's verification checks are:
    \begin{itemize}
        \item Checking Linear Combinations: Size: $T\polylog(\FF)$ and Depth $O(\log(T)+\log\log |\FF|)$.
        \item \Cref{alg:GKRinsideRR}: Depth $D$ and size $S$.
        \item Checking $O(T)$ polynomial evaluations:\\Size $\tilde O(T\cdot \log|\FF|(\lgK+\lncol))$ and depth $\poly(\log(T\log|\FF|(\lgK+\lncol)))$ since we can run multipoint polynomial evaluation.
        \item Emptiness checks: Size $T\cdot \poly(\lgK+\lncol, \secpar, \log |\FF|)$ and depth $\polylog(T, \lgK+\lncol,\log |\FF|, \secpar)$.
    \end{itemize}
    Thus, putting all these together, we get the desired result.
    
    We verify the parameters set:
    \begin{itemize}
        \item $\frac{2r+1}{|\FF|}\le 2^{-\secpar}$
        \item 
        \ifFOCS
        $\begin{aligned}[t]
            &\left(\binom{\K}{4d}\cdot |\FF|^d\right)\left(\frac{(\lgK+\lncol)}{|\FF|}\right)^T \\
            &\le \left(\frac{(2\lgK+2\lncol)^{m}\cdot |\FF|}{|\FF|^{m}}\right)^{4dL}\le \frac{1}{|\FF|}\le 2^{-\secpar}.
        \end{aligned}$
        \else
        $\begin{aligned}[t]
            \left(\binom{\K}{4d}\cdot |\FF|^d\right)\left(\frac{(\lgK+\lncol)}{|\FF|}\right)^T \le \left(\frac{(2\lgK+2\lncol)^{m}\cdot |\FF|}{|\FF|^{m}}\right)^{4dL}\le \frac{1}{|\FF|}\le 2^{-\secpar}.
        \end{aligned}$
        \fi
        \ifFOCS
        \item $\begin{aligned}[t]
            &\left(\binom{\K}{d}\cdot |\FF|^d\right)\left(\frac{T(\lgK+\lncol)}{|\FF|}\right)^T \\
            &\le \left(\frac{T(\lgK+2\lncol)^{8m}\cdot |\FF|}{|\FF|^{8m}}\right)^{dL}\le \frac{1}{|\FF|}\le 2^{-\secpar}.
        \end{aligned}$
        \else
        \item $\begin{aligned}[t]
            \left(\binom{\K}{d}\cdot |\FF|^d\right)\left(\frac{T(\lgK+\lncol)}{|\FF|}\right)^T \le \left(\frac{T(\lgK+2\lncol)^{8m}\cdot |\FF|}{|\FF|^{8m}}\right)^{dL}\le \frac{1}{|\FF|}\le 2^{-\secpar}
        \end{aligned}$
        \fi
        \ifFOCS
        \item $\begin{aligned}[t]
            \err_{\GKR} &\le \epsilon_{\GKR}(D, \log S, \abs{\FF}) \\
            &\quad + (\epsilon_{\GKR}(D, \log S, \abs{\FF}))^T \\
            &\quad \cdot\left(\binom{\K}{4d} \abs{\FF}^{4d}\right)^\ncol\\
            &\le O\left(\frac{D\log S}{|\FF|}\right)\\
            &\quad +\left(\frac{(O(D\log S)^m\cdot |\FF|}{|\FF|^m}\right)^{4d\ncol}\\
            &\le \frac{1}{|\FF|}\le 2^{-\secpar}
        \end{aligned}$
        \else
        \item $\begin{aligned}[t]
            \err_{\GKR} &\le \left(\epsilon_{\GKR}(D, \log S, \abs{\FF}) + (\epsilon_{\GKR}(D, \log S, \abs{\FF}))^T \cdot\left(\binom{\K}{4d} \abs{\FF}^{4d}\right)^\ncol\right)\\
            & \le O\left(\frac{D\log S}{|\FF|}\right)+\left(\frac{(O(D\log S)^m\cdot |\FF|}{|\FF|^m}\right)^{4d\ncol}\le \frac{1}{|\FF|}\le 2^{-\secpar}
        \end{aligned}$
        \fi
        \item $\exp(-d/12r)\le 2^{-\secpar}$.
    \end{itemize}
    Combining these we get the desired result.
    \end{proof}

\subsubsection{The Base Protocol}

\begin{algorithm}
    \setstretch{1.1}
    \caption{Base protocol in $\RR$.}
    \label{alg: RRconclude}
    \textbf{Input Parameters:} $\lgK$, $\lgK_{\text{new}}$ (where $\lgK\ge \lgK_{\text{new}}$), $\ncol$, $\secpar$, $d \in \NN$ and $\FF$ is a field.\\
    \textbf{Input:} $\bm j \in \left(\FF^{\lgK_{\text{new}}+\lncol}\right)^T$, $\bm v \in \FF^T$. \\
    \textbf{Other Parameters:} Let $\K \coloneqq 2^\lgK$ and $\K^{\text{new}} \coloneqq 2^{\lgK_{\text{new}}}$.\\
    \textbf{Prover Auxiliary Input:}
    A function $\transcMat_{\text{new}} \in \bin^{\lgK_{\text{new}}} \times \bin^{\lncol}\to\FF$. \\
    \textbf{Verifier Auxiliary Input:}
    A description $\lrag H$ of a $\FF$-linear function $H: (\FF^\ncol)^{\K}\to (\FF^{\ncol})^{\K_{\text{new}}}$,
    such that its every output coordinate is a linear combination of exactly $\frac{\K}{\K_{\text{new}}}$ elements of $\FF^{\ncol}$.
    (We assume that the set of rows that fold into any given row of $\K_{\text{new}}$ are easy to compute given the description $\left<H\right>$.)\\
    \textbf{Verifier Query Access:}
    Query access to $\transcMat: \{0,1\}^{\lgK} \times \bin^{\lncol}\rightarrow \FF$. \\
    \textbf{Output:} A succinct description $\left<Q\right>$ of $Q \subset \{0,1\}^\lgK$ of size $\left\lceil\frac{\secpar\cdot \K}{d}\right\rceil$ and a succinct predicate $\left<\Phi\right>:\FF^{|Q|\times \ncol}\rightarrow \{0,1\}$.
    \begin{algorithmic} [1]
    \State $\cP$ explicitly sends all of $\transcMat_{\text{new}}$ i.e. the entire string in $\FF^{2^{\lgK_{\text{new}}} \times \ncol}$ to $\cV$. 
    \State $\cV$ checks if all the extensions are correct i.e. $\transcMat_{\text{new}}\in \pval(\bm j, \bm v)$ (and rejects otherwise). 
    \State $\cV$ randomly selects $Q'$ to be a uniform subset of $\left\lceil\frac{\secpar\cdot \K}{d}\right\rceil$ of $[\K_{\text{new}}]$.
    \State $\cV$ outputs $Q$ as the set of rows that the rows of $H|_{Q'}$ depends on. $\Phi$ checks if $H(\transcMat|_{Q})=\transcMat_{\text{new}}|_{Q'}$. 
\end{algorithmic}
\end{algorithm}
\begin{claim}
    If $\Deltac(\pval(\bm j, \bm v), H(\transcMat))\ge d$  then with probability $\ge 1 - 2^{-\secpar}$, we have that $\Phi(\transcMat|Q)=0$.
    Additionally, \Cref{alg: RRconclude} has the following complexities:
    \begin{itemize}
        \item Round Complexity: $O(\lgK_{\text{new}})$
        \item Communication Complexity: $O(2^{\lgK_{\text{new}}}\cdot \ncol\cdot \log(\FF)+T\cdot (\lgK_{\text{new}}+\lncol)\cdot \log(\FF))$. 
        \item Verifier runtime: $O(\K_{\text{new}}\cdot \ncol\cdot \poly(\lncol+\lgK, \log \FF)\cdot T)$
        \item Prover runtime: $\K_{\text{new}}\ncol\cdot \log|\FF|\poly(T, \lgK_{\text{new}}+\lncol)$
    \end{itemize}
\end{claim}

\begin{proof}
    When the distance is at least $d$,
    the verification passes exactly when none of the rows in the sampled $Q'$ are inconsistent with $H$.
    This probability is at most $(1-\frac{d}{\K_{\text{new}}})^{|Q'|} \le 2^{-\secpar}$.
\end{proof}
\subsubsection{The complete protocol:}

\begin{theorem}[IPP for $\pval$ with column distance]
    \label{thm: copy of DcRR}
    Let $\lgK,\lncol, T, \secpar, d \in \NN$. Let $\FF$ be a constructible field ensembles such that $\FF$ has characteristic $2$.
    If
    \begin{itemize}
        \item $|\FF|\ge \Omega(2^{\secpar}\cdot T^2(\lgK+\lncol)^2)$,
        \item $T>8d\cdot 2^{\lncol}\cdot \lgK$,
        \item $d\ge 16\cdot \lgK\cdot \secpar$
    \end{itemize}

    Then for any $\bm j \in (\FF^{\lgK+\lncol})^T, v\in \FF^T$, the matrix $\pval(\bm j,\bm v)$ has a public coin unambiguous $\IPP$ (\Cref{alg: RRbetter}) as long as $\Deltac(\pval(\bm j, 0))>4d$:
    \begin{itemize}
        \item Soundness error: $\lgK\cdot 2^{4-\secpar}$
        \item Prover to Verifier communication complexity: $T\log |\FF|\cdot \poly(\lgK,\lncol)+\ncol\cdot \log|\FF|\cdot \poly(d)$.
        \item Verifier to Prover communication complexity: $T\log |\FF|\cdot \poly(\lgK,\lncol)$ 
        \item Round Complexity: $\poly(\lgK,\lncol, \log(T\log|\FF|))$
        \item Verifier runtime:
        \ifFOCS\\
        \fi
        $T\log|\FF|\cdot \poly(\lgK,\lncol, \log(T\log|\FF|), \log |\FF|) + \tilde{O}(\ncol\cdot \log|\FF|\cdot \poly(d))$ 
        \item Prover runtime: $\poly(2^{\lgK+\lncol},T\log |\FF|)$
        \item Query complexity : $\abs{\RRS} = \frac{4\secpar\cdot 2^{\lgK+\lncol}}{d}$.
    \end{itemize}
    The final verifier verdict is given by a succinctly described predicate $\Phi : \FF^{\frac{4\secpar}{d}\cdot2^{\lgK+\lncol}}\mapsto \{0,1\}$ along with a succinctly described set $Q\subseteq [2^\lgK]$ such that $|Q|=\frac{4\secpar\cdot 2^{\lgK}}{d}$, and the verdict is $\Phi(f|Q\times \{0,1\}^\lncol)$ with $|\left<Q\right>|=\tilde O(\lgK^2d\sigma\cdot +\lgK\log|\FF|+\sigma\cdot \poly(d))$ and $|\left<\Phi\right>|=|\left<Q\right>|+\sigma\cdot \poly(d)\cdot \ncol\cdot \log|\FF|$.
\end{theorem}

\begin{proof}
    Observe that the total error is at most the sum of the errors in each round which we know is at most $2^{3-\secpar}$ for the first $\lgK$ rounds if $r=\lgK$ in \Cref{lem: RR Reduction Lemma}. This would result in that if $\Deltac(\transcMat, \pval(\bm j, v))\ge d$ initially, then after $i$ rounds, we have $\Deltac(\transcMat_{\text{prot}}, \pval(\bm j, \bm v))\ge d(1-\frac{1}{\lgK})^i$. Thus, with probability at least $\lgK \cdot 2^{3-\secpar}$, we have that when \Cref{alg: RRconclude} is called then indeed $\Deltac(\transcMat_{\text{prot}},\pval(\bm j, \bm v))>d/4$ as required. Thus, the total error is at most $\lgK\cdot 2^{4-\secpar}$ as required.

    The total round, time, communication complexities are just the sum of $\lgK$ times the complexities of \Cref{alg: RRfolding} and the complexity of \Cref{alg: RRconclude},
    as desired.
    The final query complexity is as described in \Cref{alg: RRconclude}.

    We discuss the succinct descriptions in the next subsection.
\end{proof}

\begin{algorithm}
  \setstretch{1.1}
  \caption{Efficient $\IPP$ for $\pval$ with column distance}
  \label{alg: RRbetter}
  \textbf{Input Parameters:} $m$, $\ncol \in \NN$, $\secpar$, $d \ge 16\secpar\cdot \lgK$, $\FF$ is a field with $\abs{\FF}\ge \FboundT$.\\
  \textbf{Input:} $\bm j\in (\FF^{\lgK+\lncol})^T, \bm v\in \FF^T$.\\
  \textbf{Other Parameters:} Let $r \coloneqq \lgK$.\\
  \textbf{Prover Auxiliary Input:} $\transcMat:\{0,1\}^{\lgK+\lncol}\rightarrow\FF$.\\
  \textbf{Verifier Query Access:} Query access to $\transcMat:\{0,1\}^{\lgK+\lncol}\rightarrow\FF$.\\
  \textbf{Output:} Succinct descriptions of $Q \subset [\K]$ and $\Phi : \FF^{|Q| \times \ncol}\rightarrow \{0,1\}$.
  \begin{algorithmic}[1]
    \State $\cV$ maintains a linear function $H: \left(\FF^{\lncol}\right)^{\K}\rightarrow \left(\FF^{\lncol}\right)^{\K_{\text{prot}}}$.
    \State Set the variable parameters: $\lgK_{\text{prot}} \coloneqq \lgK$, $H_{\text{prot}} \coloneqq Id$ and $\transcMat_{\text{prot}} \coloneqq \transcMat$.
    \While{\textsf{True}}
    \If{$2^{\lgK_{\text{prot}/50}}\le 4d$:}
    \State Both parties run \Cref{alg: RRconclude} on input parameters $\lgK$, $\lgK_{\text{prot}}$, $\ncol$, $\secpar$, $d/4$, $\FF$, and input $\bm j$, $\bm v$, prover auxiliary input $\transcMat_{\text{prot}}$, verifier auxiliary input $\left<H_{\text{prot}}\right>$ and verifier query access to $\transcMat$, obtaining $\left<Q\right>$ and $\left<\Phi\right>$.
    \State $\cV$ outputs $\left<Q\right>$ and $\left<\Phi\right>$.
    \Else \State Both parties run \Cref{alg: RRfolding} on input parameters $\lgK_{\text{prot}}, \ncol$, $\secpar$ and $T$, input $\bm j, \bm v$ and prover auxiliary input $\transcMat_{\text{prot}} = \transcMat_{0,\text{prot}}\|\transcMat_{1,\text{prot}}$,
    obtaining $\tau$, $c$, $\bm j_{\text{new}}, \bm v_{\text{new}}$ and $\transcMat_{\text{new}}=\transcMat_0+c\transcMat_1\circ \tau$.
    \State 
    Update $\bm j = \bm j_{\text{new}}$, $\bm v = \bm v_{\text{new}}$, $\lgK_{\text{prot}}=\lgK_{\text{prot}}-1$, and $\transcMat_{\text{prot}} = \transcMat_{\text{new}}$.
    \State Update $H_{\text{prot}}=G\circ H_{\text{prot}}$ where $G(x_0\|x_1)=x_0+cx_1\circ \tau$. $\left<H_{\text{prot}}\right>=\left<H_{\text{prot}}\right>\|(\tau, c)$. 
    \EndIf
    \EndWhile
\end{algorithmic}
\end{algorithm}

\subsection{Succinct Descriptions}
We verify that the output parameters $(Q, \C)$ have succinct descriptions. The size of $\left<Q\right>$ is given by the randomness of the choice in $\Cref{alg: RRconclude}$ which is $\poly(d)\cdot \sigma \cdot \log|\FF|$ and the $(c, \pi)$ pairs from the $\lgK$ rounds. Thus, giving us a total complexity of $\tilde O(\lgK^2d\sigma\cdot +\lgK\log|\FF|+\sigma\cdot \poly(d))$ since the seed length of a single $\pi$ is $\tilde O(\lgK \cdot d \cdot \sigma)$.

Additionally, $|\left<\C\right>|$ is just given by $|\left<Q\right>|+|Q'|\cdot \ncol\cdot \log|\FF|$ since $\left<Q\right>$ contains the entire details of $\left<H\right>$. Thus, we only need $\transcMat_{\text{prot}}|Q'$ which is $\lncol\cdot |Q'|\cdot \log|\FF |=\lncol\log|\FF|\sigma\cdot \poly(d)$ bits. 
\paragraph{Circuit $G$ for generating $Q$}
Let $\lgK'=\lgK_{\text{{prot}}}$ when $\Cref{alg: RRconclude}$ is called.

The set $Q$ has succinct description 
\ifFOCS
\begin{align*}
\lrag Q = \left(\lrag Q', \begin{Bmatrix}
    \begin{aligned}[t]
    &\pi: \pi \text{ is used in the folding }\\
    &\text{ steps of the protocol}
    \end{aligned}
\end{Bmatrix}\right),
\end{align*}
\else
\[\lrag Q = (\lrag Q', \set{\pi: \pi \text{ is used in the folding steps of the protocol}}),\] 
\fi
where $Q' \subset \bin^{m'}$ are the $\frac{4\sigma\cdot 2^{\lgK'}}{d}$ random rows selected in the end of the iteration and $\lrag R \in \bin^{\abs{Q'} \cdot m'}$ are the random coins used to generate $Q'$.

Given $(i, \lrag Q)$, 
$G$ needs to generate the $i$-th element of $Q$.
The idea is to carry out the folding steps in reverse,
which can be done as follows.
First, it generates the full set $Q'$ using $\lrag R$,
and let $Q_0 \coloneqq Q'$,
which requires $\tO(\abs{Q'})$ gates and $\tO(1)$ depth.
If $i \le \abs{Q}/2$,
the $i$-th element of $Q$ would lie in the first half before folding,
so $G$ prepends a $0$ to every element in $Q_0$.
Otherwise,
$G$ prepends a $1$ to every element in $Q_0$.
Since the indices in the second half were permuted by $\pi$,
$G$ also applies $\pi^{-1}$ to every element in $R_0$ to reconstruct their locations in the original index space.
The circuit $G$ repeats this process $(\lgK-\lgK')$ times,
until it obtains $\abs{Q'}$ indices that lies in $\bin^{\lgK}$,
and it outputs the $(i \bmod \abs{Q'})$-th index.

\paragraph{Circuit $C(\transcMat|Q \times \bin^\lncol, \lrag \C)$ for checking $\C(\transcMat|Q \times \bin^\lncol)$}
Let $Q' \subset \bin^{m'}$ be the $\left\lceil\frac{4\sigma\cdot 2^{\lgK'}}{d}\right\rceil$ random rows selected in the end of the iteration, 
and $\lrag Q' \in \bin^{\abs{Q'} \cdot m'}$ be the random coin used to generate it.
The predicate $\C$ can be described by the pair
\ifFOCS
\begin{align*}
\left(\lrag Q', 
\begin{Bmatrix}
    \begin{aligned}[t]
    &(c, \pi): (c, \pi) \text{ is used in the folding }\\
    &\text{ steps of the protocol}, \bm \psi 
    \end{aligned}
    \end{Bmatrix}\right),
\end{align*}
\else
\[(\lrag Q', \set{(c, \pi): (c, \pi) \text{ is used in the folding steps of the protocol}}, \bm \psi),\]
\fi
where $\psi$ is the set of values claimed by the prover on $Q'$.

The circuit first enumerates the entire $Q$ using $\lrag Q'$ and all the $\pi$'s (as in the last paragraph),
and then it applies the $\FF$-linear function $\phi: \FF^{\abs{Q}\times 2^\lncol} \to \FF^{2^\lncol}$ determined by all the $(c,\pi)$ used in the folding steps of the protocol, on the $2^\lncol$ columns of the truth table of $f$ in parallel.
This can be done by iteratively applying the corresponding $\pi$ and figuring out which of $c$ to multiply in every folding step for each position in $Q$,
and requires $\tO(\abs{Q} \cdot 2^{\lncol} + \secpar +2^\lncol \Flog)$ gates and $\tO(1)$ depth.

\subsection{Construction of the Verifier's Verification Circuit}

We now wish to understand the size and depth of the verifier's verification circuit $\cV_{\RRrow}$. $\cV_{\RRrow}$ on reading the transcript of the interaction and the verifier decides whether to reject or to check $\Phi$.  
\begin{claim}
    If the transcript of the interaction is of length $\mathrm L$, then $\cV_{\RRrow}$, the verification circuit, has the following properties:
    \begin{itemize}
        \item Size of $\cV_{\RRrow}$ is $\tilde{O}(\mathrm L)$ where we ignore $\poly(\lgK +\lncol)$ factors.
        \item Depth of $\cV_{\RRrow}$ is $\tilde{O}(\log \mathrm{L})$ where we ignore $\poly(\lgK+\lncol)$ factors.
    \end{itemize}
\end{claim}

\begin{proof}
    Observe that the verifier conducts the following checks:
    \begin{enumerate}[label=(\arabic*)]
        \item $\tilde O(T\lgK)$ linearity checks to see if the verifier broke the instances into two claims correctly.
        \item $\tilde{O}(\lgK)$ $\pval$ emptiness checks which are run via a GKR.
        \item $\tilde{O}(\lgK)$ runs of \Cref{alg:GKRinsideRR} to apply the pairwise independent map and convert $f_1$ into $f_1\circ \tau$.
        \item $\tilde O(\lgK)$ recalculations of input points $\bm j$ for recursing via evaluating low degree curves. 
        \item Checking if $\tilde O (T)$ low degree extensions are correct.
    \end{enumerate}

    Now, step $(1)$ is clearly linear in it's input. Step $(2)$ and $(3)$ can be done in nearly linear input size and logarithmic depth due to \Cref{lem:GKR} and \Cref{lem:RVW}. We may also need to check $\tilde{O}(\lgK)$ low degree extensions which can be done in nearly linear of their input length as well. Thus, we just need to show that interpolation and evaluation of univariate polynomials can be done in nearly linear number of field operations. We know that this is possible due to our assumptions about our fields and multipoint interpolation and evaluation results (Corollary 10.8 and 10.12 in \cite{vzGG13}). 

    Thus, all the checks can be done in nearly linear size of the input transcript. The depth of the overall circuit is also clearly $\tilde O(\log \mathrm L)$.
\end{proof}
 
\section*{Acknowledgement}
The authors thank Ron Rothblum, Guy Rothblum, and the anonymous FOCS reviewers for their helpful comments and suggestions.

M.M.H. is partially supported by National Institute of Health (NIH) R01HG010959 (to B.B.).
M.M.H., R.G. and Y.T.K. are supported by the Defense Advanced Research Projects Agency (DARPA) under Contract No. HR0011-25-C-0300 (to Y.T.K.). 
Any opinions, findings and conclusions or recommendations expressed in this material are those of the author(s) and do not necessarily reflect the views of the Defense Advanced Research Projects Agency (DARPA).
\pagestyle{plain}

\bibliographystyle{alpha}
\bibliography{abbrev3,crypto,ref}

\end{document}